\documentclass[sigconf]{acmart}
\usepackage[utf8]{inputenc} 
\usepackage[T1]{fontenc}    
\usepackage{hyperref}       
\usepackage{url}            
\usepackage{booktabs}       
\usepackage{amsfonts}       
\usepackage{nicefrac}       
\usepackage{microtype}      
\usepackage{xcolor}         
\usepackage{bm}
\usepackage{mdframed}
\usepackage{algorithm, algorithmic}
\usepackage{appendix}
\usepackage{amsfonts,amsmath}
\usepackage{subfigure}
\usepackage{graphicx}
\usepackage{amsthm}
\usepackage{enumitem}
\usepackage{thm-restate}
\usepackage{multicol}
\usepackage{xcolor}
\usepackage{soul}

\DeclareMathOperator{\Tr}{Tr}

\newtheorem{remark}{Remark}[section]

\usepackage{calc}

\usepackage{accents}

\newcommand{\vardbtilde}[1]{\tilde{\raisebox{0pt}[0.85\height]{$\tilde{#1}$}}}

\copyrightyear{2023}
\acmYear{2023}
\setcopyright{rightsretained}
\acmConference[CCS '23]{Proceedings of the 2023 ACM SIGSAC Conference on Computer and Communications Security}{November 26--30, 2023}{Copenhagen, Denmark}
\acmDOI{10.1145/3576915.3623142}
\acmISBN{979-8-4007-0050-7/23/11}

\begin{document}
\title{Geometry of Sensitivity: Twice Sampling and Hybrid Clipping in Differential Privacy with Optimal Gaussian Noise and Application to Deep Learning} 

\begin{abstract}
{We study the fundamental problem of the construction of optimal randomization in Differential Privacy (DP).} 
Depending on the clipping strategy or additional properties of the processing function, the corresponding sensitivity set theoretically determines the necessary randomization to produce the required security parameters. 
{Towards the optimal utility-privacy tradeoff, finding the minimal perturbation for properly-selected sensitivity sets stands as a central problem in DP research. } 
In practice, $l_2/l_1$-norm clippings with Gaussian/Laplace noise mechanisms are among the most common setups. 
However, they also suffer from the {\em curse of dimensionality}. 
For more generic clipping strategies, the understanding of the optimal noise for a high-dimensional sensitivity set remains limited. 
This raises challenges in mitigating the worst-case dimension dependence in privacy-preserving randomization, especially for deep learning applications.  

In this paper, we revisit the geometry of high-dimensional sensitivity sets and present a series of results to characterize the {\em non-asymptotically} optimal Gaussian noise for Rényi DP (RDP). 
Our results are both negative and positive: on one hand, we show the curse of dimensionality is tight for a broad class of sensitivity sets satisfying certain symmetry properties; but if, fortunately, the representation of the sensitivity set is asymmetric on some group of orthogonal bases, we show the optimal noise bounds need {\em not} be explicitly dependent on either dimension or rank. 
We also revisit sampling in the high-dimensional scenario, which is the key for both privacy amplification and computation efficiency in large-scale data processing. We propose a novel method, termed {\em twice sampling}, which implements both sample-wise and coordinate-wise sampling, to enable Gaussian noises to fit the sensitivity geometry more closely. 
With {\em closed-form} RDP analysis, we prove twice sampling produces asymptotic improvement of the privacy amplification given an additional $l_{\infty}$-norm restriction, especially for {\em small} sampling rate. We also provide concrete applications of our results on practical tasks. Through tighter privacy analysis combined with twice sampling, we efficiently train ResNet22 in low sampling rate on CIFAR10, and achieve 69.7\% and 81.6\% test accuracy with $(\epsilon=2,\delta=10^{-5})$ and $(\epsilon=8,\delta=10^{-5})$ DP guarantee, respectively.
\end{abstract}



\author{Hanshen Xiao}
\affiliation{%
   \institution{Massachusetts Institute of Technology \\   Cambridge, MA, USA}}
\email{hsxiao@mit.edu}

\author{Jun Wan}
\affiliation{%
   \institution{Massachusetts Institute of Technology\\
    Cambridge, MA, USA }}
\email{junwan@mit.edu}

\author{Srinivas Devadas}
\affiliation{%
   \institution{Massachusetts Institute of Technology\\
    Cambridge, MA, USA}}
\email{devadas@mit.edu}

\maketitle

\flushbottom

\section{Introduction}
Emerged as the de-facto privacy risk measurement, Differential Privacy (DP) provides a semantic and input-independent worst-case guarantee regarding the hardness to infer the participation of an individual input from any release. 
At a high level, there are two steps to differentially privatize a data processing protocol $\mathcal{F}: \mathcal{X}^* \to \mathbb{R}^d$.

First, to capture the worst-case influence/effect from an individual to the output of $\mathcal{F}$, we need to determine the sensitivity set $\mathsf{S}$ of $\mathcal{F}$, where $\mathsf{S}= \{\pm\big(\mathcal{F}(X) - \mathcal{F}(X\cup x) \big): X \in \mathcal{X}^*, x\in \mathcal{X}\}.$ 
Sensitivity set $\mathsf{S}$ includes all the possible changes to the output when we arbitrarily remove an individual datapoint $x$ 
With $\mathsf{S}$, the second step is to randomize $\mathcal{F}$ such that the distribution divergence between its randomized version $\mathcal{RF}$ on two arbitrary adjacent data sets is close enough. Here, adjacent data sets denote a pair of sets that only differ in a single datapoint. 
Mathematically, this can be described as  $\sup_{X \in \mathcal{X}^*, x\in \mathcal{X}}\rho\big(\mathbb{P}_{\mathcal{RF}(X)}\|\mathbb{P}_{\mathcal{RF}(X\cup x)}\big) \leq \epsilon(\delta)$, where $\rho$ is some divergence metric for two distributions $\mathbb{P}_{\mathcal{RF}(X)}$ and $\mathbb{P}_{\mathcal{RF}(X\cup x)}$, and $\epsilon(\delta)$ is the security parameter. 

With different motivations, there are many commonly-used metrics $\rho$. 
For example, when $\rho$ is selected to be infinity divergence $\mathcal{D}_{\infty}(\mathbb{P}_a\|\mathbb{P}_b) = \sup_{z}\{\max\{\pm\log \frac{\mathbb{P}_a(z)}{\mathbb{P}_b(z)}\}\}$, i.e., the largest log ratio between the probability density functions, the above becomes the well-known $\epsilon$-DP definition   \cite{dwork2006a,dwork2006b}. 
Small $\epsilon$-DP guarantee suggests that, for arbitrary $X$ and $x$, either Type I or Type II error in a hypothesis testing to infer whether the true input is $X$ or $X \cup x$ is large
\cite{dong2022gaussian,geng-approxDP}. 
Similarly, if one selects $\rho$ to be (symmetrized) $\alpha$-Rényi divergence, then it becomes the $(\alpha,\epsilon)$-Rényi DP (RDP) \cite{renyi}, which excels in DP composition. 
When the randomized $\mathcal{RF}$ is unbiased with respect to $\mathcal{F}$, such as zero-mean noise perturbation, producing required security parameters reduces to determining the worst-case divergence between two distributions. The difference of their means is captured by $\mathsf{S}$.
In particular, we call the elements in $\mathsf{S}$ that cause the largest output divergence as the dominating (worst-case) sensitivity.  

As the central problem in DP research, how to find or closely approximate the sensitivity set and accordingly select the optimal randomization to produce a tight privacy-utility tradeoff remains largely open, especially for high-dimensional and complicated data processing. 
The two seemingly simple privatization steps are actually difficult in practice, as summarized below.

\noindent \textbf{{Intractable Tight Sensitivity}}: First, even for simple mean estimation of a dataset, one cannot claim a bounded sensitivity set without assumptions on an individual datapoint. 
Moreover, even if the processing $\mathcal{F}$ is bounded or a bounded sensitivity set $\mathsf{S}$ is given, the dominating sensitivity is in general NP hard to determine  
\cite{sensitivity-NPhard}. 
Therefore, it is in general {\em impossible} to perfectly achieve the optimal privacy-utility tradeoff for arbitrary data processing. 
To this end, an alternative operation to ensure tractable sensitivity is clipping. 
In most practical applications, instead of characterizing the actual sensitivity set of the target processing function $\mathcal{F}$, one could propose some approximated and analyzable sensitivity set $\mathsf{S}$, and then artificially project the output of processing $\mathcal{F}$ to $\mathsf{S}$. 
For example, $l_2/l_1$-norm clipping is equivalent to a projection into an $l_2/l_1$ norm ball. 
Though such straightforward clipping based on $l_2/l_1$ norm is easy to implement and analyze, its approximation performance and the clipping bias caused are rarely studied for practical high-dimensional tasks.  

\noindent \textbf{{Inefficient Randomization}:} 
{Second, even if the sensitivity set is given or can be closely approximated, existing randomization tools could be inefficient. 
Currently, noise (isotropic Gaussian/Laplace mechanisms) and (sub)sampling (Poisson sampling) are the two most-commonly used approaches to produce DP guarantees. 
However, as only sufficient conditions, both methods have efficiency problems, which are mainly twofold. 
One one hand, they may not perfectly capture the geometry and the introduced perturbation could be sub-optimal.
On the other hand, to produce a better utility-privacy tradeoff, they may also incur high implementation overhead, for example, requiring a large batch of subsampled data and consequently a high memory requirement for the privatized algorithm, as explained below.} 

As for {independent perturbation}, ideally, the injected noise is expected to reflect the geometry of the sensitivity set such that, at a cost of minimal variance, possibly all elements in $\mathsf{S}$ of largest norm could be the dominating sensitivity. 
Unfortunately, to our knowledge, {\em non-asymptotically} optimal noise in terms of minimal variance is only known for $l_1$-norm sensitivity for $\epsilon$-DP  \cite{laplaceoptimal}\footnote{\cite{laplaceoptimal} shows that the optimal noise distribution is a mixture of the uniform distribution and the geometric distribution for $\epsilon$-DP.}. 
For asymptotic results, \cite{geometry-hardt} and  \cite{geometry-approx, de2012lower} provide asymptotic lower bounds of necessary perturbation for $\epsilon$-DP and $(\epsilon,\delta)$-DP, respectively.
They consider applications in private linear query, where the sensitivity set is in a form $\mathsf{S}=A\mathcal{B}_1$ or $\mathsf{S}=A\mathcal{B}_2$ for some matrix $A$ and $l_1/l_2$-ball $\mathcal{B}_1/\mathcal{B}_2$.
Thus, in particular for $l_2$ and $l_1$ norm sensitivity, Gaussian and Laplace noises are known to produce {\em asymptotically} tight utility-privacy tradeoff for $\epsilon$-DP and $(\epsilon, \delta)$-DP, respectively. 
For both cases, a scale of $\tilde{\Theta}({\sqrt{d}}/{\epsilon})$ perturbation is required \cite{bassily2014private,laplaceoptimal}\footnote{Throughout this paper, we use scale to denote the expected $l_2$ norm of a random vector, i.e., $\mathbb{E}_{v}[\|v\|_2].$}, known as the {\em curse of dimensionality}.
However, for more generic sensitivity sets or when one has more fine-grained approximation besides simple $l_2/l_1$-norm restrictions, Gaussian/Laplace may fail to capture the privacy gain from those additional constraints, as we will discuss in detail in Section \ref{sec: optimal noise}. 
Lack of powerful tools to handle more general sensitivity sets is a primary reason that $l_2/l_1$-norm clipping with Gaussian/Laplace noise becomes the almost default option in practice. 
Consequently, the curse of dimensionality is unavoidable in current privatization frameworks unless additional assumptions can be made or better randomization is known.

{Sampling is another popular randomization to enhance privacy guarantees, but its amplification power sharply drops when a smaller sampling rate is applied in current analysis.} With sampling, the chance that an individual gets selected in the processing decreases and thus the security parameter will be scaled by a factor roughly proportional to the sampling rate \cite{zhu2019poission,balle2018privacy,mironov2019r}, known as {\em privacy amplification}. 
Sampling also plays an important role in large-scale data processing for implementation efficiency. 
One classic application is Stochastic Gradient Descent (SGD), the workhorse of optimization and machine learning. 
Under mild assumptions, SGD could bring asymptotic improvement on gradient computation compared to full batch GD {when achieving the same convergence accuracy} \cite{boyd2004convex}. 
However, sampling itself cannot provide meaningful DP guarantees, and thus it has to be accompanied by noise mechanisms. 
Though processing subsampled data requires less noise, from a signal-to-noise-ratio (SNR) perspective, simultaneously, less data is applied to average out the DP noise. Under existing analysis frameworks, it is shown that the SNR could be worse with smaller sampling rate in many practical applications such as DP-SGD \cite{CCS2016, mcmahanlearning}. 
State-of-the-art works generally select very large sampling rate (>0.3) \cite{deepmind} with massive overhead to produce a better utility-privacy tradeoff. Such a conflict between privacy and efficiency remains a challenge.

Therefore, in order to efficiently achieve the optimal or near-optimal utility-privacy tradeoff, three fundamental questions concerning (high-dimensional) sensitivity geometry need to be addressed.

a). \textbf{{What is a proper clipping method to efficiently approximate the sensitivity set of practical high-dimensional data processing?}} 
{While a sufficiently large $l_2/l_1$-norm ball can encompass arbitrary bounded sensitivity sets, such an isotropic clipping approach can be loose and costly in practice, particularly when the power of processed output is not uniformly distributed across the entire space. 
Ideally, we aim to develop clipping methods that rely on a few simple yet stable statistics/features, allowing them to be broadly applicable in practical data processing, while accurately capturing the dominant part of the actual sensitivity set.}

b). \textbf{{How can we overcome the curse of dimensionality?}} {As larger models are employed, especially in the development of deep learning {\cite{openai2023gpt4}}, noise scaling with output dimensionality $d$ poses a significant challenge for DP applications. 
Given the impossibility results presented in {\cite{bassily2014private}} and {\cite{laplaceoptimal}}, where the curse of dimensionality is unavoidable for $l_2/l_1$-norm clipping, a critical question is to identify the form of sensitivity sets for which the noise bound can be independent of $d$. Moreover, we need to explore the feasibility of the corresponding clipping methods that facilitate dimensionality-free noise, while being practically applicable.}

c). \textbf{{How can we design randomization methods that align with sensitivity geometry and have minimal implementation overhead?}} {Ideally, we seek optimal randomization techniques that perturb the processing minimally, while fitting its underlying geometry. Furthermore, to achieve a meaningful privacy-utility tradeoff, we require efficient implementation, which includes simple clipping and noise generation procedures, while allowing for low sampling rates to be used.}

\vspace{-0.1 in }
\subsection*{{Contributions and Paper Organization}}
In this paper, we set out to answer the above three questions and tackle the challenges of private high-dimensional processing both in theory and practice. Our contributions and the remaining contents are summarized as follows. 

\begin{figure}[t]
  \centering
  \includegraphics[width=0.99 \linewidth]{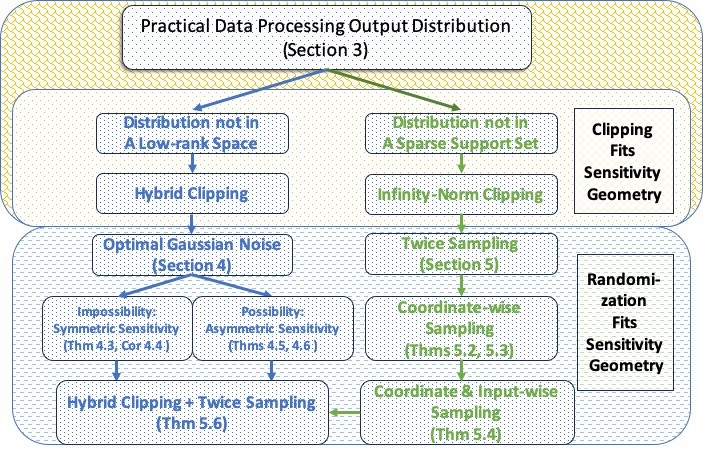}
\vspace{-0.1 in}
\caption{{Workflow of Paper Organization}}
\vspace{-0.35 in}
\label{Fig. workflow}.
\end{figure}  

a). {In Section {\ref{sec: warmup}}, we commence with a preliminary empirical study on the statistical features of practical data processing. We illustrate this using examples of biological gene data and gradients of neural networks (ResNet22) on CIFAR10 image samples. In fact, the distributions of practical high-dimensional processing are more intricate than expected, defying simple categorization or description based on sparsity or low rankness. 
Instead, we observe the existence of a principal subspace wherein the distributions are concentrated, while the power in the residue subspace is also non-negligible. 
To capture this property, where the power of processed output distribution is not uniform, we propose \textbf{smooth hybrid clipping}, involving multiple subspace embeddings alongside an additional $l_{\infty}$-norm restriction. 
Specifically, we experimentally demonstrate that proper $l_{\infty}$-norm clipping causes negligible changes in many complex high-dimensional processing pipelines, and the hybrid clipping can closely approximate the distribution geometry while introducing only small clipping bias.}

b). {In Section {\ref{sec: optimal noise}}, we delve into the question of when the curse of dimensionality is unavoidable and when it can be circumvented. With a specific focus on the generic Gaussian mechanism within the context of RDP, we introduce novel methods to prove optimality and characterize non-asymptotic optimal noise for a broad class of high-dimensional sensitivity sets. On one hand, we present a negative result, revealing that for a class of symmetric sensitivity sets, the curse of dimensionality is unavoidable, and isotropic Gaussian noise already achieves optimality (Theorem {\ref{thm: signAndperm}}). These symmetric sets include arbitrary mixtures of $l_p$-norm balls (Corollary {\ref{cor: l_p mix}}), where a single $l_2$-ball is a special case. On the other hand, we characterize the optimal noise form for generic hybrid clipping, which may allocate different clipping budgets to different subspaces. Remarkably, we show that the scale of the optimal noise can be $O(1)$, without explicit dependence on either the dimension or the rank of the release space.}

c). {In Section {\ref{sec: twice-sampling}}, we further strengthen privacy amplification from sampling in the low-sampling rate regime, and enforce the randomness of sampling to fit the sensitivity geometry restricted by the $l_{\infty}$-norm. 
We present a more fine-grained algorithmic analysis and introduce \textbf{twice sampling}, a novel method that employs both input-wise and coordinate-wise Poisson sampling to enhance efficiency and privacy simultaneously. 
We provide rigorous closed-form RDP analysis, demonstrating that with the assistance of the $l_{\infty}$-norm restriction, twice sampling achieves asymptotic improvement in privacy amplification (Theorems {\ref{thm: privacy coordiante sampling}}-{\ref{thm: sampling twice}}). 
This advancement alleviates limitations on low sampling rates or small additive noise in standard sample-wise sampling to produce useful amplification in higher-order RDP, as explained in Section {\ref{sec: simulation sampling}}. 
As a result, this fundamental improvement enables practical high-dimensional processing to attain better utility-privacy tradeoffs with low overhead, utilizing a small set of subsampled data.}

{We provide simple experimental results on the applications of DP-SGD for training ResNet22 on CIFAR10 and SVHN datasets in Section {\ref{sec: additional exp}}, and show the advantage of hybrid clipping and twice sampling with much sharpened noise bound. 
Even in these small-scale examples, we improved the noise variance by almost an order of magnitude compared to that of standard DP-SGD (with only $l_2$-norm clipping and input-wise sampling) under the same setup. 
Consequently, our results allow running DP-SGD in low sampling rates but with competitive performance as that from state-of-the-art empirical results in {\cite{deepmind}}, which implements nearly full-batch gradient descent after extensive fine-tuning. Our method thus provides a significant efficiency advantage. We discuss related works in Section {\ref{sec: related-works}}. We conclude and discuss the limitations of hybrid clipping in Section {\ref{sec: conclusion}}.} {{Fig. \ref{Fig. workflow}} illustrates the paper organization}. Additional discussions on the construction of clipping for practical black-box machine learning can be found in 
Appendix \ref{sec: add discussion}. 


\vspace{- 0.05 in}
\section{Preliminaries}
\label{sec: backgroud}
\noindent \textbf{$(\epsilon, \delta)$-DP and Rényi-DP (RDP)}: We first formally define the two widely studied and applied DP variants and show their relationship. 
\vspace{- 0.2 in}
\begin{definition}[Differential Privacy \cite{mironov2019r}]
\label{def:DP}
Given a universe $\mathcal{X}^*$, we say that two datasets $X,X'\subseteq \mathcal{X}^*$ are adjacent, denoted as $X \sim X'$, if $X = X' \cup x$ or $X' = X \cup x$ for some additional datapoint $x \in \mathcal{X}$. A randomized algorithm $\mathcal{M}$ is said to be $(\epsilon,\delta)$-differentially-private (DP) if for any pair of adjacent datasets $X,X'$ and any event set $O$ in the output domain of $\mathcal{M}$, it holds that
\begin{equation*}
\vspace{- 0.05 in}
	\mathbb{P}(\mathcal{M}(X)\in O)\leq e^{\epsilon}\cdot \mathbb{P}(\mathcal{M}(X')\in O)+\delta.
\end{equation*}
\end{definition}

\begin{definition}[Rényi Differential Privacy \cite{renyi}]
\label{def: RenyiDP}
A randomized algorithm $\mathcal{M}$ satisfies $(\alpha, \epsilon(\alpha))$-Rényi Differential Privacy (RDP), $\alpha >1$, if for any pair of adjacent datasets $X \sim X'$, $\mathsf{D}_{\alpha}(\mathbb{P}_{\mathcal{M}(X)} \| \mathbb{P}_{\mathcal{M}(X')} ) \leq \epsilon(\alpha).$ Here, $\mathbb{P}_{\mathcal{M}(X)}$ and $\mathbb{P}_{\mathcal{M}(X')}$ represent the distributions of $\mathcal{M}(X)$ and $\mathcal{M}(X')$, respectively, and 
 \begin{equation}
  \vspace{- 0.05 in}
 \label{Rényi-diver}
 \mathsf{D}_{\alpha}(\mathsf{P} \| \mathsf{Q}) = \frac{1}{\alpha-1} \log \int \mathsf{q}(o) (\frac{\mathsf{p}(o)}{\mathsf{q}(o)})^{\alpha} ~d o,
 \end{equation}
represents $\alpha$-Rényi Divergence between two distributions $\mathsf{P}$ and $\mathsf{Q}$ whose density functions are $\mathsf{p}$ and $\mathsf{q}$, respectively. 
\end{definition}
\vspace{- 0.05 in}
RDP can be used to elegantly handle the composition of privacy leakage.
The conversion from RDP to $(\epsilon, \delta)$-DP is characterized in the following lemma. 
\vspace{-0.05 in}
\begin{lemma}[Advanced Composition via RDP and Conversion \cite{renyi}]
\label{lemma: composition-RDP}
For any $\alpha>1$ and $\epsilon>0$, the class of $(\alpha, \epsilon(\alpha))$-RDP mechanisms satisfies $(\tilde{\epsilon}, \tilde{\delta})$-differential privacy under $T$-fold adaptive composition for any $\tilde{\epsilon}$ and $\tilde{\delta}$ such that
\[
\tilde{\epsilon} = T\epsilon(\alpha) - \log(\tilde{\delta})/(\alpha-1).
\]
\end{lemma}
\vspace{- 0.05 in}
\noindent For simplicity, in this paper we focus on RDP with integer $\alpha$. To randomize a deterministic algorithm $\mathcal{F}$ to produce the required security parameters, we need to characterize the sensitivity set and especially the dominating sensitivity element(s). 
\vspace{- 0.05 in}
\begin{definition}[Sensitivity Set] For a deterministic function $\mathcal{F}$, its sensitivity set $\mathsf{S}$ is defined as 
$$ \mathsf{S} = \{\bm{s} = \mathcal{F}(X) - \mathcal{F}(X'): X \sim X'\},$$
for any adjacent datasets $X$ and $X'$. 
\label{def: sensitivity set}
\end{definition}
\vspace{- 0.05 in}
For a given DP definition and a randomization strategy, the dominating sensitivity is the element in $\mathsf{S}$ that causes the maximal privacy loss. 
In RDP, if the sensitivity set is an $l_2$-norm ball of radius $c_2$, the Gaussian mechanism via adding isotropic noise from $\mathcal{N}(\bm{0},\sigma^2\cdot \bm{I})$ is known to produce $(\alpha, \alpha c^2_2/(2\sigma^2))$-RDP \cite{renyi}.  In the following, we formally define the clipping operator, essentially a projection. 

\begin{definition}[Clipping] For any vector $\bm{v}$ and any given set $\mathsf{S}$, a clipping operator $\mathcal{CP}$ on $\bm{v}$ with respect to $\mathsf{S}$ and a distance metric $\rho$ is defined as 
$\mathcal{CP}(\bm{v}, \mathsf{S}, \rho) = \arg \inf_{\bm{s} \in \mathsf{S}} \rho(\bm{v}, \bm{s}).$
\end{definition}
\vspace{- 0.05 in}
In this paper, the selection of $\rho$ is not our main focus, since we basically assume via clipping, the clipped output is within some given set $\mathsf{S}$. 
Thus, we will often use $\mathcal{CP}(\cdot)$ to denote some generic clipping operator. In particular, for the classic $l_2$-norm clipping with parameter $c_2$, it can be defined as $\mathcal{CP}(\bm{v}, c_2) = \bm{v}\cdot \min\{1, \frac{c_2}{\|\bm{v}\|_2}\}$, i.e., a projection to an $l_2$-ball $\mathcal{B}_2$ of radius $c_2$ under $\|\cdot\|_2$ metric distance, where we use $\|\cdot\|_p$ to denote the $l_p$-norm. 



\noindent \textbf{DP-SGD and Private Machine Learning}: 
In a supervised machine learning task, we are given a dataset $\{(x_i, y_i), i=1,2,\cdots,n\}$, where $x_i$ and $y_i$ represents feature and label, respectively, and a model $\mathsf{M}(x,w)$ with parameter $w$ to learn.
The objective optimization problem is 
\begin{equation*}
\vspace{-0.05 in }
    \min_{w} f(w) = \frac{1}{n}\cdot \sum_{i=1}^n l(\mathsf{M}(w,x_i),y_i),
\end{equation*}
where $l(\cdot,\cdot)$ is some loss function. 
DP-SGD can be described as follows. At the $t$-th iteration, with previous iterate $w^{(t-1)}$, we implement input-wise Poisson sampling with parameter $q$ to select a batch $S^{(t)}$ of samples from the entire set.
For each sample $(x_i,y_i)$, we calculate the per-sample gradient $\nabla l(\mathcal{N}(w,x_i),y_i)$. To ensure bounded sensitivity, most existing DP-SGD works usually adopt $l_2$-norm clipping and given some stepsize $\eta$, a noisy SGD is implemented as
\begin{equation}
         w^{(t)} = w^{(t-1)} - \eta\big(\sum_{(x_i,y_i) \in S^{(t)}}  \mathcal{CP}\big(\nabla l(\mathcal{N}(w,x_i),y_i)\big) + e^{(t)}\big),
    \label{DP-SGD}
    \vspace{-0.05 in}
\end{equation}
where $e^{(t)}$ is some Gaussian noise. Running for $T$ iterations with a total privacy budget $(\epsilon, \delta)$, one may select $e^{(t)} \sim \mathcal{N}(0,\sigma^2\cdot \bm{I}_d)$ where $\sigma = O({\sqrt{Td\log(1/\delta)}}/{\epsilon})$ by the composition bound \cite{CCS2016}. 
\vspace{- 0.05 in}
\section{A Warm-Up: Geometry of Practical High-dimensional Data}
\label{sec: warmup}
Before formally presenting our results, we aim to provide more intuition about the geometry of practical high-dimensional data. 
Two examples are presented: one involving biometric gene data and the other focusing on the gradients of a neural network, which serve as a fundamental component in deep learning. In each example, we divide the entire dataset into two equal parts. 
One half of the samples is considered public, while the other half is assumed private. 
{It is worth mentioning that this setup is designed to model the scenarios where one has access to public data or possesses prior knowledge about the sensitive data to process. The goal is to construct efficient clipping methods that can accurately represent the processed output distribution. At this stage, we do not consider the privacy implications.}

{Another critical motivation behind these experiments is to evaluate the performance of classic dimension-reduction clipping methods, such as sparsification {\cite{Luo2021sparse}}, {\cite{zhang2021sparse}}, {\cite{Zhu2021sparse}} (preserving only significant coordinates) or low-rank embedding {\cite{embedyu2021}} (projection to a subspace). From a theoretical perspective, these strategies can artificially alleviate the curse of dimensionality, as the scale of noise is now determined by the Hamming weight after sparsification or the rank of embedding. However, their corresponding clipping bias remains largely unclear in practice. In particular, if these approaches fail to capture practical output distributions, one crucial question we have to answer is: what other features can we reliably learn (from public data) to design improved clipping techniques?}

\vspace{- 0.1 in}
\subsection{Biometric Gene Data}

\label{sec: bio data}
We adopt the {\em Gene Expression Cancer RNA-Seq Dataset} from the UCI Machine Learning Repository\footnote{\url{https://archive.ics.uci.edu/ml/datasets/gene+expression+cancer+RNA-Seq}}, which contains $800$ samples, each of dimension $d=20,531$. 
As mentioned above, we evenly split the data into two parts and each sample is normalized to $1$ in $l_2$-norm. 
For private data, in Fig. \ref{Fig. gene statistics}(a), we plot the absolute value of coordinate-wise {\em coefficient of variation}, which is the ratio between the standard deviation and the mean of each coordinate. 
Higher coefficient of variation suggests greater dispersion.  
We can see that the individual private sample is of heavy diversity and most coordinates bear large dispersion. 

We then compute the average of the magnitude of each coordinate from public samples as a measurement of significance, and sort their indices in a descending order of the significance score. 
In Fig. \ref{Fig. gene statistics}(b), we consider a sparsification method where we only preserve the first $\beta\%$ coordinates of largest significance for each private sample. 
The x-axis represents the quantile $\beta\%$ and the y-axis records the average of the $l_2$-norm of the residual component. 
Here, we define the residual component as the remaining $(1-\beta\%)$ less significant coordinates. 
Still, such approximation error does not drop sharply as $\beta$ increases, which means that the data distribution does not enjoy a strongly concentrated sparsity.

We then consider implementing more involved Principal Component Analysis (PCA) \cite{abdi2010principal} on the public data. 
In Fig. \ref{Fig. gene statistics} (c), we consider projecting each private sample to the principal subspace spanned by the eigenvectors of the $k$ largest eigenvalues from public data. 
The x-axis represents the $k$ for $k=1,2,\cdots,400$, given that we only have 400 public samples, and the y-axis records the $l_2$-norm of the residual component. 
Compared to Fig. \ref{Fig. gene statistics}(b), low-rank embedding produces a better approximation.
However, the residual remains as a non-negligible component. 
Finally, we record the mean and the variance of the $l_2$-norm of private samples' projection into the $k$-th principal space (spanned by the largest $k$ eigenvectors), for $k=1,2,\cdots, 400$, in Fig. \ref{Fig. gene statistics}(d), and per-sample $l_{\infty}$-norm in Fig. \ref{Fig. gene statistics}(e). 
Interestingly, the $l_2$-norm of the main components is a much more stable statistic, whose standard deviation is only about $0.015$. 
Moreover, it is noted that the $l_{\infty}$-norm is much smaller than the $l_2$-norm and actually mostly smaller than $0.02$. 
This is not surprising, given our observation in Fig. \ref{Fig. gene statistics}(b). 
The data is not of strong sparsity and its power is shared by many coordinates.   
\begin{figure*}[t] 
    \subfigure[Coordinate-wise Coefficient of Variation]
    {\includegraphics[width=.19\linewidth]{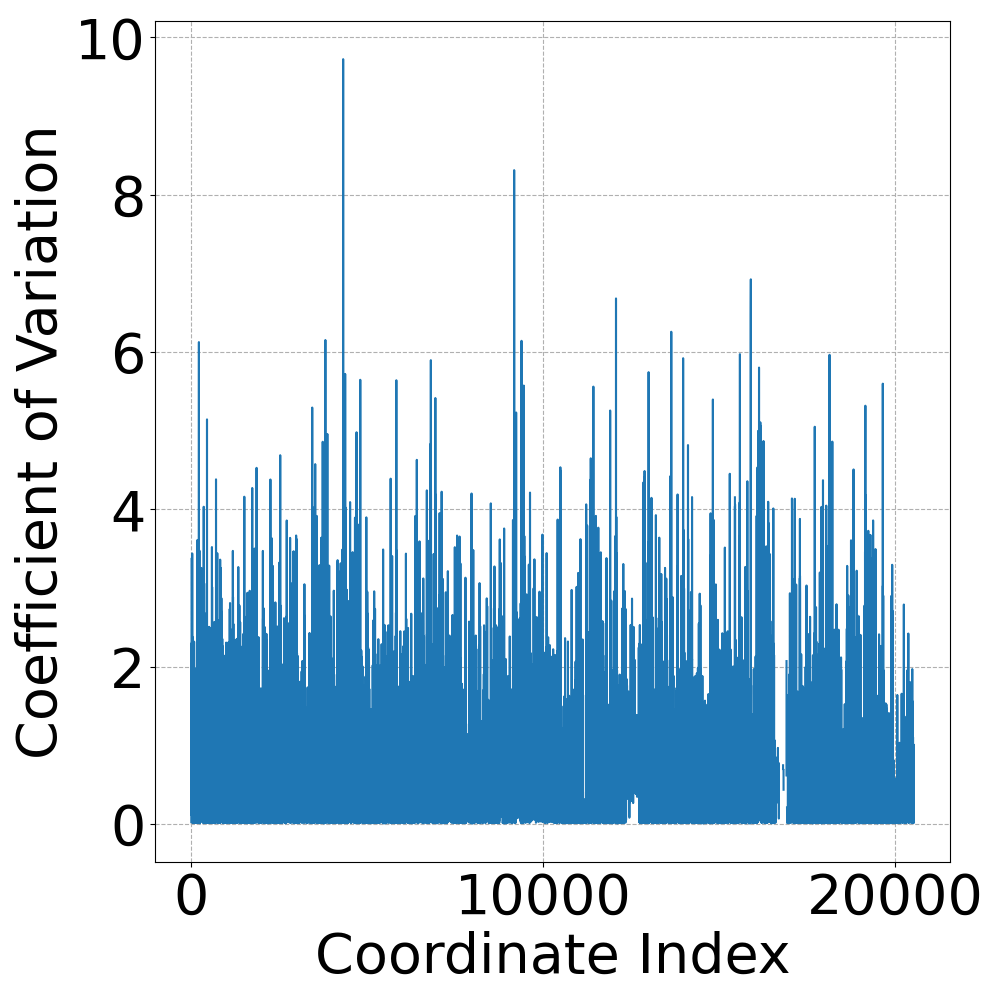}
    \label{Fig: bio_a}}
    \subfigure[Sparsification]
    {\includegraphics[width=.19\linewidth]{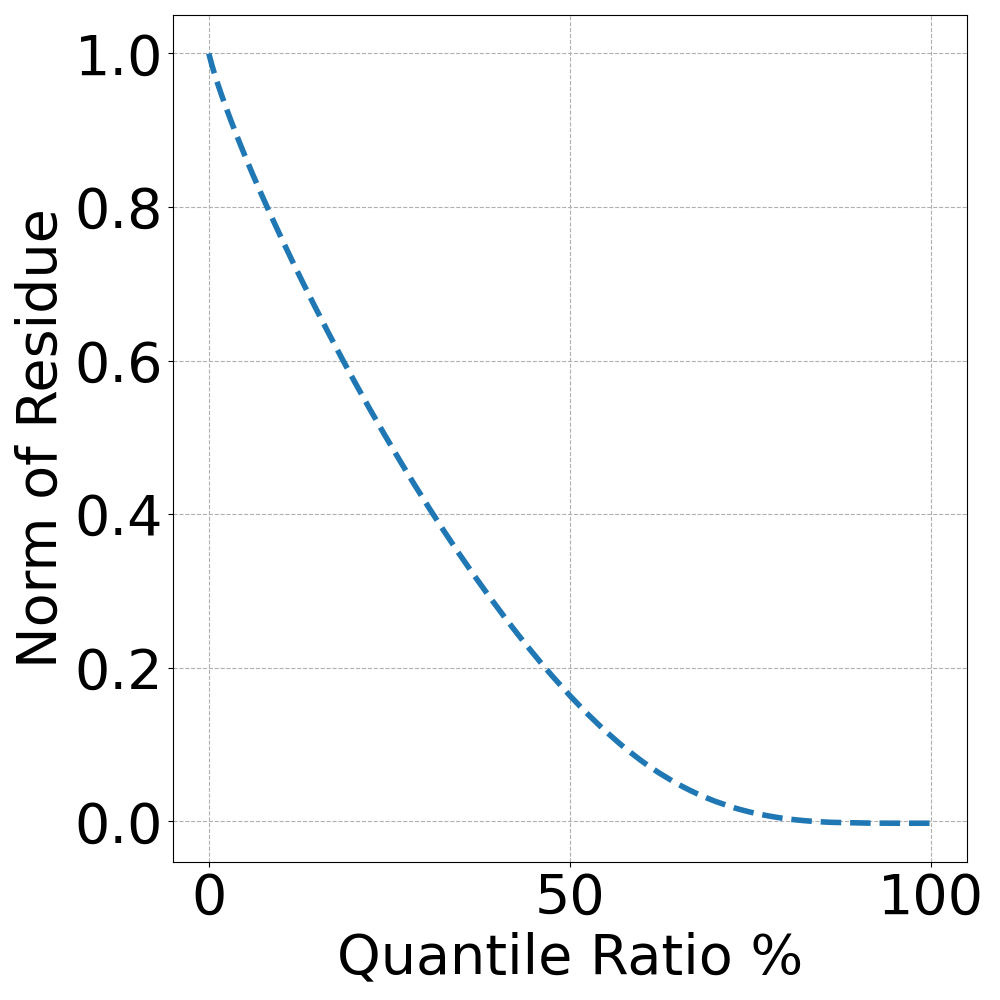}
    \label{Fig: bio_b}}
    \subfigure[Low-rank Embedding]
    {\includegraphics[width=.19\linewidth]{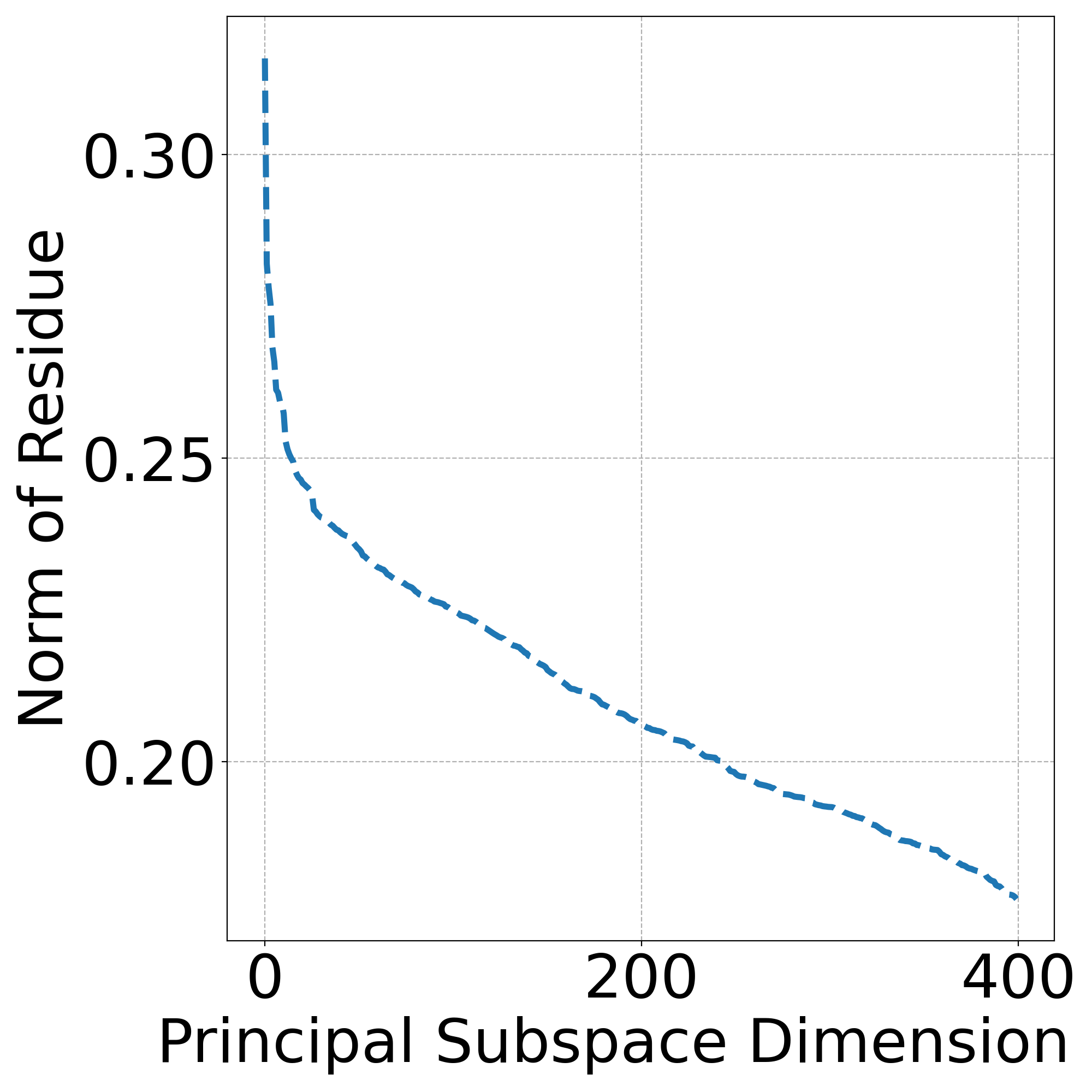}
    \label{Fig: bio_c}}
    \subfigure[$l_2$ Norm of Principal Component]
    {\includegraphics[width=.19\linewidth]{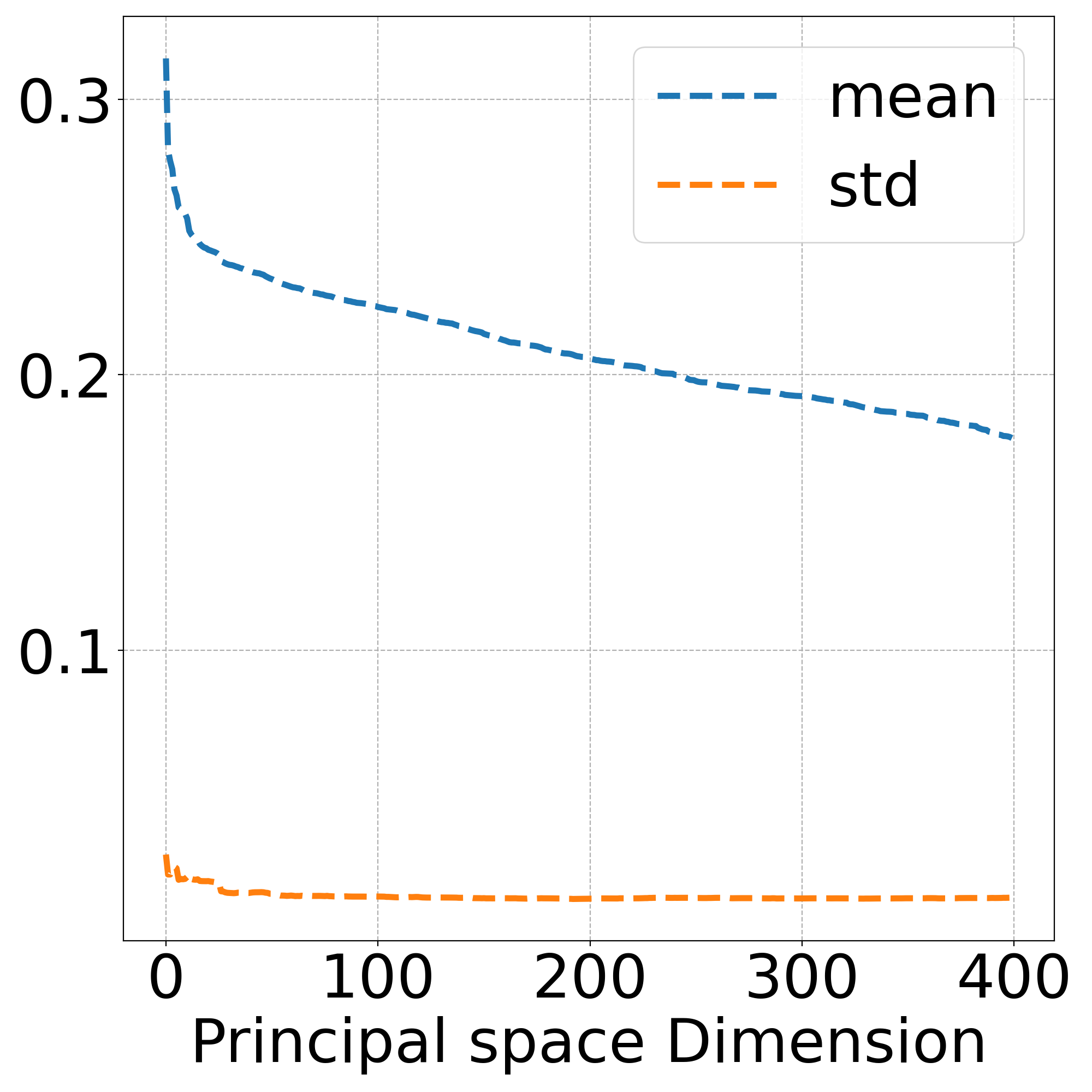}
    \label{Fig: bio_d}}
    \subfigure[$l_{\infty}$ Norm]
    {\includegraphics[width=.19\linewidth]{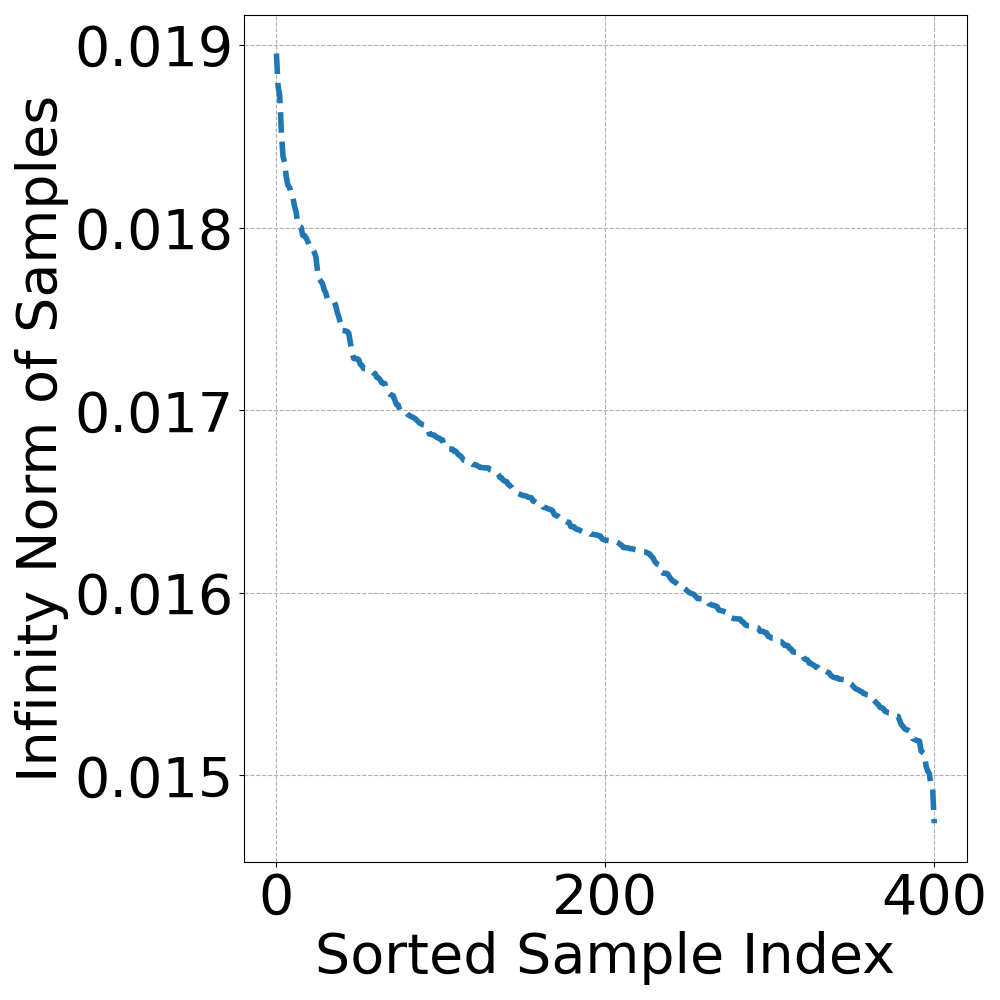}
    \label{Fig: bio_e}}\
    \vspace{-0.2 in}
    \caption{Statistics of High-dimensional Gene Data with Sparsification and Low-rank Approximation}
    \vspace{-0.2 in}
    \label{Fig. gene statistics} 
\end{figure*}
\begin{figure*}[t] 
    \subfigure[Coordinate-wise Coefficient of Variation]
    {\includegraphics[width=.19\linewidth]{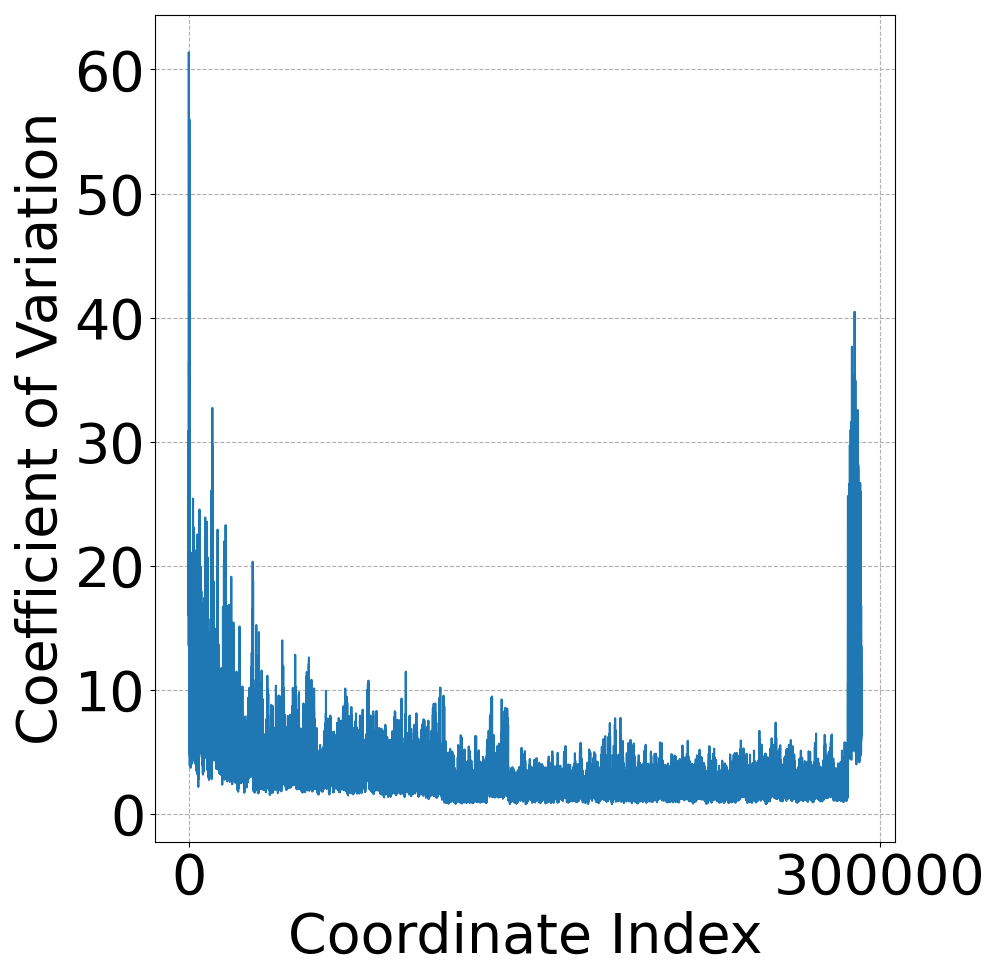}
    \label{Fig: grad_a}}
    \subfigure[Sparsification]
    {\includegraphics[width=.19\linewidth]{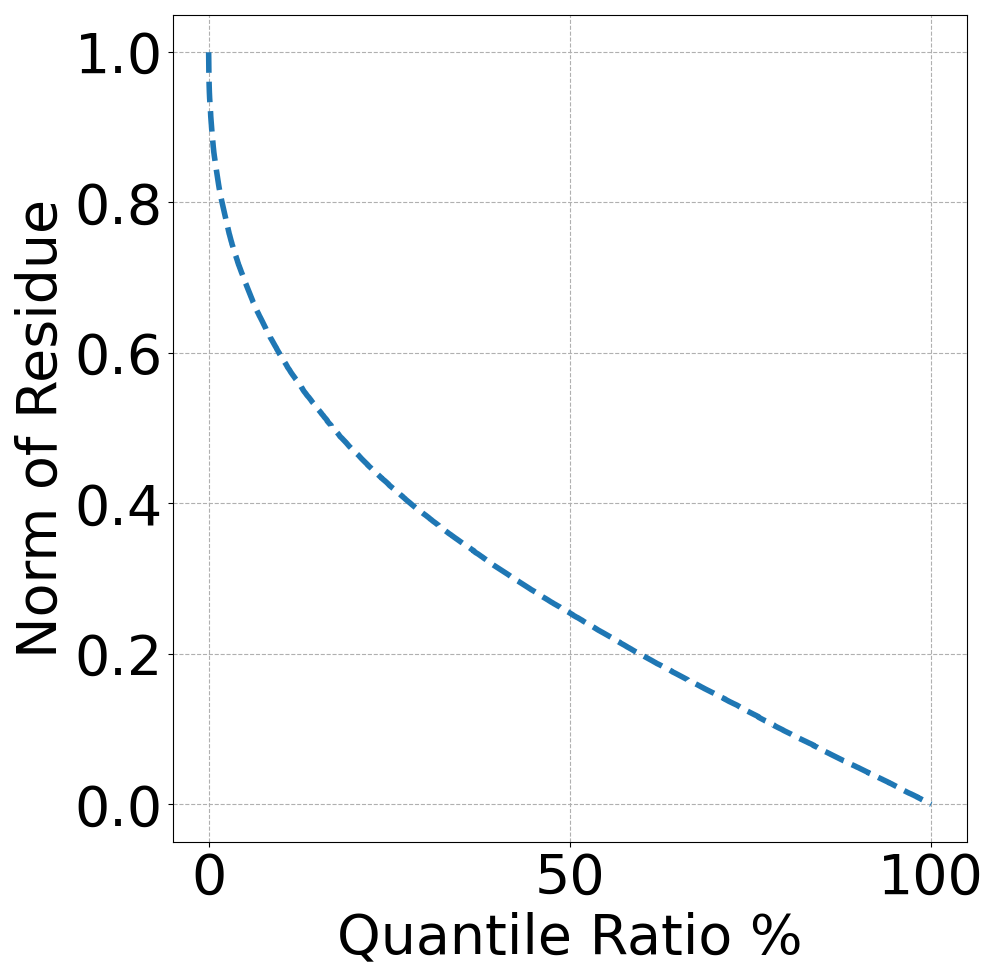}
    \label{Fig: grad_b}}
    \subfigure[Low-rank Embedding]
    {\includegraphics[width=.19\linewidth]{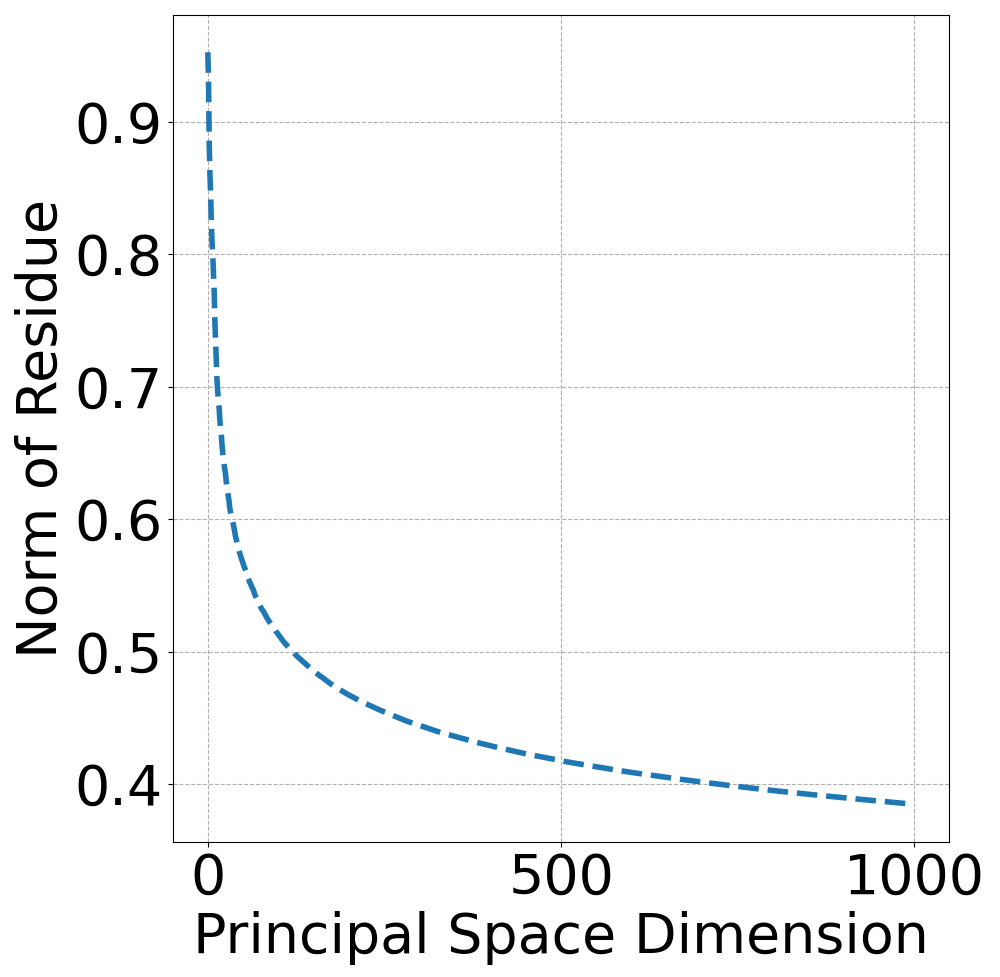}
    \label{Fig: grad_c}}
    \subfigure[$l_2$ Norm of Principal Component]
    {\includegraphics[width=.19\linewidth]{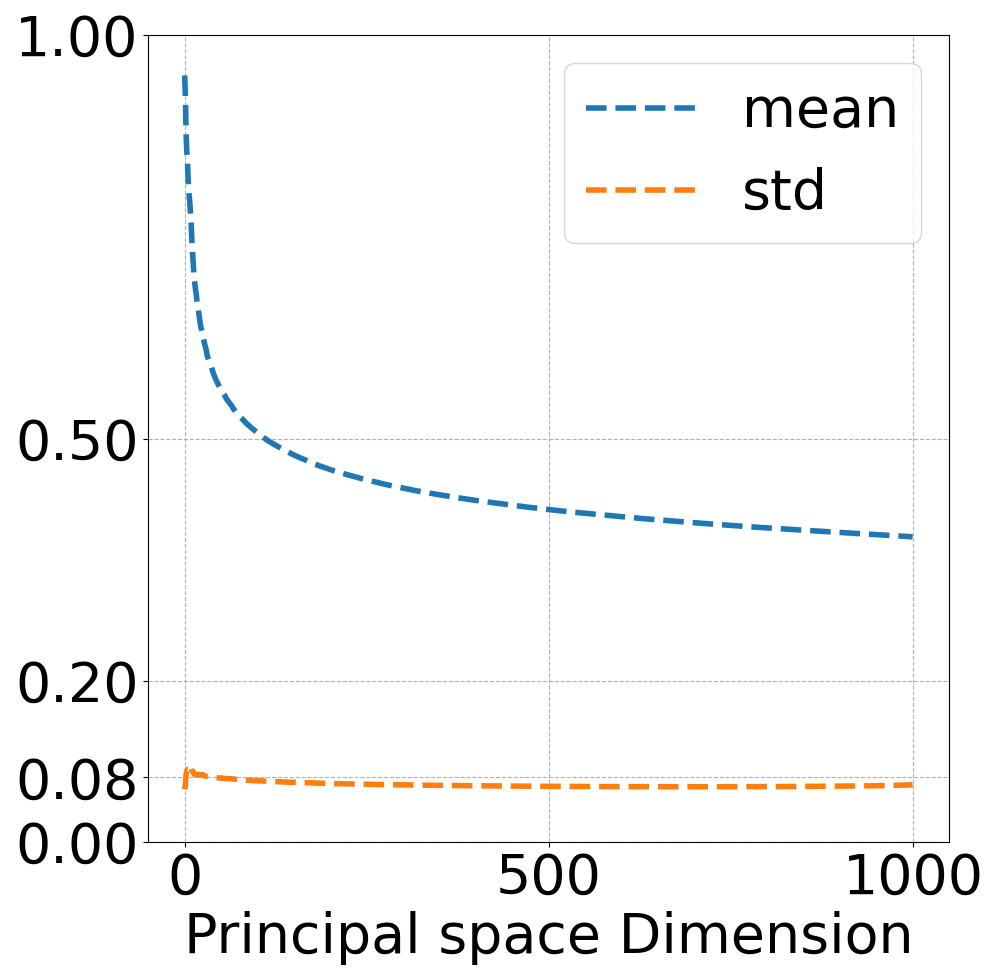}
    \label{Fig: grad_d}}
    \subfigure[$l_{\infty}$ Norm]
    {\includegraphics[width=.19\linewidth]{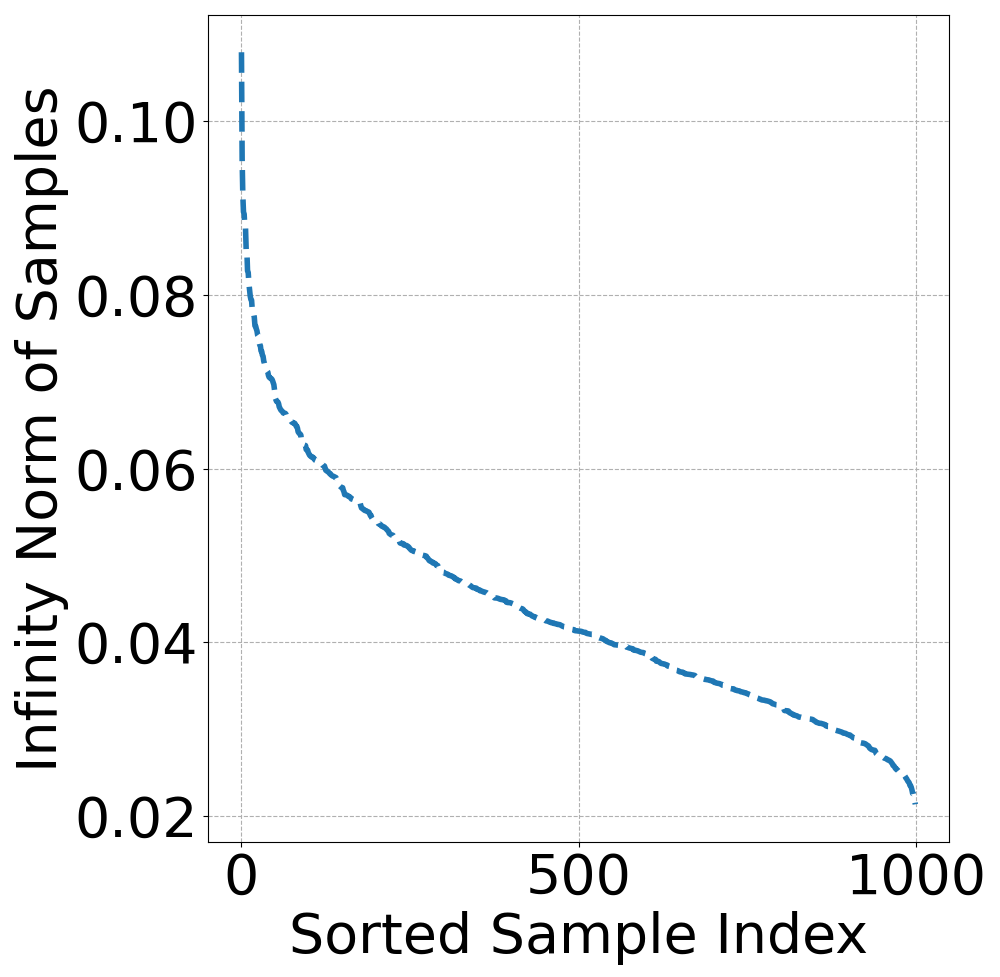}
    \label{Fig: grad_e}}\
    \vspace{-0.2 in}
    \caption{Statistics of Per-sample Gradient of ResNet22 on CIFAR 10 with Sparsification and Low-rank Approximation}
    \vspace{-0.15 in}
    \label{Fig. grad statistics} 
\end{figure*}

\vspace{- 0.1 in}
\subsection{Stochastic Gradient in Deep Learning}
\label{sec: gradient data}
We further study the stochastic gradient of ResNet22 \cite{resnet} on CIFAR10 \cite{cifar}, a benchmark dataset for an object recognition task in image processing. 
The number of parameters in ResNet22 is $291, 898$ which is also the dimension $d$ of the gradient. 
We select 2,000 samples from the CIFAR10 set and similarly split them into the public and private subsets, each of 1,000 samples. 
We run gradient descent using the private samples and record the private per-sample gradients in the 100th iteration, where the gradient descent has already entered a stable convergence phase. 
We also evaluate the gradients of the public data at the same iteration. 
The 2,000 private and public  per-sample gradients are clipped to 1 in $l_2$-norm. 
We conduct the same experiments as described in Section \ref{sec: bio data} and the results are shown in Fig. \ref{Fig. grad statistics}. 
In this more complicated and higher-dimensional example, the dispersion of per-sample gradients is even more significant with a larger residue component in both sparsification and low-rank approximation. However, stochastic gradient also shares very similar properties to gene data, where from Fig. \ref{Fig. grad statistics}(d,e), the norm of the principal component is stable with standard deviation of about $0.08$ and mostly the per-sample gradient's $l_{\infty}$-norm is smaller than $0.1$. This suggests that putting an additional $l_{\infty}$-norm with parameter $c_{\infty} \geq 0.1$ on already clipped per-sample gradient in $l_2$-norm to $1$ generally does {\em not} cause any change to utility loss. 

\vspace{- 0.1 in}
\subsection{A Short Summary}
In some simpler cases where the data is {\em concentrated}, and of strong sparsity or largely distributed in some low-rank space, artificial sparsification or low-rank embedding could significantly mitigate the dimensionality curse even using the current noise mechanism by post-processing projection.
Here, we have to stress that the {\em concentration} requirement is because, in most applications of DP, we need to clip each individual rather than an aggregation. 
Thus, even if populationally the data distribution has desired properties, we may still not guarantee good approximation for each individual with a large clipping error \cite{xiao2023theory}, and let along the scenario of heavy-tailed data distribution \cite{heavy_tail1,heavy_tail2}. 
The two examples presented above are cases where simple sparsification or low-rank embedding lead to large bias. 
{However, such failure also has two very interesting and meaningful implications that open up a new possibility: we can still learn from the fact that the distribution is neither sparse nor low-rank concentrated, and construct useful clipping!}

First, weaker sparsity suggests that the data is more randomly distributed across the entire space. 
Thus, in general, we may expect a small $l_{\infty}$-norm given that the data does not concentrate on few coordinates. 
Second, though having high variance, the norm of the component of each sample projected in some subspace is an aggregated statistic, which is usually more stable when the rank of the subspace is larger. More importantly, from Fig. \ref{Fig. gene statistics}(c) and Fig. \ref{Fig. grad statistics}(c), though the residue component is in a heavy tail which decays slowly, the scale of data in different subspaces is not uniform or identical. 
Therefore, with the consideration of both clipping bias and distribution geometry, a more smooth clipping method is to split the whole $d$-dimensional space into multiple (relatively large) subspaces/blocks and assign different clipping budgets to each of them. 
Besides, an additional proper $l_{\infty}$-norm clipping can be implemented afterwards, which is mostly free of making changes to the output. 

Compared to clipping simply by $l_2$-norm, the above-mentioned hybrid clipping puts more restriction on the produced sensitivity set. 
Intuitively, we should expect better utility-privacy tradeoff compared to the case where we only know the worst-case $l_2$-norm. 
However, construction of randomization to reflect such sensitivity restriction is non-trivial. 
The classic Gaussian mechanism injects noise only determined by the worst-case $l_2$-norm or alternatively one may separately apply the Gaussian mechanism to each subspace and derive an upper bound of privacy loss via composition \cite{embedyu2021}. 
Unfortunately, as we will show in Section \ref{sec: optimal noise}, this strategy could be far from optimal. Moreover, a more complicated restriction like $l_{\infty}$-norm cannot be captured in such a manner. We will provide more intuition through examples and Fig. \ref{Fig: illustration} in Section \ref{sec: additional exp}. 
The goal of the remainder of this paper is to study the {\em non-asymptotically} optimal Gaussian noise for a wide class of sensitivity sets and accordingly construct proper randomness to fit the high-dimensional geometry we observe here.  

\vspace{-0.05 in}
\section{Sensitivity Geometry and Optimal Gaussian Noise}
\label{sec: optimal noise}
Before proceeding, we first formally define the problem of optimal Gaussian noise in terms of minimal variance in the context of RDP. 
Given a deterministic processing function $\mathcal{F}$ and its corresponding sensitivity set $\mathsf{S} \subset \mathcal{R}^d$, we set out to determine a $d$-dimensional multivariate Gaussian noise $\bm{e} \sim \mathcal{N}(\bm{0}, \Sigma_0)$, where $\Sigma_0$ is a $d \times d$ covariance matrix, satisfying
\begin{equation}
\begin{aligned}
 &\inf_{\Sigma_0} \Tr(\Sigma_0)~ \text{s.t.}~ \sup_{ s\in \mathsf{S}}  \mathcal{D}_{\alpha} \big(\mathcal{N}(\bm{0},\Sigma_0)\| \mathcal{N}(s,\Sigma_0)\big) \leq \epsilon(\alpha), 
\end{aligned}
\label{optimal main}
\end{equation}
for a required $(\alpha, \epsilon(\alpha))$-RDP guarantee. 
Since the covariance matrix $\Sigma$ must be (semi) positive definite, its singular value decomposition (SVD) can be expressed as $\Sigma_0 = U\Sigma U^T$ where $U$ is some unitary matrix formed by its eigenvectors and $\Sigma = \text{Diag}\big(\sigma^2_1, \sigma^2_2, \cdots, \sigma^2_d \big)$ is a diagonal matrix. 
Thus, the objective we want to minimize is $\mathbb{E}\|\bm{e}\|^2 = \Tr(\Sigma_0)= \Tr(\Sigma)$, the trace of $\Sigma$. 
The second inequality in (\ref{optimal main}) captures the constraint to produce an $(\alpha, \epsilon(\alpha))$-RDP. 
Recall Definition \ref{def: RenyiDP}, $(\alpha, \epsilon(\alpha))$-RDP is equivalent to saying that for arbitrary two adjacent datasets $X$ and $X'$, $\mathcal{D}_{\alpha}\big(\mathcal{N}(\mathcal{F}(X),\Sigma_0)\|\mathcal{N}(\mathcal{F}(X'),\Sigma_0)\big) \leq \epsilon(\alpha).$ 
By the translation invariance of Rényi divergence, 
\begin{align*}
  \mathcal{D}_{\alpha}\big(\mathcal{N}& (\mathcal{F}(X),\Sigma_0)\|\mathcal{N}(\mathcal{F}(X'),\Sigma_0)\big)\\
  & =\mathcal{D}_{\alpha}\big(\mathcal{N}(\bm{0},\Sigma_0)\|\mathcal{N}(\mathcal{F}(X')-\mathcal{F}(X),\Sigma_0)\big),  
\end{align*}
where a uniform shift on the distributions by $-\mathcal{F}(X)$ does not change the divergence. 
Thus, the RDP definition can be transformed to the version in (\ref{optimal main}) via the worst case on the sensitivity set $\mathsf{S}$. 
The $\alpha$-Rényi divergence between two multivariate Gaussians indeed has a closed form \cite{van2014renyi} and (\ref{optimal main}) can be rewritten as 
\begin{equation*}
\begin{aligned}
&\inf_{U, \Sigma=\text{Diag}\{\sigma^2_1, \cdots, \sigma^2_d\}} \sum_{l=1}^d \sigma^2_l,\\
&\text{s.t.}~ \sup_{ \bm{s}\in \mathsf{S}} \frac{\alpha}{2} \cdot \bm{s}U\cdot \Sigma^{-1} \cdot (\bm{s}U)^T \leq \epsilon(\alpha).
\end{aligned}
\vspace{- 0.05 in}
\end{equation*}
Therefore, for fixed privacy guarantee $(\alpha, \epsilon(\alpha))$, determining the minimal noise is equivalent to finding a unitary transform matrix $U$ such that the $\alpha$-Rényi divergence is minimal conditioned on the noise's variance being $1$, i.e., $\sum_{l=1}^d \sigma^2_l=1$.
This can be formally stated as a min-max problem as follows, 
\begin{equation}
\begin{aligned}
\inf_{U, \bm{\sigma}} \mathcal{L}(U,\bm{\sigma},\mathsf{S}) &= \inf_{U, \bm{\sigma}}\sup_{ \bm{s}\in \mathsf{S}}  \bm{s}U\cdot \Sigma^{-1}\cdot (\bm{s}U)^T,  \\
&\text{s.t.}~
\sum_{l=1}^d \sigma^2_l=1,
\end{aligned}
\vspace{- 0.05 in}
\label{optimal main 2}
\end{equation}
where $\bm{\sigma}=(\sigma_1, \cdots, \sigma_d)$. In the following, we will use $\mathcal{L}(U,\bm{\sigma},\mathsf{S}) = \sup_{ \bm{s}\in \mathsf{S}}  \bm{s}U\cdot \Sigma^{-1}\cdot (\bm{s}U)^T=  \sup_{ \bm{s}\in \mathsf{S}}  \bm{s}U\cdot \text{Diag}\{\sigma^{-2}_1, \cdots, \sigma^{-2}_d\}\cdot (\bm{s}U)^T$ to represent the target privacy loss function. Below, we will present two sets of results to answer the above min-max problem for a broad class of sensitivity sets, which characterize the optimal noise for most commonly-used clipping methods. 

\vspace{- 0.1 in}
\subsection{Symmetric Sensitivity Set}
We first consider the scenario where the sensitivity set $\mathsf{S}$ satisfies certain symmetry properties, formally defined as follows. 
\vspace{- 0.05 in}
\begin{definition}[Sign Invariance]
A set $\mathsf{S}$ satisfies \textit{sign invariance} if for any $(z_1, \dotsb, z_d)\in \{-1, 1\}^d$ and any $\bm{s} = (s_1, \dotsb, s_d)\in \mathsf{S}$, $(z_1 s_1, \dotsb, z_d s_d)$ is also in $S$.
\label{def: sign invariance}
\end{definition}
\vspace{- 0.1 in}
\begin{definition}[Permutation Invariance] 
A set $\mathsf{S}$ satisfies \textit{permutation invariance} if for any permutation $\pi$ over $\{1, 2, \dotsb, d\}$ and any $\bm{s} = (s_1, \dotsb, s_d)\in \mathsf{S}$, $(s_{\pi(1)}, \dotsb, s_{\pi(d)})$ is also in $S$.
\label{def: permutation invariance}
\end{definition}
\vspace{- 0.05 in}
Definitions \ref{def: sign invariance} and \ref{def: permutation invariance} basically say that for any element $\bm{s} \in \mathsf{S}$, when we arbitrarily change the sign of, or permute its coordinates, the resultant element is still within $\mathsf{S}$. The following theorem shows that if $\mathsf{S}$ satisfies Definitions \ref{def: sign invariance} and \ref{def: permutation invariance}, then for any selection of unitary matrix $U$, the $\inf_{\bm{\sigma}}\mathcal{L}(U,\bm{\sigma},\mathsf{S})$ is identical, and the isotropic Gaussian is already the optimal.  
\vspace{- 0.05 in}
\begin{restatable}[Optimal Noise for Symmetric $\mathsf{S}$]{theorem}{OptimalSymmetric}
If the sensitivity set $\mathsf{S}$ is invariant to sign and permutation as defined in Definitions \ref{def: sign invariance} and \ref{def: permutation invariance}, conditional on $\sum_{i = 1}^d \sigma^2_i = 1$, the optimal privacy loss is achieved when we select $\sigma_1 = \sigma_2 = \dotsb = \sigma_d = 1 / \sqrt{d}$ and is independent of the selection of $U$.
\label{thm: signAndperm}
\end{restatable}

\begin{proof}
See Appendix \ref{app: thm: signAndperm}.
\end{proof}
\vspace{- 0.05 in}
Theorem \ref{thm: signAndperm} is a negative result: for symmetric $\mathsf{S}$, the curse of dimensionality is tight. 
Due to the invariance to $U$, we simply select $U=\bm{I}_d$ to be the identity matrix, and $\sigma_1=\sigma_2 = \cdots = \sigma_d$ is identical to $\sigma_0$, then
\[
\mathcal{L}(U,\bm{\sigma},\mathsf{S}) = \sup_{\bm{s} \in \mathsf{S}} \frac{\|\bm{s}\|_2^2}{\sigma^2_0},
\]
and thus the minimal variance of noise $\bm{e}$ is just $\sigma^2_0d$ and $\sigma_0$ is only determined by the worst-case $l_2$-norm of the elements in $\mathsf{S}$.  

An immediate corollary from Theorem \ref{thm: signAndperm} is that if we use a mixture of $m$ kinds of $l_{p}$-norm clippings of parameters $\big((p_1,c_{p_1}), \cdots,$ $(p_m,c_{p_m})\big)$, respectively, and the resultant set  $\mathsf{S}$ is in a form 
\[
\mathsf{S} = \cap_{j=1}^m \{\bm{s}~|~\|\bm{s}\|_{p_j} \le c_{p_j}\},
\]
which is the intersection of $m$ many $l_{p_j}$ balls of radius $c_{p_j}$, respectively, then it is not hard to verify that such $\mathsf{S}$ is also invariant to sign and permutation. 
Thus, the optimal strategy, formalized by Corollary \ref{cor: l_p mix}, is still to add isotropic noise where the deviation of each coordinate is proportional to the maximal of the largest $l_2$-norm in the set.
\vspace{- 0.05 in}
\begin{corollary}[Optimal Noise for Mixture $l_p$-norm Clipping]
If $\mathsf{S} = \cap_{j=1}^m \{\bm{s}~|~\|\bm{s}\|_{p_j} \le c_{p_j}\}$, where $p_j$ and $c_{p_j}$ are positive real numbers for $j=1,2,\cdots,m$, then the optimal Gaussian noise to achieve arbitrary required $(\alpha, \epsilon(\alpha))$-RDP is in an isotropic form $\mathcal{N}(\bm{0},\sigma_0 \cdot \bm{I}_d)$, where there exists some constant $b_0$ determined by $\alpha$ and $\epsilon(\alpha)$ such that 
$$ \sigma_0 = b_0 \cdot \max_{\bm{s}}\{\|\bm{s}\|_2~|~\forall j \in \{1,2,\cdots, m\}, \|\bm{s}\|_{p_j} \leq c_{p_j}\}.$$
\label{cor: l_p mix} 
\end{corollary}
\vspace{- 0.2 in}
Thus, still as a negative result, from Corollary \ref{cor: l_p mix}, Gaussian noise cannot capture the gain from the {\em additional} $l_{\infty}$-norm restriction (discussed in Section \ref{sec: warmup}), unless it becomes trivial to decrease the global $l_2$-norm bound of $\mathsf{S}$. 
In Section \ref{sec: twice-sampling}, we will show how to address this problem and utilize the $l_{\infty}$-norm using different methods and randomization. 
\vspace{- 0.1 in}
\subsection{Hypercube Sensitivity Set}
\label{sec: hypercube}
Given the negative results on symmetric $\mathsf{S}$ and the observation from Section \ref{sec: warmup} where, in general, for learnable high-dimensional distribution, the power of $\mathsf{S}$ will not be uniform across the entire space, we are motivated to consider the asymmetric case. 
We consider the following scenario that, on some set of orthogonal unit basis vectors $\bm{u}_1, \bm{u}_2, \cdots, \bm{u}_d$, where $\bm{u}_l \in \mathbb{R}^d$, $\|\bm{u}_l\|_2=1$, for $l=1,2,\cdots, d$, $\mathsf{S}$ is a hypercube in a form
\begin{equation}
    \mathsf{S} = \{\bm{s}=\sum_{l=1}^d v_l\bm{u}_l: v_l \in [-V_l,V_l], l=1,2,\cdots, d\}.
    \label{hypercube}
\end{equation}
In other words, the projection of $\mathsf{S}$ along any base $u_l$ is an interval $[-V_l,V_l]$ for some non-negative constant $V_l$, and $l=1,2, \cdots, d$. Interestingly, we will prove in the following theorem that the scale of the optimal noise does not need to be explicitly dependent on either the dimension $d$ or the rank (the number of non-zero $V_l$). 
\vspace{- 0.05 in}
\begin{restatable}[Optimal Noise for Hypercube]{theorem}{OptimalHyperCube}
If  $\mathsf{S}$ is a hypercube defined in (\ref{hypercube}), 
then the optimal privacy loss is achieved when we select $U=(\bm{u}_1, \dotsb, \bm{u}_d)$ and select $\bm{\sigma}=(\sigma_1, \cdots, \sigma_d)$ such that $\sigma_l = \sqrt{V_l \over \sum_{j=1}^d V_j}$. Or equivalently, to achieve required $(\alpha, \epsilon(\alpha))$-RDP, there exists some constant $b_0$ such that for the optimal noise $\bm{e}$, $\sigma_l =  b_0 \cdot \sqrt{V_l \cdot \sum_{j=1}^d V_j}$ and $\mathbb{E}[\|\bm{e}\|^2] = b^2_0 \cdot (\sum_{l=1}^d V_l)^2.$
\label{thm: cube}
\end{restatable}
\vspace{- 0.05 in}
\begin{proof}
    See Appendix \ref{app: thm: cube}.
\end{proof}
\vspace{- 0.05 in}
Theorem \ref{thm: cube} states that if we know $\mathsf{S}$ is a hypercube under some basis, then to produce the optimal Gaussian noise, we should also select the unitary $U$ formed by the exact basis and add noise of variance $\sigma^2_l$ along $\bm{u}_l$ proportional to $V_l$. 
Thus, the optimal noise scale $\mathbb{E}[\|\bm{e}\|]$ is determined by the $l_1$-norm of the vector $(V_1, V_2, \cdots, V_l)$, i.e., the sum of side lengths of the hypercube. 
If its $l_1$-norm is constant, then we only need to add {\em constant} noise independent of the dimension $d$ or the rank. 
This is an elegant example where noise fits the geometry and we only add the {\em necessary} amount to each direction.
Simply using an isotropic noise could be far from optimal. 

With Theorem \ref{thm: cube}, we can also show the optimal noise for hybrid clipping, where we assign different clipping budgets to $m$ orthogonal subspaces. 
Suppose the $j$-th subspace is of rank $r_j$ and $\sum_{j = 1}^m r_j = d$. For a hybrid clipping, we clip the projection of the release in the $j$-th subspace to the $l_2$-norm of parameter $c_{2j}$. 
Without loss of generality, we transform $\mathsf{S}$ back to the representation with the natural one-hot unit bases, and the produced sensitivity set $\mathsf{S}$ is in a form
\begin{equation}
    \mathsf{S} = \{\bm{s}=(\bm{s}_1, \bm{s}_2, \dotsb, \bm{s}_m)~|~ \bm{s}_j \in \mathbb{R}^{d_j} ~~\text{and}~~ \|\bm{s}_j\|_2 \le c_{2j}\}.
    \label{sensitivity set hybrid clipping}
\end{equation}
Notice that $\mathsf{S}$ is invariant to sign, but it is only invariant to permutation in each subspace (segment). We will show in the following theorem that the optimal noise is to still use the original basis and select $\sigma^2_j \propto c_{2j} / \sqrt{r_j}$
for all bases in the $j$-{th} subspace. The proof is a combination of Theorems \ref{thm: signAndperm} and \ref{thm: cube}.
\vspace{- 0.05 in}
\begin{restatable}{theorem}{OptimalHybrid}
Given a hybrid clipping with a sensitivity set $\mathsf{S}$ described in (\ref{sensitivity set hybrid clipping}), 
the optimal noise is to select $U$ formed by the same basis and the standard deviation $\sigma_j$ is in a form {
$\sigma_j = \sqrt{c_{2j} \over  \sqrt{r_j} \sum_{l=1}^m c_{2l}\sqrt{r_l}}$} for all bases in the $j$-th subspaces. Or equivalently, to achieve required $(\alpha, \epsilon(\alpha))$-RDP, there exists some constant $b_0$ determined by $\alpha$ and $\epsilon(\alpha)$ such that the optimal noise $\bm{e}$ is to add isotropic noise in each subspace with {$\sigma_j  = b_0\cdot \sqrt{c_{2j}\sum_{l=1}^m c_{2l}\sqrt{r_l} \over \sqrt{r_j} }$}
and $\mathbb{E}[\|\bm{e}\|^2] = b^2_0 \cdot \big(\sum_{l=1}^m c_{2l}\sqrt{r_l}\big)^2$.
\label{thm: hybrid clipping}
\end{restatable}
\vspace{- 0.05 in}
\begin{proof}
    See Appendix \ref{app: thm: hybrid clipping}.
\end{proof}
\vspace{- 0.05 in}

{Comparing Theorems {\ref{thm: cube}} and {\ref{thm: hybrid clipping}}, the hybrid clipping indeed captures a more coarse but generic partition of the entire space. 
In Theorem {\ref{thm: cube}}, we specify the power (budget) of the sensitivity set (clipping) along the direction of each basis, which, strictly speaking, represents $d$ rank-1 subspaces. 
However, in Theorem  {\ref{thm: hybrid clipping}}, we only consider $m$ subsapces (basis subsets) and assign a local $l_2$-norm bound on each. Thus, Theorem {\ref{thm: cube}}  is a special case of Theorem  {\ref{thm: hybrid clipping}} if we fix $r_j=1$.  Moreover, it is not surprising that in Theorem  {\ref{thm: hybrid clipping}}, the optimal noise bound is still in a weighted average form, which is determined by both the local dimension $r_j$ and the $l_2$-norm power $c_{2j}$.}

We can compare Theorem \ref{thm: hybrid clipping} with the standard DP analysis by composition. 
If we apply the standard Gaussian mechanism to each subspace and upper bound the total privacy loss via composition, then the variance of noise required is ${\Theta}\big(m(\sum_{j=1}^m c^2_{2j}r_j)\big)$. 
Let $\bm{c}_2=(c_{21}, \cdots, c_{2m})$, when $r_1 = \cdots = r_m = d/m$ identically, the variance bound from composition is $\Theta(d\|\bm{c}_2\|_2^2)$, while Theorem \ref{thm: hybrid clipping} improves it to $\Theta({d\|\bm{c}_2\|_1^2/ m})$, reducing by a factor of $\Theta(m{\|\bm{c}_2\|_2^2/\|\bm{c}_2\|_1^2})$.
In practical applications, such as the examples in Section \ref{sec: warmup}, the dimension of the residual component space could be large.
Suppose $r_1 = \cdots = r_{m - 1} = r$ for some constant $r$ and $r_{m} = d - (m - 1)r$, and $c_{21} = \dotsb = c_{2m}$ for some constant $c$.
The variance bound from composition is $\Theta(mdc^2)$, while the bound from Theorem \ref{thm: hybrid clipping} is $\Theta(c^2(\sqrt{d - (m - 1)r} + (m - 1)\sqrt{r})^2)$, which is $\Theta(dc^2)$ when $m= o(d)$. Theorem \ref{thm: hybrid clipping} thus reduces the noise variance bound by a factor $m$. As a final remark, the results in Theorem \ref{thm: hybrid clipping} can be generalized to arbitrary hybrid clipping once  the projection of the sensitivity set $\mathsf{S}$ in each subspace satisfies the symmetry property in Definition \ref{def: sign invariance} and Definition \ref{def: permutation invariance}. We assume $l_2$-norm clipping in Theorem \ref{thm: hybrid clipping} for presentation simplicity.    
\vspace{-0.05 in}
\section{Twice Sampling}
\label{sec: twice-sampling}
So far, we have solved the first half of the problem where we showed carefully-constructed optimal Gaussian noise {\em can} reflect the desired asymmetric high-dimensional geometry where the sensitivity magnitude varies in different subspaces. 
However, as Corollary \ref{cor: l_p mix} suggests, in RDP with the pure Gaussian mechanism, the isotropic Gaussian is already optimal for any mixture of $l_p$-norm clippings, and the noise scale is only determined by the maximal $l_2$-norm of the elements in the sensitivity set. 
Thus, we cannot expect tighter privacy analysis to capture additional, non-trivial $l_{\infty}$-norm restrictions (without decreasing the worst-case $l_2$-norm), unless the randomization is not (purely) Gaussian noise. 
We are then motivated to consider whether it is possible that, provided extra analyzable randomization beyond only Gaussian noise, $l_{\infty}$-norm geometry can be properly reflected. 
{In particular, given that the $l_{\infty}$-norm is a coordinate-wise property, can independent sampling across coordinates match the $l_{\infty}$-norm geometry and improve Gaussian noise? We will answer the above questions affirmatively with a carefully-designed sampling strategy, termed twice sampling, to address both the privacy and efficiency challenges.}    

\vspace{- 0.1 in}
\subsection{Coordinate-Wise Poisson Sampling}
Though sampling seems a promising way to introduce fresh randomness, we notice that a straightforward application of standard privacy amplification results still cannot enjoy the gain from the additional $l_{\infty}$-norm constraint. 
Classic amplification results basically state that for any mechanism $\mathcal{M}$ 
satisfying $(\epsilon, \delta)$-DP or $(\alpha, \epsilon(\alpha))$-RDP, $\mathcal{M}$ on $q$-Poisson sampled data satisfies $O(q\epsilon, q\delta)$-DP \cite{li2012sampling} or $O(\alpha, q^2\alpha(e^{\epsilon(2)}-1))$-RDP \cite{zhu2019poission}. 
As Corollary \ref{cor: l_p mix} already gives a negative answer to improving the privacy analysis for the Gaussian mechanism before input sampling,
we cannot obtain a better privacy bound compared to existing works on input-wise subsampled Gaussian \cite{mironov2019r,zhu2019poission}. 
To this end, instead of sampling on the input data dimension, we consider coordinate-wise sampling to exploit the $l_{\infty}$ restriction, a per-coordinate property, and formally describe it as Algorithm \ref{alg: coordinate sampling}.

{At a high level, Algorithm \ref{alg: coordinate sampling} samples independently for each coordinate of a given function $\mathcal{F}$'s outputs.}
We can imagine a matrix $Y \in \mathbb{R}^{n \times d}$, where each row corresponds to a clipped processing on an individual datapoint $\mathcal{CP}(\mathcal{F}(x_i))$. Standard mean estimation based on input-level sampling basically works as sampling a subset of the rows of and returning its empirical mean \cite{mironov2019r}. As a comparison, coordinate-wise sampling independently selects elements in each column of the matrix and returns their empirical mean as an estimation. It is not hard to verify that Algorithm \ref{alg: coordinate sampling} produces an unbiased estimation and the variance of estimation error is the same as that of $q$-input-wise sampling, since they share exactly the same marginal distribution in each coordinate. A formal statement is given as follows.
\vspace{- 0.05 in}
\begin{proposition}[Unbiasedness and Estimation Variance] For an arbitrary processing function $\mathcal{F}$ and an input set $X= \{x_1, x_2, \cdots x_n\}$, let $Y=\{ y_i = \mathcal{CP}(\mathcal{F}(x_i)), i=1,2,\cdots, n\}$ and $\mu= \frac{1}{n} \cdot \sum_{i=1}^n y_i$. Suppose $\bm{o}'$ is the aggregation of a subset of $Y$ generated by $q$-input-wise  Poisson sampling and $\bm{o}$ is the output of Algorithm \ref{alg: coordinate sampling} with $q$-coordinate-wise sampling, then
$\mathbb{E}[\frac{\bm{o}}{nq}]=\mathbb{E}[\frac{\bm{o'}}{nq}]=\mu$ and $\mathbb{E}[\|\frac{\bm{o}}{nq}-\mu\|^2] = \mathbb{E}[\|\frac{\bm{o}'}{nq}-\mu\|^2].$
\label{prop: unbiasedness} 
\end{proposition}
\vspace{- 0.05 in}
In the following theorem, we show coordinate-wise sampling does reflect the additional $l_{\infty}$-norm restriction.
\vspace{- 0.05 in}
\begin{restatable}[Privacy Amplification of Coordinate-Wise Sampling]{theorem}{PriAmpCoor} 
In Algorithm \ref{alg: coordinate sampling}, if we select a mixture of $l_{\infty}$-norm clipping and $l_p$-norm for some $p \in (0,2]$, with parameters $c_{\infty}$ and $c_{p}$, respectively, where $d_0\cdot c_{\infty}^p =(c_p)^p$ and $d_0 \leq d$, then the dominating sensitivity of Algorithm \ref{alg: coordinate sampling} is in a form $$(\underbrace{c_{\infty}, \cdots, c_{\infty}}_{d_0},\underbrace{0,\cdots,0}_{d-d_0}).$$
Moreover, the $(\alpha,\epsilon(\alpha))$-RDP of Algorithm \ref{alg: coordinate sampling} has a closed form, where
\begin{equation}
\vspace{-0.1 in}
\epsilon(\alpha) = \frac{d_0}{\alpha-1}\cdot \log\big((1-q)^{\alpha} + \sum_{v=1}^{\alpha}\binom{\alpha}{v}(1-q)^{\alpha-v}q^{v}e^{\frac{v(v-1)c^2_{\infty}}{2\sigma^2}} \big).
\vspace{-0.05 in}
\label{coordinate-sample-bound}
\end{equation}
\label{thm: privacy coordiante sampling}
\end{restatable}
\vspace{- 0.05 in}
\begin{proof}
\vspace{-0.05 in}
    See Appendix \ref{app: thm: privacy coordiante sampling}. 
\end{proof}
\vspace{- 0.05 in}
\begin{remark}
By the proof of Theorem \ref{thm: privacy coordiante sampling}, for general sensitivity set $\mathsf{S}$ of a function $\mathcal{F}(\cdot)$, if $\mathsf{S}$ is a convex set, then under $q$-coordinate-wise Poisson sampling and the Gaussian mechanism, for RDP the dominating sensitivity must be on the boundary of $\mathsf{S}$. In particular, if $\mathsf{S}$ is a polytope, then the dominating sensitivity must be within the vertices of $\mathsf{S}$. Thus, in Theorem \ref{thm: privacy coordiante sampling}, $d_0\cdot c_{\infty}^p =(c_p)^p$ is not necessary and can be relaxed to $d_0\cdot c_{\infty}^p + (c')^p =(c_p)^p$, for some $c' \in (0,c_{\infty})$. In this case, for RDP, the dominating sensitivity is instead in a form $( c_{\infty}, \cdots, c_{\infty} ,c',0,\cdots,0)$.
\vspace{-0.1 in}
\end{remark}
Theorem \ref{thm: privacy coordiante sampling} characterizes the form of dominating sensitivity in a sensitivity set  $\mathsf{S} = \{\bm{s} | \|\bm{s}\|_p \leq c_p, \|\bm{s}\|_{\infty} \leq c_{\infty} \}$ in Algorithm \ref{alg: coordinate sampling}, when we use a mixture of $l_p$-norm for $p \in (0,2]$ and $l_{\infty}$ clipping on the processing function $\mathcal{F}$. 
Roughly speaking, to achieve the worst-case divergence, the adversary will concentrate their sensitivity budget on a few coordinates and maximize the magnitude of each. 
When $d_0 \to 1$, $c_{\infty} \to c_p$ and the $l_{\infty}$ restriction becomes weaker,  (\ref{coordinate-sample-bound}) reduces to the regular input-wise subsampled Gaussian mechanism with sampling rate $q$ \cite{mironov2019r}. 
The dominating sensitivity turns out to be the one-hot vector. 
In other words, without the $l_{\infty}$-norm restriction, the privacy amplification by coordinate-wise sampling is the same as that of input-wise sampling with the same rate $q$.  

However, given the additional $l_{\infty}$-norm restriction, for example, if we select $p=2$, Theorem \ref{thm: privacy coordiante sampling} suggests that not all of the elements on the $l_2$-norm ball sphere in $\mathsf{S}$ behave as the dominating sensitivity, which is the key to render a sharpened privacy guarantee. This does not contradict our previous analysis on the pure Gaussian mechanism, since in Algorithm \ref{alg: coordinate sampling} the distribution of each output coordinate is an independent Gaussian mixture rather than pure Gaussian. In the following theorem, we present 
a formal quantification on the asymptotic improvement of the privacy analysis by coordinate-wise sampling. 
\vspace{- 0.05 in}
\begin{restatable}[Asymptotic Privacy Improvement through Coordinate-wise Sampling]{theorem}{AsympPrivCoor} 
Under the same setup as Theorem \ref{thm: privacy coordiante sampling}, let $\tau = (c_p/\sqrt{2}\sigma)^2$, then for any sampling rate $q$, when $d_0$ is sufficiently large, the $(\alpha, \epsilon(\alpha))$-RDP bound (\ref{coordinate-sample-bound}) of Algorithm \ref{alg: coordinate sampling} converges to $\epsilon(\alpha) = \alpha q^2\tau.$ 
As a comparison, when $d_0=1$, (\ref{coordinate-sample-bound}) is equivalent to input-wise subsampled Gaussian mechanism with rate $q$. In the regime where $q$ is small such that $q < 1/(2\alpha e^{\alpha\tau})$, then $\epsilon(\alpha) = \Theta\big(\alpha q^2 (e^{\tau}-1)\big)$; when $\alpha(\alpha-1)\tau \geq 2$ and $q$ is relatively large such that $q \geq 1/(e^{(\alpha-1)\tau/2})$, then $\epsilon(\alpha) = \Omega(\alpha\tau)$. 
\label{thm: coordinate-wise-sampling improvement}
\end{restatable}
\begin{proof}
\vspace{- 0.05 in}
See Appendix \ref{app: thm: coordinate-wise-sampling improvement}.
\vspace{- 0.05 in}
\end{proof}

Theorem \ref{thm: coordinate-wise-sampling improvement} provides a tight quantification on the improvement through coordinate-wise sampling when we are allowed to assume sufficiently small, but non-trivial, $l_{\infty}$-norm restriction. In Theorem \ref{thm: coordinate-wise-sampling improvement}, $\tau$ is essentially the $\epsilon(2)$ of the $(2,\epsilon(2))$-RDP bound for the pure Gaussian mechanism. Thus, for a standard $q$ input-wise subsampled Gaussian mechanism, when $q$ is sufficiently small, though the produced $\epsilon$ is still $\Theta(q^2)$, Algorithm \ref{alg: coordinate sampling} improves the factor from $\Theta(e^{\tau}-1)$ to $\tau$, which will be helpful given small noise and large $\tau$.

Theorem \ref{thm: coordinate-wise-sampling improvement} also characterizes an important phenomenon for subsampled Gaussians that, when $\alpha$ is relatively large, the effect of privacy amplification will diminish, where the security parameter $\epsilon(\alpha)$ is independent of $q$. This is also observed in previous works \cite{zhu2019poission, wang2019subsampled}. This means that $\epsilon(\alpha)$ in RDP with larger $\alpha$ will not benefit from sampling and we may not use larger $\alpha$ to obtain tighter composition results in practice. More details can be found in Fig. \ref{Fig. amplification RDP} in Section \ref{sec: simulation sampling}. This is also one of the primary reasons why smaller sampling rate produces worse SNR when provided small noise. As a comparison, coordinate-wise sampling can always provide a tight bound $\epsilon(\alpha) = \alpha q^2\tau$ to enjoy the $q^2$ amplification in any setups with assistance of small enough $l_{\infty}$-norm restriction. This brings asymptotic improvement on the converted $(\epsilon, \delta)$ guarantee produced and could significantly narrow the performance gap using smaller sampling rate as shown later in Section \ref{sec: simulation sampling} and Section \ref{sec: additional exp}.

\begin{algorithm}[t]
\caption{Private Aggregation with $q$-Coordinate-Wise Poisson (i.i.d.) Sampling}
\begin{algorithmic}[1]
\STATE \textbf{Input:} A processing function $\mathcal{F}(\cdot): \mathcal{X}^* \to \mathbb{R}^d$, a sensitive input set $X = \big\{x_1, x_2, ... ,x_n\big\} \in \mathcal{X}^n$ of $n$ datapoints, Poisson sampling rate $q$, clipping operator $\mathcal{CP}$ and Gaussian noise variance parameter $\sigma$.
\STATE Apply $\mathcal{F}$ on each $x_i$ for $i=1,2,\cdots, n$, and apply $\mathcal{CP}$ on $\mathcal{F}(x_i)$, suppose $y_i = \mathcal{CP}(\mathcal{F}(x_i))$. 
\FOR{$l=1,2,...,d$}
   \STATE Apply independent Poisson sampling with rate $q$ on index $[1:n]$ and obtain an index set $\mathcal{I}_l = \big\{[1_l], [2_l], \cdots, [B_l]\big\}$, where $[\cdot]$ represents some permutation.
   \STATE Compute the aggregation of the $l$-th coordinate of selected indexes $\sum_{i \in \mathcal{I}_l} y_i(l)$, and independently generate a noise $e_l \sim \mathcal{N}(0,\sigma^2)$.
   \STATE $o_l = \sum_{i \in \mathcal{I}^{(l)}} y_i(l) + \mathcal{N}(0,\sigma^2)$. 
\ENDFOR 
\STATE \textbf{Output}: $\bm{o} = (o_1, o_2, \cdots, o_d)$. 
\end{algorithmic}
\label{alg: coordinate sampling}
\end{algorithm}

Before the end of this section, we have a final remark on the generalization of Algorithm \ref{alg: coordinate sampling}, where coordinate-wise sampling with enhanced privacy can be applied to more generic processing beyond aggregation. In general, for an arbitrary function $\mathcal{F}$ and a dataset $X$, for the $l$-th coordinate estimation, we may randomly sample a subset $J_l$ from $X$ and take $\mathcal{F}(J_l)(l)$ as the output. The results in Proposition \ref{prop: unbiasedness} and Theorem \ref{thm: privacy coordiante sampling} also apply to such a scenario if one can ensure the following sensitivity guarantee: for any selection of $\bar{J}=(J_1, J_2, \cdots, J_d)$ from arbitrary $X$ and a differing datapoint $x$, the $d$-dimensional vector $\big(\mathcal{F}(J_1)(1)-\mathcal{F}(J_1\cup x)(1), \cdots, \mathcal{F}(J_d)-\mathcal{F}(J_d\cup x)\big)$ is within the intersection between an $l_p$-norm ($p\in(0,2])$ and an $l_{\infty}$-norm ball.

\subsection{Twice Sampling Algorithm}
\label{subsec: twice sampling}

As shown in the previous section, coordinate-wise Poisson sampling could enable the Gaussian mechanism to benefit from an additional $l_{\infty}$-norm sensitivity restriction, which is usually free for high-dimensional tasks. 
However, it could be inefficient since, in general, preprocessing is required to compute $\mathcal{F}(x_i)$ for each datapoint $x_i$  before the sampling.
Though in some applications, for example, mean estimation on a given dataset or when the computation of each coordinate of $\mathcal{F}(x_i)$ is independent, such preprocessing is not necessary and Algorithm \ref{alg: coordinate sampling} can still be implemented in $O(ndq)$ time.  
However, deep learning is a negative example, where the gradient computation of a neural network requires back-propagation \cite{cilimkovic2015neural}. Before we can evaluate a single coordinate of a gradient, we need to first calculate the gradients of other parameters in latter layers in a sequential manner, which requires $O(nqd^2)$ time without preprocessing. Thus, in general, Algorithm \ref{alg: coordinate sampling} would take $O(\max\{nqd^2,nd\})$ time, rather than the ideal $O(nqd)$ complexity if we adopt a standard $q$-input-wise sampling. To tackle this efficiency challenge, we propose an alternative method termed {\em twice sampling} formally presented in Algorithm \ref{alg: sampling twice}.  

Twice sampling is a neat composition of input-wise sampling and coordinate-wise sampling. Instead of applying Algorithm \ref{alg: coordinate sampling} on the entire data, we will first apply $q_1$-Poisson sampling on the dataset to generate a subset, of expected size $nq_1$, as the input to Algorithm \ref{alg: coordinate sampling}. Thus, from an efficiency perspective, twice sampling only takes $O(nq_1d)$ time at most as we only need to preprocess the subset of samples generated from the first round $q_1$-sampling. It is not hard to verify that, for the marginal distribution of each coordinate, each sample will be selected with probability $q_1q_2$. However, we must stress that twice sampling is {\em not} equivalent to an independent coordinate-wise sampling with parameter $q_1q_2$, since now different coordinates become correlated; or they are only independent conditional on the selection from the first round input-wise sampling. Thus, the privacy analysis of twice sampling is non-trivial and more complicated compared to that of Algorithm \ref{alg: coordinate sampling}. After a careful study on the Rényi divergence between Gaussian mixture models, we derive a closed-form RDP bound of Algorithm \ref{alg: sampling twice}, summarized as the following theorem.  
\vspace{- 0.05 in}
\begin{restatable}[Privacy Amplification from Twice Sampling]{theorem}{PrivAmplTS}
If a data processing function $\mathcal{F}$ is perturbed with the Gaussian mechanism, and under $q_2$-coordinate-wise Poisson sampling, it satisfies $\big(\alpha_0, \epsilon_0(\alpha_0)\big)$-RDP for $\alpha_0 = 2,3, \cdots$, then with twice sampling of parameters $(q_1, q_2)$, it satisfies $(\alpha, \epsilon(\alpha))$-RDP, where $\epsilon(\alpha)$ is in a form 
{
\begin{equation}
\epsilon(\alpha) = \frac{\log\big((1-q_1)^{\alpha}+ \sum_{v=1}^{\alpha}\binom{\alpha}{v}(1-q_1)^{\alpha-v}q^{v}_1e^{(v-1)\epsilon_0(v)} \big)}{\alpha-1}. 
\label{twice-sampling-bound}
\end{equation}
}
\label{thm: sampling twice}
\end{restatable}
\begin{proof}
\vspace{-0.1 in}
    See Appendix \ref{app: thm: sampling twice}. 
\vspace{-0.1 in}
\end{proof}

{In Theorem {\ref{thm: sampling twice}}, the two sampling parameters $(q_1, q_2)$ both affect the final privacy bound {(\ref{twice-sampling-bound})} produced, where, on one hand, {(\ref{twice-sampling-bound})} straightforwardly decreases as the input-wise sampling rate $q_1$ decreases; meanwhile, the coordinate-wise sampling rate $q_2$ also has an implicit influence on {(\ref{twice-sampling-bound})}, captured by $\epsilon_0(v)$. It is noted that we assume the objective processing function $\mathcal{F}$ satisfies $(\alpha_0, \epsilon_0(\alpha_0))$-RDP with $q_2$-coordinate-wise sampling only. The $(\alpha_0, \epsilon_0(\alpha_0))$-RDP is captured by Theorem {\ref{thm: privacy coordiante sampling}} and, in the same setup, a smaller $q_2$ also leads to a smaller $\epsilon_0(\alpha_0)$, and thus smaller $\epsilon(\alpha)$ in ({\ref{twice-sampling-bound}}).} 

From Theorem \ref{thm: sampling twice}, there are two steps to calculate concrete RDP parameters of twice sampling. First, via (\ref{coordinate-sample-bound}) in Theorem \ref{thm: privacy coordiante sampling}, we can determine the RDP parameters  $(\alpha_0, \epsilon(\alpha_0))$  given different $\alpha_0$ for a coordinate-wise sampling of rate $q_2$. Second, plugging those numbers into (\ref{twice-sampling-bound}) in Theorem \ref{thm: sampling twice} will produce the final RDP bound. Comparing Theorem \ref{thm: privacy coordiante sampling} with Theorem \ref{thm: sampling twice}, we have several comments. On one hand, twice sampling is a tradeoff between efficiency and privacy. As mentioned before, though marginally the distribution of each coordinate of $(q_1, q_2)$-twice sampling is equivalent to that of $q_1q_2$ coordinate-wise sampling, the privacy enhancement of twice sampling is weaker. In (\ref{coordinate-sample-bound}), if we select $q=q_1q_2$, we will obtain a smaller $\epsilon(\alpha)$ bound compared to that in (\ref{twice-sampling-bound}). This is not surprising, as in twice sampling we put more restriction on the sampling randomness, where coordinate-wise sampling is only implemented on subsampled data rather than the entire set, and thus less privacy amplification is produced. But it is worthwhile to note that for both cases, Algorithm \ref{alg: coordinate sampling} and Algorithm \ref{alg: sampling twice} indeed asymptotically achieve the same amplification. We present a formal statement as follows.

\begin{restatable}[Asymptotic Privacy Improvement from Twice Sampling]{corollary}{AsympPrivTSim}
\label{cor: twice-sampling improvement}
For $(q_1, q_2)$-twice sampling Algorithm \ref{alg: sampling twice} with a mixture of $l_p$-norm $(p \in (0,2])$ and $l_{\infty}$-norm clipping as described in Theorem \ref{thm: coordinate-wise-sampling improvement}, if $d_0$ is sufficiently large, then Algorithm \ref{alg: sampling twice} satisfies $(\alpha, \epsilon(\alpha))$-RDP where $\epsilon(\alpha) $ equals
$${\frac{\log\big((1-q_1)^{\alpha}+ \sum_{v=1}^{\alpha}\binom{\alpha}{v}(1-q_1)^{\alpha-v}q^{v}_1e^{v(v-1)q^2_2\tau} \big)}{\alpha-1}},$$
where $\tau = (c_p/\sqrt{2}\sigma)^2$. When $q_1$ and $q_2$ are sufficiently small, $\epsilon(\alpha) = O(\alpha(q_1q_2)^2\tau)$.
\end{restatable}
\vspace{-0.1 in}
\begin{proof}
    See Appendix \ref{app: cor: twice-sampling improvement}. 
\vspace{-0.05 in}
\end{proof}

Thus, in general, a smaller $q_2$ selected will narrow down the gap between the privacy guarantee from coordinate-wise sampling with rate $q=q_1q_2$ and twice sampling with $(q_1, q_2)$. As a tradeoff, for a fixed $q$, a smaller $q_2$ implies a larger $q_1$ and the time complexity $O(nq_1d)$ of twice sampling increases. Fortunately, as shown below, in practice, the gap is not big where we only need to pay a small overhead for a significant privacy enhancement. 

\begin{algorithm}[t]
\caption{$(q_1, q_2)$-Twice Poisson Sampling}
\begin{algorithmic}[1]
\STATE \textbf{Input:} A processing function $\mathcal{F}(\cdot): \mathcal{X}^* \to \mathbb{R}^d$, sensitive input set $X = \big\{x_1, x_2, ... ,x_n\big\} \in \mathcal{X}^n$ of $n$ datapoints, Poisson sampling rates $q_1$ and $q_2$.
\STATE Apply Poisson sampling with rate $q_1$ on index $[1:n]$ and obtain an index set $\mathcal{I} = \big\{[1], [2], \cdots, [B]\big\}$ of size $B$. Let $X_{\mathcal{I}} = \{x_{[1]}, \cdots, x_{[B]} \}$.
\STATE Apply Algorithm \ref{alg: coordinate sampling} with coordinate-wise Poisson sampling with parameter $q_2$ on $\mathcal{F}$ with respect to $X_{\mathcal{I}}$.
\end{algorithmic}
\label{alg: sampling twice}
\end{algorithm}

\begin{figure*}[t]
  \centering
  \includegraphics[width=0.33\linewidth]{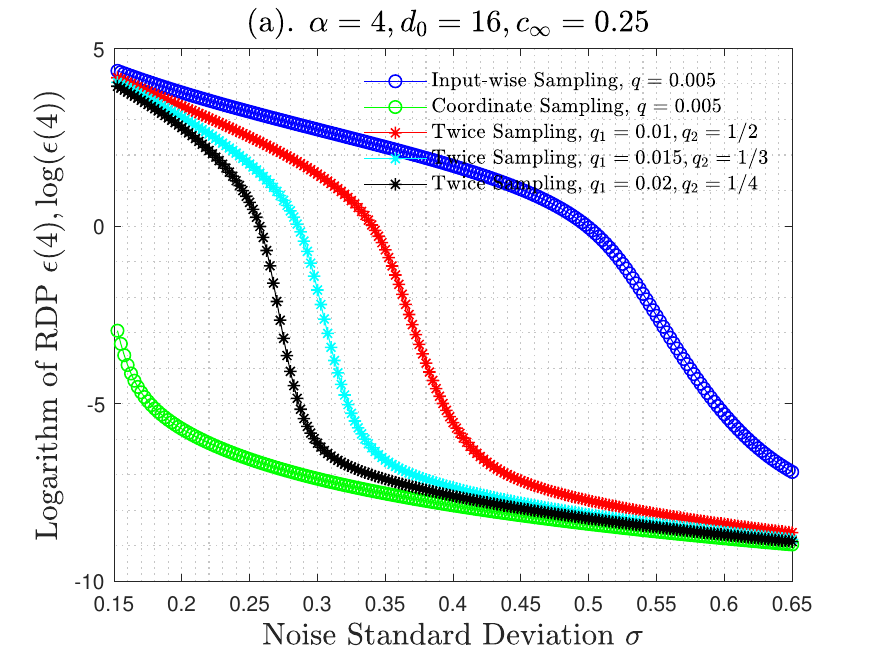}
  \includegraphics[width=0.33\linewidth]{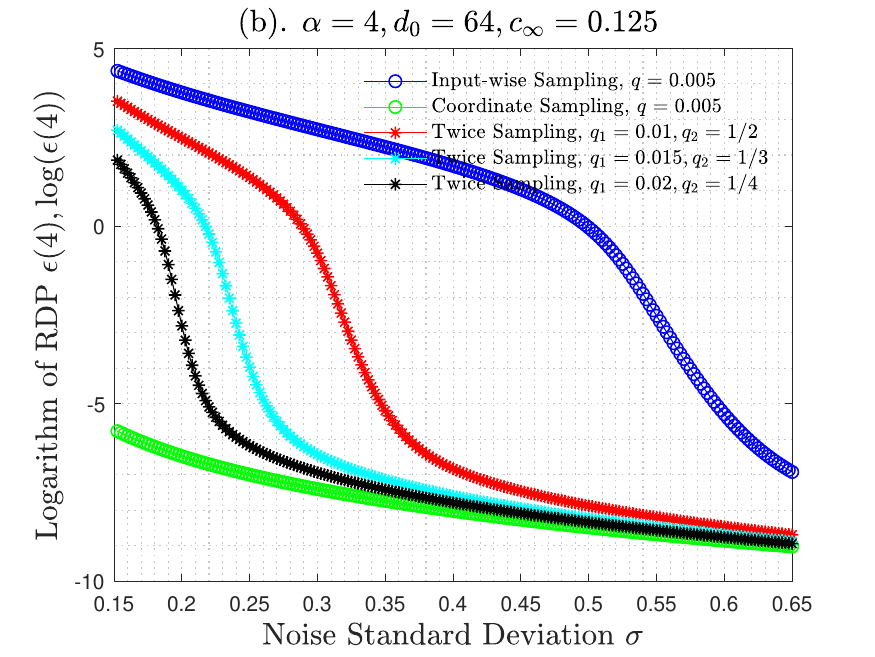}
  \includegraphics[width=0.33\linewidth]{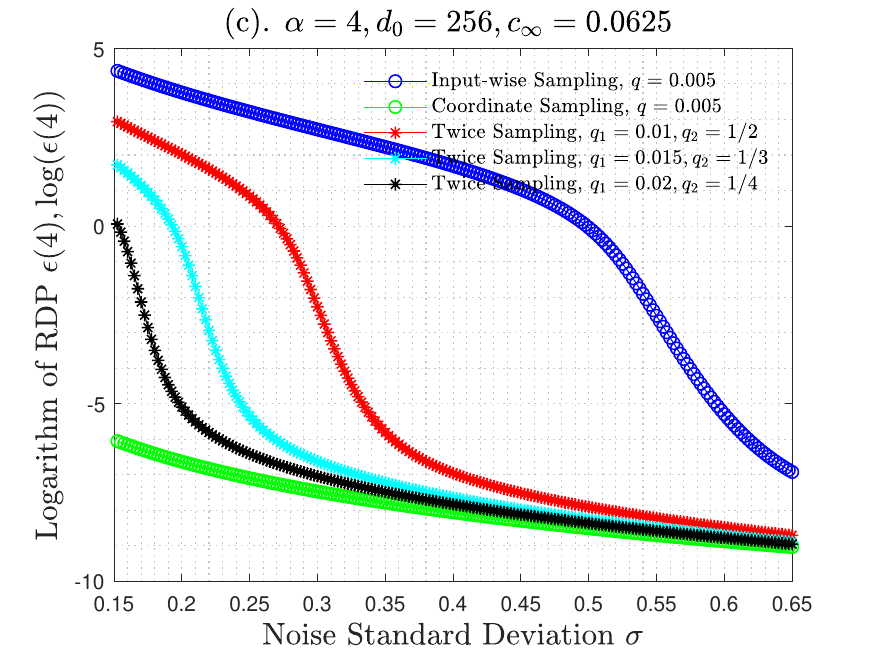} 

  \includegraphics[width=0.33\linewidth]{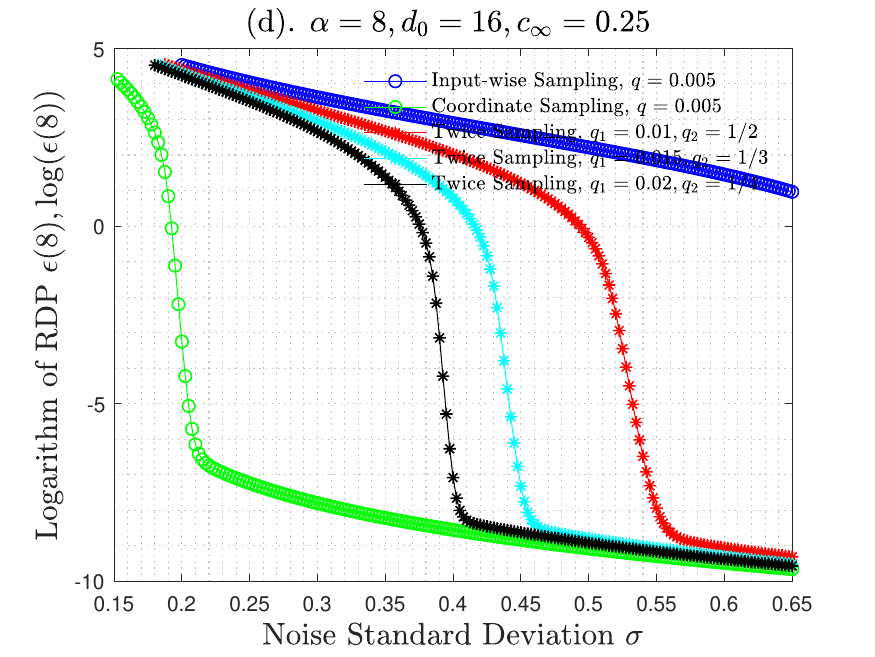}
  \includegraphics[width=0.33\linewidth]{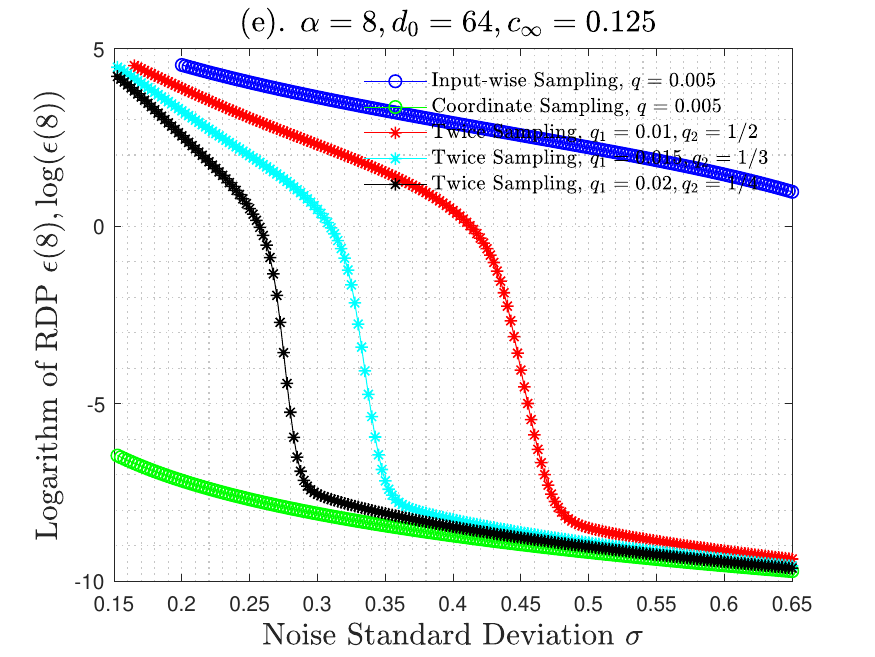}
  \includegraphics[width=0.33\linewidth]{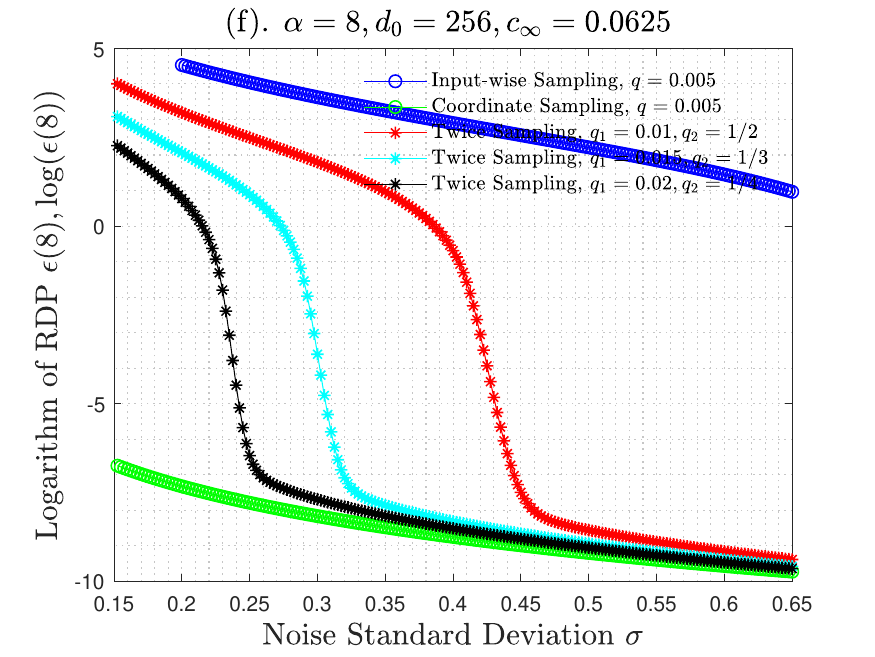} 
  
\caption{Comparison between the Logarithm $\log(\epsilon(\alpha))$ Produced by Input-wise Sampling $q=0.005$, Coordinate-Wise Sampling $q=0.005$ and Twice Sampling $q=q_1q_2=0.005$. $l_2$-norm clipping threshold $c_2 =1$}

\label{Fig. amplification RDP}.
\end{figure*}  

\subsection{Simulation on Privacy Amplification}
\label{sec: simulation sampling}
In the following, we provide simulations on the privacy guarantees produced by an input-wise sampling with $q=0.005$, coordinate-wise sampling with $q=0.005$ and twice sampling $q=q_1q_2=0.005$ for $q_2 \in \{1/2, 1/3, 1/4\}$, combined with the Gaussian mechanism. 
In Fig. \ref{Fig. amplification RDP}, we consider a mixture clipping of $l_2$-norm with fixed $c_2=1$ and $l_{\infty}$-norm with parameter  $c_{\infty}=1/\sqrt{d_0}$ for $d_0$ varying from $\{16,64,256\}$. 
In each subfigure of Fig \ref{Fig. amplification RDP}, the x-axis is the standard deviation $\sigma$ of the injected Gaussian noise. 
The y-axis shows $\log(\epsilon(\alpha))$. In Fig. \ref{Fig. amplification RDP} (a-c), we select $\alpha=4$ while in Fig. \ref{Fig. amplification RDP} (d-f), where we select $\alpha=8$. 
We have the following important observation which also supports our theory in Theorems \ref{thm: privacy coordiante sampling}-\ref{thm: sampling twice}. 

First, it is noted that there is a {\em critical point}: when $\sigma$ is smaller than some critical point, $\epsilon(\alpha)$ is close to $\alpha/(2\sigma^2)$, the $\alpha$-th order of RDP of the pure Gaussian mechanism, independent of $q$, as analyzed in Theorem \ref{thm: coordinate-wise-sampling improvement}; when $\sigma$ is larger than this critical point, $\epsilon(\alpha)$ quickly converges to $\Theta(q^2)$. This is true for any of the sampling methods. But it should be noted that, given larger $d_0$ and consequently smaller $c_{\infty}$, or smaller $\alpha$, this critical point will also be smaller. One can compare the lines of the same color in Fig. \ref{Fig. amplification RDP}. This matches the results of Theorem \ref{thm: coordinate-wise-sampling improvement} and when $d_0 \to \infty$, the critical point  will approach $0$, and twice (coordinate-wise) sampling can always enjoy $\Theta(q^2)$ amplification.

Second, once twice sampling has passed the turning point, the difference between the privacy bound of coordinate-wise sampling and twice sampling is much smaller and also more insensitive to the selection of $d_0$ and $q_2$. 
Indeed, to produce practical security parameters, the difference is almost negligible when $d_0 \geq 50$ and $q_2 \leq 0.5$. 
When $q_2=0.5$, we only need to double the number of subsampled data to preprocess $2q_1n$ samples in expectation.

To provide more intuition about the improvement, in Fig. \ref{Fig. amplification eps-delta DP}, we convert the RDP bound under $T=10,000$ composition to $(\epsilon, \delta)$ using Lemma \ref{lemma: composition-RDP}. 
We select $\delta=10^{-5}$ and the y-axis of Fig. \ref{Fig. amplification eps-delta DP} is $\log(\epsilon)$ rather than $\epsilon$ to illustrate an asymptotic improvement on the exponent. This captures the scenario where we apply a DP-SGD of sampling rate $q=0.005$ for $T=10,000$ iterations.  
From Fig. \ref{Fig. amplification eps-delta DP}(b), we achieve $\epsilon=8$ by applying twice sampling with rates $(q_1=0.01,q_2=1/2)$ and $(q_1=0.015,q_2=1/3)$.
While under $q=0.005$, input-wise sampling can only provide a non-meaningful/weak guarantee $\epsilon=30.4$ and $\epsilon=88.5$, respectively. 
In general, such improvement will be more significant as the sampling rate $q$ gets smaller. In both Algorithm \ref{alg: coordinate sampling} or Algorithm \ref{alg: sampling twice}, when $q \to 1$, the effect of sampling diminishes and the corresponding RDP analysis is closer to the case studied in previous section with the pure Gaussian mechanism, restricted by our negative results on possible improvement. In contrast, when $q \to 0$, the corresponding RDP analysis is closer to a sum of divergences between Gaussian mixture models that reflects the coordinate-wise restriction.

\subsection{Hybrid Clipping and Twice Sampling} 
\label{sec: noise optimization}
In this section, we combine the results in Sections \ref{sec: hypercube} and \ref{subsec: twice sampling}, and describe a hybrid clipping with twice sampling. 
For the hybrid clipping side, given a $d$-dimensional vector $\bm{v} \in \mathbb{R}^d$ and $d$ orthogonal unit bases in $m$ subsets, where each subset is a form $U_j = \{\bm{u}_{j1}, \cdots, \bm{u}_{jr_j}\}$ for $j=1,2,\cdots, r_j$, and $\sum_{j=1}^m r_j =d$, we clip the projection of $\bm{v}$ in each subspace separately. 
Suppose that the expression of $\bm{v}$ under the selected basis is $\bm{v} = \sum_{j=1}^m\sum_{l=1}^{r_j} v_{jl}\bm{u}_{jl}$. 
For each $j \in \{1,2,\cdots, m\}$, we first apply clipping with parameter $c_j$ on the projection of $\bm{v}$ in the $j$-th subspace, i.e., $\sum_{l=1}^{r_j} v_{jl}\bm{u}_{jl}$. Thus, the clipped $\bm{v}$ can then be expressed as 
$\tilde{\bm{v}} = $ $\sum_{j=1}^m \mathcal{CP}(\sum_{l=1}^{r_j}v_{jl}\bm{u}_{jl}, c_j) = \sum_{j=1}^m\sum_{l=1}^{r_j} \tilde{v}_{jl}\bm{u}_{jl}.$ Next, we apply an additional $l_{\infty}$-norm clipping on each coordinate $\tilde{v}_{jl}$. By selecting such clipping in Algorithm \ref{alg: sampling twice}, the only difference compared to regular twice sampling is that we now implement the coordinate-wise sampling with respect to the coordinate in the expression under the given bases $\bm{u}_{jl}$ rather than the natural one-hot bases. However, this does not change the privacy analysis. 
One may imagine that we apply a uniform transform on the processing data by the unitary matrix determined by $\{\bm{u}_{jl}\}$ at the beginning, and it becomes equivalent to conducting the hybrid clipping and twice sampling on the natural bases. We are now able to enjoy the sharpened privacy analysis from both Theorem \ref{thm: hybrid clipping} and Theorem \ref{thm: sampling twice}. The remaining problem is to optimize the noise variance. For example, we can take $l_2$-norm clipping as the building block $\mathcal{CP}$. 
We conclude with the following theorem.

\begin{restatable}[Noise Optimization for Hybrid Clipping and Twice Sampling]{theorem}{FinalTheorem}
Given $m$ sets of orthogonal unit bases $U_j = \{\bm{u}_{j1}, \cdots, \bm{u}_{jr_j}\}$ where $\sum_{j=1}^m r_j=d$, we consider a clipping strategy where for any $\bm{v} \in \mathbb{R}^d$, we clip its projection in the $j$-th subspace spanned by $U_j$ in $l_2$-norm and $l_{\infty}$-norm with parameters $c_{2j}$ and $c_{\infty j}$, respectively, such that $c_{2j} = \sqrt{d_{0j}}c_{\infty j}$, for $j=1,2,\cdots,m$. 
Here, $d_{0j} \le r_j$ are the infinity norm parameters. We inject an isotropic Gaussian noise $\mathcal{N}(\bm{0}, \sigma^2_j\cdot \bm{I}_{r_j})$ embedded into the $j$-th subspace. Then, the following optimized noise parameter can ensure $(\alpha, \epsilon(\alpha))$-RDP of the $(q_1, q_2)$-twice sampling with the above clipping strategy,
\begin{equation}
\vspace{-0.1 in}
\begin{aligned}
& {\bm{\sigma}=(\sigma_1, \cdots, \sigma_m)}= \arg \min_{\bm{\sigma}} \sum_{j=1}^m r_j\sigma^2_j\\
& \text{s.t.} ~ \frac{\log\big((1-q_1)^{\alpha}+ \sum_{v=1}^{\alpha}\binom{\alpha}{v}(1-q_1)^{\alpha-l}q^{v}_1e^{(v-1)\epsilon_0(v)} \big)}{\alpha-1} \leq \epsilon(\alpha),
\end{aligned}
\label{noise optimization form}
\end{equation}
where $\epsilon_0(v)$ for $v = 2,3,\cdots$, is in a form,
\begin{equation*}
\vspace{-0.1 in}
\sum_{j=1}^m \frac{d_{0j}}{v-1}\log\big((1-q_2)^{v} + \sum_{l=1}^{v}\binom{v}{l}(1-q_2)^{v-l}q^{l}_2e^{\frac{l(l-1)c^2_{\infty j}}{2\sigma^2_j}} \big).
\vspace{-0.05 in}
\end{equation*}
\label{cor: optimized noise hybrid+twice}
\end{restatable}

Theorem \ref{cor: optimized noise hybrid+twice} then allows us to numerically solve (\ref{noise optimization form}) and obtain the optimized noise variance under the privacy constraint. We leave a closed-form (approximated) solution to (\ref{noise optimization form}) as an open problem.  

\begin{figure*}[t]
  \centering
  \includegraphics[width=0.33\linewidth]{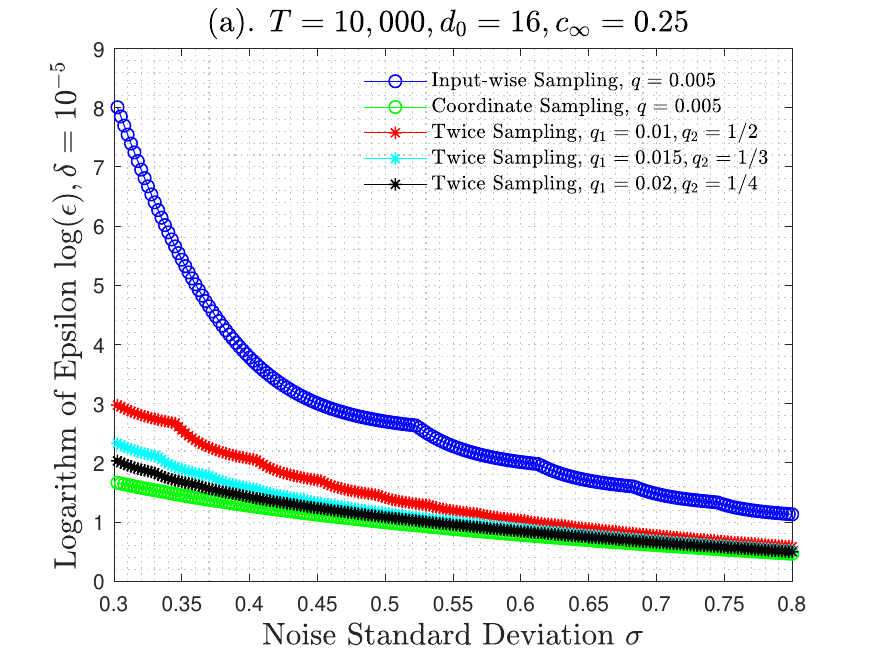}
  \includegraphics[width=0.33\linewidth]{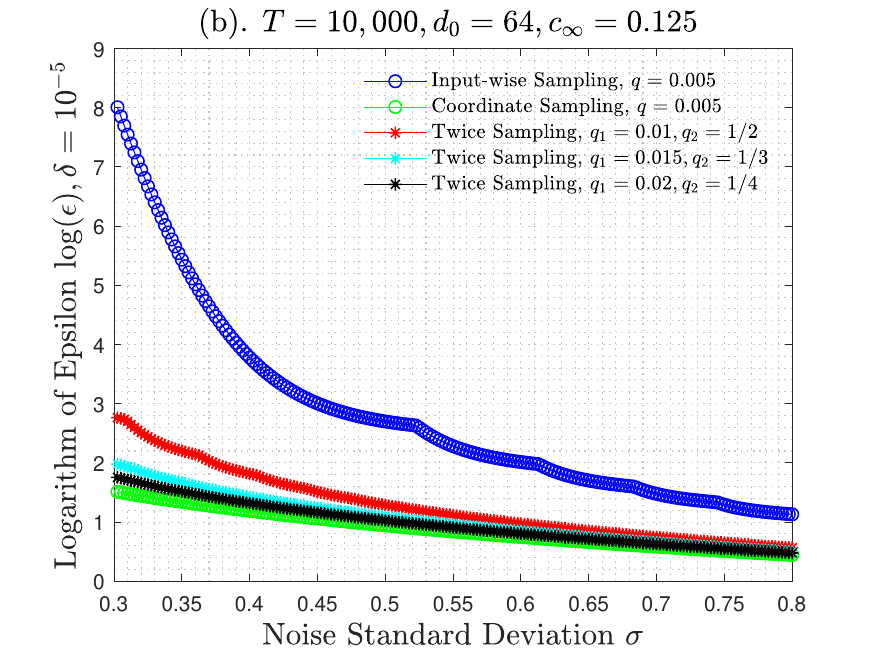}
  \includegraphics[width=0.33\linewidth]{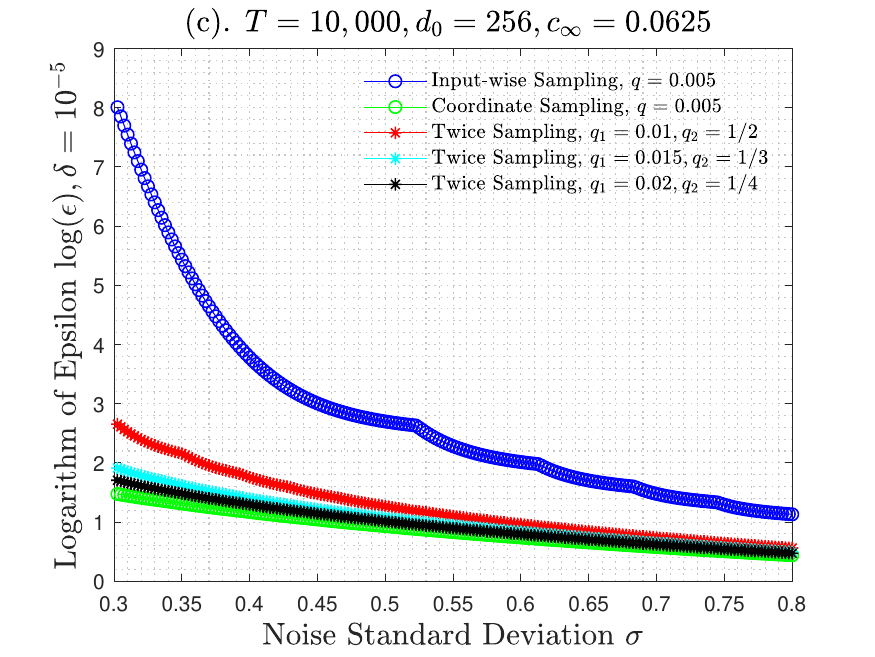} 
\vspace{-0.25 in}
\caption{Comparison between the Logarithm $\log(\epsilon)$ of $(\epsilon, \delta=10^{-5})$-DP Converted by the $T=10, 000$ Composition RDP Bound Produced by Input-wise Sampling $q=0.005$, Coordinate-wise Sampling $q=0.005$ and Twice Sampling $q=q_1q_2=0.005$, $c_2 =1$.}
\label{Fig. amplification eps-delta DP}
\vspace{-0.05 in}
\end{figure*}

\begin{figure}[t] 
    \subfigure[Illustration of Clipping]
    {\includegraphics[width=.475\linewidth]{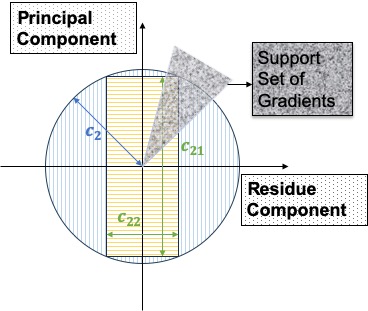}
    \label{Fig: illustration_a}}
    \subfigure[Illustration of Noise]
    {\includegraphics[width=.49\linewidth]{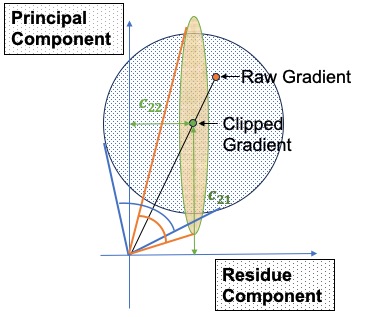}
    \label{Fig: illustration_b}}
    \vspace{-0.2 in}
    \caption{{Illustration of the comparison between Hybrid Clipping and Isotropic Clipping and the corresponding optimal Gaussian noise.}}
    \vspace{-0.25 in}
    \label{Fig: illustration} 
\end{figure}

\section{Applications}

\label{sec: additional exp}
\noindent In this section, we apply our results to privacy-preserving deep learning.
As described in Section \ref{sec: backgroud}, DP-SGD for $T$ iterations is essentially a $T$-adaptive composition on gradient mean estimation, where one may simply take the processing function $\mathcal{F}$ as gradient computation and all our results are straightforwardly applicable. We provide results on training ResNet22 on CIFAR10 and SVHN, respectively, where we assume the entire training dataset is private. 

{ We first continue the example in Section \ref{sec: gradient data}, where we consider the mean estimation of $1,000$ per-sample gradients evaluated on CIFAR10 in {\em one iteration}.
The estimation error is formed by two parts: the clipping error \cite{xiao2023theory} and the DP noise. 
Under the same setup as that described in  Section \ref{sec: gradient data}, the gradients, represented as $d=291,898$-dimensional vectors, are produced from a ResNet22 network on 1,000 randomly sampled CIFAR10 datapoints. 
We select a global $l_2$-norm clipping threshold $c_2 = \sqrt{2.5^2 + 1^2}= 2.69$, and for the hybrid clipping we consider a principal component of rank $r_1=1,000$ and a residue component of rank $r_2=d-r_1$. 
Based on the average power of the projection of the gradients in each subspace, we select a local clipping threshold $c_{21}=2.5$ and $c_{22}=1$ for the two subspaces, respectively. 

To provide more intuition, in Fig. \ref{Fig: illustration} (a), we illustrate both the standard isotropic $l_2$-norm clipping, captured by projection into the blue ball, and the hybrid clipping, captured by projection into the yellow cube. 
The x-axis and y-axis represent the residue and the principal space, respectively. The grey region represents the support set of gradient distributions (true sensitivity geometry). 
It is worth noting that both clipping methods enjoy the same global $l_2$-norm clipping budget, where the difference is that hybrid clipping allocates it differently to different subspaces. 
Moreover, under such a setup, the hybrid clipping cube is fully contained within the isotropic $l_2$ clipping ball. 
To measure the utility of clipped gradient mean estimation, we consider $\cos(\theta) = \frac{\langle v_0, v_c \rangle}{\|v_0\|\|v_c\|}$, the cosine of the angle $\theta$ between the raw gradient mean $v_0$ and the clipped gradient mean $v_c$.
A cosine similarity $\cos(\theta)$ closer to $1$ implies a more accurate estimation on the true gradient direction.  On average, $\cos(\theta)$ for standard $l_2$-norm clipping is 0.78 while that for the hybrid clipping is 0.76. The slight difference is because under such a parameter selection, the hybrid clipping cube is a strict subset of the standard $l_2$-norm ball, which incurs a bit more clipping error.}

{ Now, we consider adding DP noises to the average of clipped gradients, where, for example, we select the scale of DP noises to ensure an $(\epsilon=8, \delta=10^{-5})$-DP guarantee for running DP-SGD for $T=5,000$ iterations with an input-wise subsampling rate $q=1000/50000$. 
In Fig. \ref{Fig: illustration} (b), the red point represents a raw gradient before clipping, and the green one represents its clipped version. 
We illustrate the geometry of the optimal Gaussian noise for both clipping methods. 
As shown in Corollary \ref{cor: l_p mix}, isotropic noise is already optimal for standard $l_2$-norm clipping, which is captured by the blue ball in Fig. \ref{Fig: illustration} (b). 
In contrast, by Theorem \ref{thm: hybrid clipping}, the optimal noise for asymmetric hybrid clipping allocates varying noise power across different subspaces, depending on the space rank and the clipping budget. 
The thin brown ellipse captures the optimal \emph{anisotropic} noise for the above-mentioned hybrid clipping, where the noise we add is much less in the massive residue space. 
To be specific, by Theorem \ref{thm: hybrid clipping}, the variance of DP noises required for standard $l_2$-norm clipping is $5.5\times$ larger than that of hybrid clipping. 
After perturbation, the expectation of $\cos(\theta)$ of standard clipping is $2.9\times$ smaller than that of hybrid clipping. Such improvement translates to improvement in the test accuracy of the model produced as will be shown later in the section.} 

\begin{table}[t]
\scriptsize
\begin{center}
\begin{tabular}{l c c c c c} 
\toprule
Twice-sampling Rate $\backslash$ $\epsilon$ & $2$ & $2.5$ & $4$ & $8$ \\
\midrule
$(q_1 = 0.06, q_2=1/3)$ & 61.1 (86.2) & 65.0 (85.4)  & 71.4 (82.2) & 77.5 (70.8) &  \\
\hline 
$(q_1 = 0.04, q_2=1/2)$ & 60.7 (89.4) & 64.8 (88.2) & 71.1 (77.4) & 77.2 (75.1) & \\
\hline 
$(q_1 = 0.03, q_2=1/3)$ & 59.9 (67.4) & 64.7 (65.9)  & 70.8 (53.2) & 76.7 (43.3)&  \\
\hline 
$(q_1 = 0.02, q_2=1/2)$ & 59.3 (72.8) & 64.4 (71.6) & 70.1 (61.0) & 76.4 (53.5) & \\
\toprule
Input-wise Sampling Rate $\backslash$ $\epsilon$ & $2$ & $2.5$ & $4$ & $8$ \\
\midrule
\hline 
{Baseline Regular DP-SGD} $q = 0.02$ & 59.5 & 61.5  & 67.3 & 73.8 &  \\
\hline 
{Baseline Regular DP-SGD} $q=0.01$ & 55.4 & 60.0 & 64.3  & 70.2 &  \\
\midrule
\midrule
\end{tabular}
\end{center}
\caption{\textbf{Test Accuracy} (and \textbf{Noise Variance Ratio}) (\%) of {training ResNet22 on CIFAR10 with/out Twice Sampling under various $\epsilon$} and fixed $\delta=10^{-5}$.}
\label{tab: pure twice sampling} 
\end{table}

\begin{table}
\scriptsize
\begin{center}
\begin{tabular}{l c c c c c} 
\toprule
Twice-sampling Rate $\backslash$ $\epsilon$ & $2$ & $2.5$ & $4$ & $8$ \\
\midrule
$(q_1 = 0.12, q_2=1/3)$  & 69.7 (24.5) & 72.1 (22.1)  & 76.4 (23.1) & 81.6 (21.5)  &  \\
\hline 
$(q_1 = 0.08, q_2=1/2)$ &  69.4 (25.4) & 72.1 (22.2)  & 76.5 (23.4)  & 81.3 (23.5) &  \\
\hline 
$(q_1 = 0.06, q_2=1/3)$ &  68.6 (21.5) & 71.6 (23.8) & 76.1 (18.4) & 80.8 (19.1) &  \\
\hline 
$(q_1 = 0.04, q_2=1/2)$ &  68.3 (23.2) & 71.4 (27.7) & 75.9 (18.7)  & 81.1 (19.3) & \\
\hline 
$(q_1 = 0.03, q_2=1/3)$ & 67.6 (15.8) & 70.3 (17.1)  & 74.7 (14.3)  & 80.2 (9.58)  &  \\
\hline 
$(q_1 = 0.02, q_2=1/2)$ & 67.8 (16.2) & 70.2 (17.3) & 74.8 (14.6) & 79.8  (12.3) & \\
\midrule
\midrule
\end{tabular}
\end{center}

\caption{{\textbf{Test Accuracy} (\%) of {training ResNet22 on CIFAR10 with both Twice Sampling and Hybrid Clipping provided 2,000 Public ImageNet samples under various $\epsilon$} and fixed $\delta=10^{-5}$} (\textbf{Noise Variance Ratio} (\%) between that with both Twice Sampling and Hybrid Clipping and Regular DP-SGD).}
\label{tab: hybrid+twice sampling} 
\end{table}

In the following, we consider the full implementation of DP-SGD. 
In the first set of experiments, we do not assume any public data and apply twice sampling combined with a mixture of $l_2$-norm and $l_{\infty}$-norm clipping, where we take $d_0=100$, i.e., $c_{\infty} = 0.1 \cdot c_2$. As analyzed in Section \ref{sec: gradient data}, such additional $l_{\infty}$-clipping makes negligible changes to the $l_2$-norm clipped per-sample gradient. 
In Table \ref{tab: pure twice sampling}, we report the test accuracy and the DP noise variance ($\mathbb{E}[\|\bm{e}\|^2]$) ratio between that of $(q_1, q_2)$-twice sampling and $q=q_1q_2$-input-wise sampling (shown in the brackets). 
For each selection of $\epsilon=\{2, 2.5, 4, 8\}$ (each column of Table \ref{tab: pure twice sampling}), we set a corresponding $T=\{1500, 2000, 2500, 5000\}$, and run for $5$ trials and report the median of accuracy. {We need to stress that, as the baseline, the performance of standard DP-SGD with only input-wise sampling, as reported in the last two rows of Table {\ref{tab: pure twice sampling}}, has been optimized. For each case, we search for the optimal hyperparameters, including the selections of clipping threshold $c$ and the number of iterations $T$, such that the standard DP-SGD produces the best accuracy. Then, with exactly the same selection of those hyperparameters, we further incorporate twice sampling in DP-SGD, i.e., an additional $l_{\infty}$-norm clipping and a coordinate-wise sampling. We then report the corresponding performance in the first four rows of Table {\ref{tab: pure twice sampling}}. Our goal here is to provide a clear picture on how much improvement is produced by the sharpened noise bound on the model performance.} 

Consistent with our amplification simulation in Section \ref{sec: simulation sampling}, the improvement due to twice sampling is more significant for smaller sampling rate and smaller noise (larger privacy budget). Due to twice-sampling, for medium privacy budget $\epsilon \geq 4$, the performance gap among different sampling rates is not appreciable. 
Given $\epsilon=8$, with $q_1=0.02, q_2=1/2$, where in expectation we calculate the gradients of $1,000$ samples in each iteration and randomly select $500$ for each coordinate, we achieve $76.4\%$ accuracy, while via $q=0.01$ input-wise sampling, the accuracy is only $70.2\%$; By Theorem \ref{thm: sampling twice}, the improved noise variance is only $53.5\%$ of that for $q=0.01$ input-wise sampling. 

\begin{table}[t]
\scriptsize
\begin{center}
\begin{tabular}{l c c c c c} 
\toprule
Twice-sampling Rate $\backslash$ $\epsilon$ & $2$ & $2.5$ & $4$ & $8$ \\
\midrule
$(q_1 = 0.06, q_2=1/3)$ &  79.8 (89.0) &  80.8 (85.4)  &  88.0 (82.2) & 89.5 (70.8) &  \\
\hline 
$(q_1 = 0.04, q_2=1/2)$ & 79.6 (89.4) &  80.7 (88.2) & 88.1 (77.4) & 89.3 (75.1) & \\
\hline 
$(q_1 = 0.03, q_2=1/3)$ &  78.3 (67.4) &  80.3 (65.9)  &   87.8 (53.2) & 89.1 (43.3)&  \\
\hline 
$(q_1 = 0.02, q_2=1/2)$ &  78.4 (72.8) &  80.1 (71.6) &  87.5 (61.0) &  88.8 (50.4) & \\
\toprule
Input-wise Sampling Rate $\backslash$ $\epsilon$ & $2$ & $2.5$ & $4$ & $8$ \\
\midrule
\hline 
{Baseline Regular DP-SGD} $q = 0.02$ & 78.4 & 79.6  & 83.1 & 86.3 &  \\
\hline 
{Baseline Regular DP-SGD} $q=0.01$ & 73.8 & 74.9 & 79.5  & 82.9 &  \\
\midrule
\midrule
\end{tabular}
\end{center}
\vspace{-0.05 in}
\caption{{\textbf{Test Accuracy} (and \textbf{Noise Variance Ratio}) (\%) of {training ResNet22 on SVHN with/out Twice Sampling under various $\epsilon$} and fixed $\delta=10^{-5}$.}}
\label{tab: pure twice sampling SVHN} 
\vspace{-0.35 in}
\end{table}

\begin{table}[t]
\scriptsize
\begin{center}
\begin{tabular}{l c c c c c} 
\toprule
Twice-sampling Rate $\backslash$ $\epsilon$ & $2$ & $2.5$ & $4$ & $8$ \\
\midrule 
$(q_1 = 0.06, q_2=1/3)$ &   87.3 (45.5) &  88.3 (41.1) &  89.8 (43.3) &  91.3 (31.2) &  \\
\hline 
$(q_1 = 0.04, q_2=1/2)$ &   87.2 (46.1) &  88.4 (41.3) &  89.5 (48.6)  &  91.2 (38.4) & \\
\hline 
$(q_1 = 0.03, q_2=1/3)$ & 86.5 (40.3) &   87.9 (32.4)  &  89.7 (26.9)  & 91.1 (15.6)  &  \\
\hline 
$(q_1 = 0.02, q_2=1/2)$ & 86.5  (40.7) &  87.7 (37.3) &  89.6 (33.2) &  90.9 (21.5) & \\
\midrule
\midrule
\end{tabular}
\end{center}
\vspace{-0.05 in}
\caption{{\textbf{Test accuracy}) (\%) of {training ResNet22 on SVHN with both Twice Sampling and Hybrid Clipping provided 2,000 public ImageNet samples under various $\epsilon$ and fixed $\delta=10^{-5}$}} (\textbf{Noise Variance Ratio} between that with both Twice Sampling and Hybrid Clipping and that of Regular DP-SGD).}
\label{tab: hybrid+twice sampling SVHN} 
\vspace{-0.35 in}
\end{table}

To proceed, in the second set of experiments, we assume a small amount of public data of weak similarity to CIFAR10 to enable the subspace approximation. We adopt the same setup as that of \cite{embedyu2021}, where we randomly select 2,000 samples from ImageNet \cite{Imagenet}, an image pool containing millions of images in thousands of classes, assumed to be public. We then iteratively apply the {\em power method} \cite{power-method, embedyu2021} on public data to approximate four principal subspaces of rank $\{250, 500, 1000, 1500\}$, respectively, and the subsequent residue component subspace. We then apply hybrid clipping and $(q_1, q_2)$-twice sampling together with optimized noise described in Theorem \ref{cor: optimized noise hybrid+twice}. 
This is described as Algorithm \ref{alg: hybrid+twice-sampling} in Appendix \ref{app: hybrid+twice sampling}.
In Table \ref{tab: hybrid+twice sampling}, we record the test accuracy and the ratio between the noise variance given our improved analysis via Theorem \ref{cor: optimized noise hybrid+twice} and that of a trivial, subsampled Gaussian mechanism of $q=q_1q_2$ input-wise sampling (shown in brackets). {Still, to have a clear and fair comparison, all the results reported in Table {\ref{tab: hybrid+twice sampling}} are under the same hyperparameter selections as those for standard DP-SGD in Table {\ref{tab: pure twice sampling}}. We simply further incorporate hybrid clipping by allocating the same global $l_2$-norm budget $c$ to different subspaces, depending on the expected norm of public gradients projected into each subspace.} It is noted that with further hybrid clipping, we achieve almost an order of magnitude improvement on the noise variance. For efficiency, we only use public data to approximate four, relatively small, principal components, but one may split the entire space into more subspaces with more fine-grained clipping, and apply Theorem \ref{cor: optimized noise hybrid+twice} to get even tighter noise bounds.

{We also implement the two above-described sets of experiments on SVHN datasets, shown in Tables {\ref{tab: pure twice sampling SVHN}} and {\ref{tab: hybrid+twice sampling SVHN}}, respectively. The observations are very similar. }

We want to mention that, as our main focus is to study and compare the fundamental privacy and efficiency improvement through twice sampling and the optimal noise for hybrid clipping, we do not very carefully fine-tune the neural network architectures. We only implement standard DP-SGD with proper data augmentation in the experiments, though we note that many nice empirical tricks, such as weight standardization and parameter averaging, are recently proposed in \cite{deepmind} to also significantly enhance the performance of DP-SGD in deep learning from an optimization perspective. Using large batchsize (with input-wise sampling $(q= 0.32)$),  \cite{deepmind} achieves median $62.5\%$ and $80.3\%$ accuracy on CIFAR10 on WideResNet with privacy guarantee $(\epsilon=2,\delta =10^{-5})$ and $(\epsilon=8,\delta = 10^{-5})$, respectively, from our reproduction. {With assistance of a small set of public data, we outperform the state-of-the-art with much lower overhead in terms of both memory and computation time: in all the experiments reported on CIFAR10, our effective batchsize is upper bounded by 2,000.} We release our code (see footnote 5) to help other researchers, especially from the machine learning community, to further improve our results. Besides, compared to the state-of-the-art results with gradient embedding in a same setup of $2,000$ public ImageNet data, \cite{embedyu2021} only achieves an average $73.4\%$ accuracy with $(\epsilon=8, \delta = 10^{-5})$ on CIFAR10 using ResNet20, due to looser privacy analysis.  Our code can be found on GitHub\footnote{\url{https://github.com/Hanshen-Xiao/Twice_Sampling_and_Hybrid_Clipping}}.  

\vspace{-0.05 in}
\section{Related Works}
\label{sec: related-works}
\noindent \textbf{Sensitivity Geometry}: 
Around the same time when Gaussian and Laplace mechanisms were proposed to capture $l_2/l_1$-norm sensitivity, the study on the minimal perturbation for more generic sensitivity has attracted considerable attention. 
Rooted in the applications of private linear query, {\em $K$-norm mechanism} is first proposed in \cite{K-mechanism}, which is shown to produce nearly asymptotically tight (ignoring logarithmic terms) $\epsilon$-DP utility-privacy tradeoff for a class of convex and symmetric sensitivity sets $K$. Recently, further comparison among  different selections of $K$ and the corresponding noise scale required is 
studied in \cite{awan2021structure}. Though the $K$-norm mechanism generates a geometry-adapted noise such that each element lying on the boundary of $K$ has dominating sensitivity in $\epsilon$-DP, it has two major limitations.
First, $K$-norm noise is, in general, inefficient to generate, which requires a uniform sampling over a convex set  \cite{lovasz1993random}. 
Second, the $K$-norm mechanism is most suitable for pure $\epsilon$-DP and it is known that if we switch to the approximate $(\epsilon,\delta)$-DP, there could be a ${\Omega}(\sqrt{d})$ gap between the optimal error and the $K$-norm perturbation \cite{geometry-approx,de2012lower}. 
One main motivation to consider $(\epsilon,\delta)$-DP is to enable advanced composition to upper bound the accumulated privacy loss from multiple releases. 
$T$-fold $(\epsilon_0, \delta_0)$ leakage can be bounded by $\tilde{O}(\sqrt{T}\epsilon_0,T\delta_0)$ while in pure DP, $T$-fold $\epsilon_0$-DP leakage can only be bounded as $T\epsilon_0$-DP \cite{boosting2010}. 
However, since $(\epsilon,\delta)$-DP essentially characterizes a tradeoff function between the two security parameters $\epsilon$ and $\delta$, the corresponding optimal perturbation becomes even harder to construct and analyze. 
So far, only asymptotic results are known for some special cases, mainly applications in linear query \cite{geng-approxDP,bun2014fingerprinting, discrepancy} and convex Lipschitz optimization \cite{bassily2014private}, where the sensitivity set is some transformed or variants of $l_2$ ball. 

The underlying challenges for further generalization are mainly twofold. 
{First, to show optimality, compared to many nice tools that have been developed such as hereditary discrepancy \cite{discrepancy} and fingerprint code  
\cite{bun2014fingerprinting} to prove noise lower bounds, analyzable and efficient randomization as the noise upper bound is less known besides the basic Gaussian mechanism.} This presents challenges to prove the optimality with matched upper and lower bounds. 
Second, and the more practical issue is that, to obtain tighter composition bounds, many DP variants with more complex divergence metrics $\rho$ are developed such as Rényi-DP (RDP) \cite{renyi}. This further complicates the study on optimal perturbation if we want to simultaneously use those advanced tools. Thus, with a careful balance between both theory and practice, in this paper we stick to RDP and have proposed new tricks to study the optimality of Gaussian noise. 

\noindent \textbf{Sampling and Privacy Amplification}: 
The study on DP amplification by Poisson (i.i.d.) sampling dates back to \cite{li2012sampling}. 
In general, the classic privacy amplification problem can described as follows: if a mechanism $\mathcal{M}$ satisfies certain DP guarantees, then what kind of DP guarantees does the composite mechanism $\mathcal{M}^S = \mathcal{M} \circ \mathcal{S}$ have, where $\mathcal{S}$ is some sampling subroutine on input data? 
For $(\epsilon, \delta)$-DP, Balle et al. in  \cite{balle2018privacy} provide generic amplification bounds for a class of sampling methods $\mathcal{S}$, including Poisson sampling and sampling with/out replacement. 
As for RDP, amplification for Poisson sampling with the Gaussian mechanism is studied in \cite{mironov2019r}, and Zhu and Wang present more generic algorithm-independent results in  \cite{zhu2019poission}. 
However, those classic amplification results cannot fundamentally address the curse of dimensionality, unless it can be solved for the original processing $\mathcal{M}$ before the sampling. 
In this paper, we do not take sampling as a blackbox but instead carefully study the algorithmic randomness, especially when we further implement coordinate-wise sampling. 
We have proposed a novel twice sampling protocol to force the sampling randomness to fit the desired high-dimensional geometry. 
Our more involved and closed-form RDP analysis of the proposed twice sampling could also be of independent interest to derive tighter composite sampling privacy amplification.

\noindent \textbf{Dimension Reduction and Private Deep Learning}: 
In theory, sparsity and low rank are two of the most commonly-used assumptions for learnable high-dimensional data. Clearly, when the objective processing {\em does} have certain good properties, the curse of dimensionality of DP noise can be broken. For example, given a sparsity assumption, the {\em sparse vector technique} \cite{dwork2014algorithmic}\footnote{Instead of post processing released data based on priors, in this paper we study, more fundamentally, the randomization that fits the sensitivity geometry in the first place.}
is known to only require a scale of noise logarithmically dependent on the dimension. Research on figuring out conditions when the utility loss could be (nearly) independent of the dimension remains active. One example is private optimization on generalized linear model (GLM) \cite{kairouz2021nearly}, where due to the strong concentration, the scale of subGaussian noise under bounded linear operation is constant. However, there is a large gap between theory and practice. Those good properties do not hold for many complicated processing tasks, and artificial approximation such as sparsification may cause large bias  \cite{Luo2021sparse,zhang2021sparse,Zhu2021sparse}. $l_2$-norm clipping  \cite{deepmind} and its variants, such as layer clipping \cite{layer} or subspace embedding clipping \cite{embedyu2021}, where the objective is split into several segments and each is $l_2$-norm clipped with possibly different parameters, are still the most popular options, especially for deep learning. However, as the corollaries of our results in Section \ref{sec: optimal noise}, those privacy analyses \cite{layer, embedyu2021} are sub-optimal. 

{Indeed, our results also indicate that simply projecting isotropic noise to the objective sensitivity set is inadequate to fit the geometry and leads to suboptimal performance in general. Numerous prior works follow this line to construct noise. For example, if the sensitivity set $\mathsf{S}$ is some subspace of $\mathbb{R}^d$, one may first select a large enough $l_2/l_1$ ball that contains $\mathsf{S}$ and inject a noise following a Gaussian/Laplace mechanism. 
Then, one projects the noisy output back to $\mathsf{S}$ as a postprocessing, which does not cause additional privacy risk  {\cite{song2021evading}},{\cite{embedyu2021}},{\cite{zhoubypassing2021}}. However, as shown by Theorems {\ref{thm: cube}}-{\ref{thm: hybrid clipping}} in Section {\ref{sec: optimal noise}}, the optimal noise bound can be much smaller compared to such post-projected noise.} 

Moreover, due to the lack of theory to systematically improve the privacy-utility tradeoff, current studies on private deep learning mainly focus on searching for the optimal model and hyper-parameters \cite{papernot2020making, deepmind}. {Our results, focusing on the more fundamental optimal perturbation problem, shed new light on systematically improving DP-SGD by developing more efficient high-dimensional clipping with geometry-reflected randomization.}

\vspace{-0.05 in}
\section{Conclusion and Limitations}
\label{sec: conclusion}
In this paper, we study the optimal Gaussian noise for hybrid clipping and propose twice sampling to capture two important geometry properties in practical high-dimensional data processing: asymmetric (non-uniform) distribution and free $l_{\infty}$-norm restriction. We have presented more fundamental results to sharpen the privacy analysis with better randomization and advance the understanding of high-dimensional sensitivity geometry.
There are several promising directions for further generalization. First, though we prove the optimal Gaussian noise bounds in various setups, it does not mean that a Gaussian is the optimal perturbation for desired sensitivity geometry. A next step could be to generalize our optimality results on a broader class of noise distributions, such as 
log-normal, Gumbel and Rayleigh \cite{renyi-common}, which have analyzable Rényi divergences, and explore whether they are more suitable to a certain geometry. As another direction, our results on twice sampling could also be generalized to study more complicated composition of samplings on different dimensions and enforce sampling randomness that reflects different sensitivity geometries.

\noindent
\textbf{Limitations:} Hybrid clipping in general requires stronger directional information on the processed output distribution, which usually needs assistance from public data in practice. Though in the experiments, we only assume a small amount of weakly-correlated data to help determine the embedding/projection parameters, how to  implement hybrid clipping or an even more efficient clipping method based on only sensitive data is an important question for future work.  This may require more extensive studies on practical high-dimensional data distributions and looking for more stable and easily-estimated features.

\vspace{-0.05 in}
\section{Acknowledgements}
We would like to thank Yuqing Zhu for very helpful discussions. We gratefully acknowledge the support of DSTA Singapore, Cisco Systems, Capital One, and a MathWorks fellowship. We also thank the anonymous reviewers for their constructive feedback.  

\bibliographystyle{ACM-Reference-Format}
\bibliography{ccs-sample}

\appendix
\section{Additional Discussion}
\label{sec: add discussion}
{In this section, we want to provide more intuition on how hybrid clipping avoids the curse of dimensionality and discuss more about its implication to construct efficient clipping for practical data processing with black-box success. First, why must standard $l_1/l_2$-norm clipping require a noise in a scale $\Theta(\sqrt{d})$? One intuitive explanation is because we do not know which direction the possible output change will be from.} 
With $l_2$-norm clipping, we can only guarantee that when one arbitrarily removes an individual from the input set, the change is bounded. In other words, for any unit vector $v \in \mathbb{R}^d$, the magnitude of the projection of $\mathsf{S}$ along $v$ is bounded.
However, it could appear in {\em any} possible direction in $\mathbb{R}^d$. 
Therefore, we need to ensure that {\em the noise is of sufficient power such that its variance along any direction in $\mathbb{R}^d$ is big enough to hide the possible change.} This essentially causes an unavoidable $\Theta(\sqrt{d})$ noise scale. 
Thus, in the context of aggregation with DP, when the number of samples $n \ll \sqrt{d}$, we cannot average out the noise and learn anything meaningful from the private release. 

However, the curse of dimensionality from a worst-case privacy perspective also raises a very interesting question about the empirical success of non-private high-dimensional processing, especially deep learning.
Even without DP noise, a statistical data processing still needs to handle the statistical noise of the same dimension due to the data dispersion. 
Nowadays, machine learning with increasingly large models has become a popular trend to improve prediction performance. Given access to representative datasets, non-private deep learning has witnessed many remarkable successes, where for certain image classification problems, well-trained neural networks can already achieve human-level performance 
 \cite{resnet, he2017mask}, even leaving aside the recent breakthrough by ChatGPT  \cite{openai2023gpt4} in large language models with hundreds of billions parameters. 
So {\em why does deep learning not suffer from the dimensionality curse}? 
Over the last several decades, many researchers have tried to explain this mystery by providing evidence of good structural properties of neural networks. 
For example, \cite{low-rank2022, embedyu2021} show that fine-tuned deep models could be distributed in some low-rank space. 
The success of network pruning/compression \cite{han2015deep} shows that there is usually a large redundancy in network representation. 
There are more involved analyses to prove that  under certain assumptions, the fat shattering dimension \cite{anthony1999neural} or Rademacher complexity \cite{golowich2018size} of multilayer perceptrons can be model-size independent. 

Explainable deep learning is still a very active area in machine learning and a full list of all hypotheses is beyond the scope of this paper. 
So far, it is still too early to draw a conclusion about the determining factor and thus we argue for a more conservative way to maintain the empirically successful processing $\mathcal{F}$ as a black box. Our premise is that deep learning can exploit certain good properties of practical data to avoid the dimensionality curse in the average case. 
The key problem left is, with this weak premise, how can we design efficient privatization to fit the objective blackbox processing $\mathcal{F}$ that allows noise with weaker dependence on dimension, rather than artificially modifying $\mathcal{F}$ to fit certain conditions?    

Given the state-of-the-art advances in both privacy and statistics, one of our hopes is to use better clipping to bridge the gap between the average and the worst case. 
Intuitively, if the sampling noise from average-case practical data is tolerable, so should the DP noise. 
{Hybrid clipping provides an example, where, on one hand, we avoid the curse of dimensionality by introducing more involved directional constraints on the power of sensitivity, and meanwhile we preserve the original processing as a black box with minimal change to its distribution. The additional information on the sensitivity allows us to only add necessary noise along each direction to mitigate the strict dependence on the dimensionality $d$. On the other hand, as discussed in Section {\ref{sec: warmup}}, hybrid clipping with $l_{\infty}$-norm constraint is only determined by a few stable aggregate statistics, such as the principal components and the average of the power in each of them, which capture the populational statistics of underlying processed output distribution.} This is a much more smooth operation compared to many other existing clipping methods, such as sparsification \cite{Luo2021sparse,zhang2021sparse,Zhu2021sparse}, where only significant coordinates are preserved or participate in the processing while the remaining are either frozen or removed. Though these artificial dimension-reduction techniques can also decrease the noise scale, the advantage can be easily offset by the large clipping bias produced and may not outperform simple $l_2$-norm clipping, especially in deep learning  \cite{deepmind}. 




\section{Proof of Theorem \ref{thm: signAndperm}}
\label{app: thm: signAndperm}
We adopt the following notations.
Given a unitary matrix $U$ whose columns form a basis, we use $\bm{u}_i$ to denote its $i$-{th} column or the $i$-{th} basis vector.
Similarly, we use $u_{ij}$ to represent the $j$-th coordinate of $\bm{u}_i$.
Recall that we define for any $\bm{s} = (s_1, \dotsb, s_d)$, the function $\mathcal{L}(U, \bm{\sigma}, \mathsf{S})$ can be rewritten as
\[
\mathcal{L}(U, \bm{\sigma}, \mathsf{S}) =
\sup_{\bm{s} \in \mathsf{S}} \sum_{i = 1}^d ({<\bm{s}, \bm{u}_i> \over \sigma_i})^2.
\]
For convenience, we define $\mathcal{L}(U, \bm{\sigma}, \bm{s})$, for an element $\bm{s}$ instead of a set $\mathsf{S}$, as
\[
\mathcal{L}(U, \bm{\sigma}, \bm{s}) = \sum_{i = 1}^d ({<\bm{s}, \bm{u}_i> \over \sigma_i})^2.
\]
Therefore, $\mathcal{L}(U, \bm{\sigma}, \mathsf{S}) = \sup_{\bm{s} \in \mathsf{S}} \mathcal{L}(U, \bm{\sigma}, \bm{s})$.
We now provide the proof of Theorem \ref{thm: signAndperm} as follows.

\OptimalSymmetric*
\begin{proof}
Since the set $\mathsf{S}$ is insensitive to sign and permutation, for any $\bm{s} = (s_1, \dotsb, s_d)\in \mathsf{S}$, $\bm{z} = (z_1, \dotsb, z_d)\in \{-1, 1\}^d$ and any permutation $\pi$, the transformed datapoint
\begin{equation}\label{equ:defnl}
t(\bm{s}, \bm{z}, \pi) = (z_1\cdot s_{\pi(1)}, \dotsb, z_d\cdot s_{\pi(d)})
\end{equation}
is also in $\mathsf{S}$.
We use $T(\bm{s})$ to denote the set $\{t(\bm{s}, \bm{z}, \pi)~|~\bm{z}, \pi\}$ for all selections of $\bm{z}$ and $\pi$. 
For any unitary matrix $U$, any $\bm{\sigma}$ and any $\bm{s}\in S$, we have
\begin{equation}\label{equ:sumT}
\begin{aligned}
\mathcal{L}(U, \bm{\sigma}, \mathsf{S})
& \ge {1\over |T(\bm{s})|}\sum_{\bm{s}'\in T(\bm{s})} \mathcal{L}(U, \sigma, \bm{s}') \\
& = {1\over |T(\bm{s})|}\sum_{\bm{z}, \pi} \sum_{i = 1}^d ({<t(\bm{s}, \bm{z}, \pi), \bm{u}_i> \over \sigma_i})^2,
\end{aligned}
\end{equation}
since the maximum is no less than the average of a set of numbers. Here, $|T(\bm{s})|$ represents the number of elements in $T(\bm{s})$. 
Note that
\begin{equation}
\label{equ:crossToZero}
\begin{aligned}
&~~(<t(\bm{s}, \bm{z}, \pi), \bm{u}_i>)^2 = (\sum_{j = 1}^d z_j s_{\pi(j)} u_{ij})^2 \\
=
& ~~\sum_{j=1}^d (s_{\pi(j)}u_{ij})^2 + \sum_{j\neq k} z_jz_k\cdot (s_{\pi(j)} u_{ij})\cdot (s_{\pi(k)} u_{ik}).
\end{aligned}
\end{equation}
When we sum over all $\bm{z}\in \{-1, 1\}^d$, the second term goes to 0 as $\sum_{\bm{z}} \bm{z}_j\bm{z}_k = 0$ for any $j \not = k$.
Further, the first term is not related to $\bm{z}$.
Now, we go back to (\ref{equ:sumT}) and we have
\begin{equation}\label{equ:sumT2}
\begin{aligned}
\mathcal{L}(U, \bm{\sigma}, \mathsf{S})
& \ge {1\over |T(\bm{s})|}\sum_{\bm{z}, \pi} \sum_{i = 1}^d ({<t(\bm{s}, \bm{z}, \pi), \bm{u}_i> \over \sigma_i})^2 \\
& = {1\over |\{\bm{z}\}| \cdot |\{\pi\}|}\sum_{\bm{z}, \pi} \sum_{i = 1}^d \sum_{j=1}^d ({s_{\pi(j)}u_{ij}\over \sigma_i})^2 \\
& = {1\over |\{\pi\}|} \sum_{i = 1}^d \Big( {1\over \sigma^2_i}\cdot \sum_{j=1}^d (u^2_{ij}\sum_{\pi} s^2_{\pi(j)}) \Big).
\end{aligned}
\end{equation}
An important observation here is that since we are summing over all possible permutations, $\sum_{\pi} s^2_{\pi(j)} = (d - 1)!\|\bm{s}\|_2^2$ for all $j$.
Therefore, 
\begin{equation}\label{equ:sumT20}
\begin{aligned}
\mathcal{L}(U, \bm{\sigma}, \mathsf{S})
& \ge {1\over d!} \sum_{i = 1}^d \Big({1\over \sigma^2_i}\cdot (d - 1)!\|\bm{s}\|_2^2\cdot \sum_{j=1}^d u^2_{ij} \Big) \\
& = {\|\bm{s}\|_2^2\over d} \sum_{i = 1}^d {\|\bm{u}_i\|_2^2\over \sigma^2_i}
= {\|\bm{s}\|_2^2\over d} \sum_{i = 1}^d {1\over \sigma^2_i}.
\end{aligned}
\end{equation}
Since $\sum_{i = 1}^d \sigma^2_i = 1$, H\"older inequality implies that
\[
\sum_{i = 1}^d {1\over \sigma^2_i} = \sum_{i = 1}^d {1\over \sigma^2_i}\cdot \sum_{i = 1}^d \sigma^2_i \ge d^2.
\]
The minimum is achieved when $\sigma_1 = \dotsb = \sigma_d = 1/ \sqrt{d}$.
Taking this back to (\ref{equ:sumT20}), we have $\mathcal{L}(U, \sigma, S) \ge d\|\bm{s}\|_2^2$ for any $\bm{s}\in S$.
Note that for any unitary matrix $U$, as long as we choose $\sigma_1 = \dotsb = \sigma_d = 1/ \sqrt{d}$, the privacy loss on any input $\bm{s}$ is exactly $d\|\bm{s}\|_2^2$.
This implies $\mathcal{L}(U, \bm{\sigma}, \mathsf{S})$ is exactly $d(\max_{\bm{s}\in S}\|\bm{s}\|_2)^2$ for any unitary matrix $U$ and that the optimal privacy loss is achieved when we select $\sigma_1 = \dotsb = \sigma_d = 1/ \sqrt{d}$.
\end{proof}

\section{Proof of Theorem \ref{thm: cube}}
\label{app: thm: cube}
We first prove a useful lemma. A matrix $M$ is called stochastic matrix if each entry of $M$ is non-negative and the sum of each row or column equals 1. 
\begin{lemma}\label{lem:stochasticConcave}
For any $d \times d$ doubly stochastic matrix $M$, any concave function $f$ and any non-negative $a_1\ge a_2 \ge \dotsb \ge a_d \ge 0$, we have
\[
\sum_{i = 1}^d f(\sum_{j = 1}^d a_j M_{ij}) \ge \sum_{i = 1}^d f(a_i),
\]
where $M_{ij}$ is the entry of $M$ at the crossing of $i$-th row and $j$-th column. 
\end{lemma}
\begin{proof}
Let $\Psi$ be the set of all $d\times d$ doubly stochastic matrices. 
Given $M\in \Psi$ and $a_1\ge a_2 \ge \dotsb \ge a_d \ge 0$, we define the function
\begin{equation}\label{equ:defnf}
F(M, a) = \sum_{i = 1}^d f(\sum_{j = 1}^d a_j M_{ij}).
\end{equation}
We will prove by contradiction that for any non-negative vector $a= (a_1, a_2, \cdots, a_d)$, $F(M, a)$ reaches maximum when $M$ is the identity matrix $\bm{I}$, i.e.,
\[
\bm{I} = \text{argmin}_{M\in \Psi}F(M, a).
\]
Suppose that $P = \text{argmin}_{M \in \Psi}F(M, a)$ and $P \neq I$.
Note that in (\ref{equ:defnf}), switching two rows of the matrix $M$ does not affect the output of $F(M, a)$.
Let $v_i = f(\sum_{j = 1}^d a_j P_{ij})$, we can assume w.l.o.g. that $v_1\ge v_2\ge \dotsb \ge v_d$.
If not, we can rearrange the rows of $P$ to make sure this holds without affecting the value of $F(P, a)$.

Let $i$ be the smallest such that $P_{ii}\neq 1$.
Since $P_{ii}\neq 1$, there must exists $j$ and $k$ such that $P_{ij} > 0$ and $P_{ki} > 0$.
This also implies that
\[
P_{jj} \le 1 - P_{ij} < 1 ~~~\text{and}~~~ P_{kk}\le 1 - P_{ki} < 1.
\]
Since $i$ is the smallest such that $P_{ii}\neq 1$, it must be that $j > i$ and $k > i$.
Let $\Delta = \min(P_{ij}, P_{ki})$.
We define a new matrix $Q$ such that $Q = P$ except in the following four positions,
\[
Q_{ii} = P_{ii} + \Delta, Q_{ij} = P_{ij} - \Delta, Q_{ki} = P_{ki} - \Delta, Q_{kj} = P_{kj} + \Delta.
\]
It can be verified that $Q$ remains a doubly stochastic matrix and that
\begin{equation}\label{equ:PQdiff}
\begin{aligned}
& ~~~F(P, a) - F(Q, a) \\
= & ~~~f(\sum_{l = 1}^d a_lP_{il}) + f(\sum_{l = 1}^d a_lP_{kl}) - 
    f(\sum_{l = 1}^d a_lQ_{il}) - f(\sum_{l = 1}^d a_lQ_{kl}) \\
= & ~~~ f(v_i) + f(v_k) - f(v_i + (a_i - a_j)\Delta) - f(v_k - (a_i - a_j)\Delta).
\end{aligned}
\end{equation}
Since $f$ is a concave function, for any $x \ge y$ and $\delta \ge 0$,
\[
f(x) + f(y) \ge f(x + \delta) + f(y - \delta).
\]
This is because $f(x+\delta)-f(x) = \int_{x}^{x+\delta} f'(z)dz$ while $f(y)-f(y-\delta) = \int_{y-\delta}^{y} f'(z)dz$. 
Since $f''(z) \leq 0$ and thus $f'(z)$ is non-increasing, and the above holds as assumed $x \geq y$. 

Recall that $j > i$ and $k > i$, which implies that
\[
v_k \ge v_i~~~\text{and}~~~ a_i - a_j \ge 0.
\]
Taking it back into (\ref{equ:PQdiff}), we get that $F(P, a) \ge F(Q, a)$.
This means that we can keep updating the matrix without increasing $F(\cdot, a)$.
Notice that every time we update the matrix, a non-zero position in the off-diagonal of row $i$ becomes zero.
Therefore, after finite number of updates, the matrix becomes $I$.
This implies that $F(I, a)\le F(P, a)$, which implies that $I = \text{argmin}_{M\in \Psi}F(M, a)$.

In conclusion, for any doubly stochastic $M$, any concave function $f$ and any non-negative $a_1\ge a_2 \ge \dotsb \ge a_d \ge 0$,
\[
F(M, a) \ge F(I,a) = \sum_{i = 1}^d f(a_i).
\]
\end{proof}

We now move on to the proof of Theorem \ref{thm: cube}.
\OptimalHyperCube*
\begin{proof}
The set $\mathsf{S}$ is a hypercube defined by the basis $U = (\bm{u}_1, \dotsb, \bm{u}_d)$, where
\begin{equation*}
    \mathsf{S} = \{\bm{s}=\sum_{l=1}^d v_l\bm{u}_l: v_l \in [-V_l,V_l], l=1,2,\cdots, d\}.
\end{equation*}
W.l.o.g., in the rest of the proof, we simply consider all vectors are expressed using the basis $U = (\bm{u}_1, \dotsb, \bm{u}_d)$ and the coordinate is also with respect to such expression. In other words, when we say a vector $\bm{x} = (x_1, \dotsb, x_d)$, we means that $\bm{x} = x_1\bm{u}_1 + \dotsb + x_d\bm{u}_d$.
In this way, the set $S$ can be rewritten as 
\[
 \mathsf{S} = \{(s_1, \dotsb, s_d)~|~ \forall l=1,2,\cdots, d, s_l \in [-V_l,V_l]\}.
\]
The advantage of this representation is that $\mathsf{S}$ is now invariant to sign under the new basis.
This allows us to use the same technique as in the proof of Theorem \ref{thm: signAndperm}.

Let us consider any unitary matrix $W = (\bm{w}_1, \dotsb, \bm{w}_d)$ and $\bm{\sigma}$.
For any $\bm{s} = (s_1, \dotsb, s_d) \in S$, we define set
\[
T(\bm{s}) = \{t(\bm{z}, \bm{s}) = (z_1s_1, \dotsb, z_ds_d) ~|~ z_1, \dotsb, z_d \in \{-1, 1\}\}.
\]
For any $\bm{s} \in S$, we have
\begin{equation}\label{equ:sumT3}
\begin{aligned}
\mathcal{L}(W, \bm{\sigma}, \mathsf{S})
& \ge {1\over |T(\bm{s})|} \sum_{\bm{s}'\in T(\bm{s})} \mathcal{L}(W, \bm{\sigma}, \bm{s}') \\
& = {1\over |T(\bm{s})|} \sum_{\bm{z}} \sum_{i = 1}^d ({<t(\bm{z}, \bm{s}), \bm{w}_i> \over \sigma_i})^2 \\
& = {1\over |\{\bm{z}\}|} \sum_{\bm{z}} \sum_{i = 1}^d ({\sum_{j = 1}^d z_js_j w_{ij}\over \sigma_i})^2 \\
& = {1\over |\{\bm{z}\}|} \sum_{\bm{z}} \sum_{i = 1}^d \sum_{j = 1}^d ({s_j w_{ij}\over \sigma_i})^2
= \sum_{i = 1}^d \sum_{j = 1}^d ({s_j w_{ij}\over \sigma_i})^2.
\end{aligned}
\end{equation}
Here, in the last line, we use (\ref{equ:crossToZero}) and the fact that $\sum_{\bm{z}} z_j z_k = 0$ for any $j \not = k$, to remove any crossing term that contains $z_j z_k$.
Since $\sum \sigma_i^2 = 1$, we can use H\"older inequality to show that
\begin{equation}\label{equ:sumT4}
\begin{aligned}
\mathcal{L}(W, \bm{\sigma}, \mathsf{S})
& \ge \sum_{i = 1}^d \sum_{j = 1}^d ({s_j w_{ij}\over \sigma_i})^2 \\
& = \Big(\sum_{i = 1}^d {\sum_{j = 1}^d s^2_j w^2_{ij}\over \sigma^2_i}\Big)\cdot (\sum_{i=1}^d \sigma_i^2) \\
& \ge \Big(\sum_{i = 1}^d \sqrt{\sum_{j = 1}^d s^2_j w^2_{ij}}\Big)^2.
\end{aligned}
\end{equation}
Lemma \ref{lem:stochasticConcave} implies that for any $d \times d$ doubly stochastic matrix $M$, any concave function $f$ and any non-negative $a_1\ge a_2 \ge \dotsb \ge a_d \ge 0$,
\[
\sum_{i = 1}^d f(\sum_{j = 1}^d a_j M_{ij}) \ge \sum_{i = 1}^d f(a_i).
\]
We apply this to (\ref{equ:sumT4}) where we set matrix $M$ such that $M_{ij} = w_{ij}^2$, set $f(x) = \sqrt{x}$ and $a_j = s^2_j$.
Since $\bm{w}_1, \dotsb, \bm{w}_d$ form a basis, the matrix $M$ where $M_{ij} = w_{ij}^2$ is doubly stochastic. 
This implies that 
\[
\sum_{i = 1}^d \sqrt{\sum_{j = 1}^d s^2_j w^2_{ij}} \ge \sum_{i = 1}^d |s_i|.
\]
Therefore, $\mathcal{L}(W, \sigma, S) \ge (\sum_{i = 1}^d |s_i|)^2$.
Since this holds for all $\bm{s}\in S$, we have
\[
\mathcal{L}(W, \bm{\sigma}, \mathsf{S}) \ge \max_{\bm{s}\in S}(\sum_{i = 1}^d |s_i|)^2 = (\sum_{i = 1}^d V_i)^2.
\]
Note that when we select $\bm{w}_i = \bm{u}_i$ and set
\[
\sigma_i = \sqrt{V_i \over V_1 + \dotsb + V_d},
\]
$\mathcal{L}(W, \sigma, S)$ is exactly $(\sum_{i = 1}^d V_i)^2$.
It also implies that this is the optimal noise.
\end{proof}

\section{Proof of Theorem \ref{thm: hybrid clipping}}
\label{app: thm: hybrid clipping}
\OptimalHybrid*
\begin{proof}
We first make use of the property that $\mathsf{S}$ is invariant to sign, and is invariant to permutation in each subspace.
For any $\bm{s} = (\bm{s}_1, \bm{s}_2, \dotsb, \bm{s}_m)\in \mathsf{S}$, any $\bm{z}_1, \dotsb, \bm{z}_m$ and $\pi_1, \dotsb, \pi_m$ where 
\begin{itemize}
\item $\bm{s}_i \in \mathbb{R}^{r_i}$ is the coordinate block (sub vector) in the $i$-{th} subspace,
\item $\bm{z}_i\in \{-1, 1\}^{r_i}$ is a sign vector and
\item $\pi_i$ is a permutation on $\{1, \dotsb, r_i\}$.
\end{itemize}
We use $T(x)$ to denote the set
\[
\{(t(\bm{s}_1, \bm{z}_1, \pi_1), \dotsb, t(\bm{s}_m, \bm{z}_m, \pi_m)) ~|~ \bm{z}_1, \dotsb, \bm{z}_m, \pi_1, \dotsb, \pi_m\}.
\]
Recall that the $t(\bm{s}_i, \bm{z}_i, \pi_i)$ function means applying the sign vector $\bm{z}_i$ and the permutation $\pi_i$ on $\bm{s}_i$.
It is formally defined in (\ref{equ:defnl}).
By definition, for any $\bm{s}\in \mathsf{S}$, $T(\bm{s})\subseteq \mathsf{S}$.

Consider a basis $U$, for convenience, we separate each basis vector $\bm{u}_i$ according to the subspace division of $\mathsf{S}$.
Specifically, suppose $\bm{u}_i = (\bm{v}_{i1}, \dotsb, \bm{v}_{im})$ where $\bm{v}_{ij}$ is of dimension $r_j$.
Summing $\mathcal{L}(U, \sigma, \bm{s}')$ over all $\bm{s}'\in T(\bm{s})$, we have
\begin{equation*}
\begin{aligned}
\mathcal{L}(U, \sigma, \mathsf{S})
& \ge {1\over |T(\bm{s})|} \sum_{\bm{s}'\in T(s)} \mathcal{L}(U, \sigma, \bm{s}') \\
& = {1\over |T(\bm{s})|} \sum_{\bm{z}, \pi} \sum_{i=1}^d \sum_{j=1}^m ({<t(\bm{s}_j, \bm{z}_j, \pi_j), \bm{v}_{ij}>\over \sigma_i})^2 \\
& = \sum_{j=1}^m {1\over |\{(\bm{z}_j, \pi_j)\}|} \sum_{\bm{z}_j, \pi_j} \sum_{i=1}^d ({<t(\bm{s}_j, \bm{z}_j, \pi_j), \bm{v}_{ij}>\over \sigma_i})^2
\end{aligned}
\end{equation*}
Here, we use the fact that
\[
|T(\bm{s})| = \prod_{j=1}^m |\{(\bm{z}_j, \pi_j)\}|.
\]
In (\ref{equ:sumT2}) and (\ref{equ:sumT20}), we showed that 
\[
{1\over |\{(\bm{z}_j, \pi_j)\}|} \sum_{\bm{z}_j, \pi_j} \sum_{i=1}^d ({<t(\bm{s}_j, \bm{z}_j, \pi_j), \bm{v}_{ij}>\over \sigma_i})^2 = {\|\bm{s}_j\|_2^2\over r_j}\cdot \sum_{i = 1}^d {\|\bm{v}_{ij}\|_2^2\over \sigma^2_i}.
\]
Using H\"older inequality and $\sum_{i = 1}^d \sigma^2_i = 1$, we have 
\begin{equation}
\begin{aligned}
\mathcal{L}(U, \sigma, S)
& \ge \sum_{i = 1}^d \sum_{j = 1}^m {\|\bm{s}_j\|_2^2 \cdot \|\bm{v}_{ij}\|_2^2 \over r_j \sigma^2_i} \\
& = (\sum_{i = 1}^d \sum_{j = 1}^m {\|\bm{s}_j\|_2^2 \cdot \|\bm{v}_{ij}\|_2^2 \over r_j \sigma^2_i}) \cdot (\sum_{i = 1}^d \sigma^2_i) \\
& = \Big(\sum_{i = 1}^d \sqrt{\sum_{j = 1}^m {\|\bm{s}_j\|_2^2 \cdot \|\bm{v}_{ij}\|_2^2 \over r_j}}\Big)^2.
\end{aligned}
\end{equation}
The last term is very similar to (\ref{equ:sumT4}).
However, we cannot directly apply Lemma \ref{lem:stochasticConcave} to it since $m\neq d$ and $[\|\bm{v}_{ij}\|_2^2]_{ij}$ is not a doubly stochastic matrix.
To fix this, we define $a_1, \dotsb, a_d$ such that 
\[
a_{i} = {\|\bm{s}_j\|_2 \over \sqrt{r_j}} \text{~~for all~~} i \in (\sum_{k = 1}^{j - 1} r_k, \sum_{k = 1}^{j} r_k].
\]
For convenience, we denote the range $(\sum_{k = 1}^{j - 1} r_k, \sum_{k = 1}^{j} r_k]$ as $R_j$.
Note that $a_k = a_l$ for all $k, l\in R_j$.
Therefore,
\[
{\|\bm{s}_j\|_2^2 \cdot \|\bm{v}_{ij}\|_2^2 \over r_j}
=
{\|\bm{s}_j\|_2^2\over r_j}\cdot \sum_{l\in R_j} u^2_{il}
=
\sum_{l\in R_j} a_l^2 u^2_{il},
\]
where recall that $u_{ij}$ is the $j$-th coordinate in $\bm{u}_i$.
In this way, we can rewrite (\ref{equ:sumT4}) into
\[
\mathcal{L}(U, \sigma, S) \ge \Big(\sum_{i = 1}^d \sqrt{\sum_{l = 1}^d a^2_j u^2_{ij}}\Big)^2.
\]
Now, we can apply Lemma \ref{lem:stochasticConcave}, which implies that 
\[
\mathcal{L}(U, \sigma, S) \ge (\sum_{l = 1}^d a_l)^2 = (\sum_{j = 1}^m \|\bm{s}_j\|_2\sqrt{r_j})^2.
\]
This holds for all $\bm{s} = (\bm{s}_1, \bm{s}_2, \dotsb, \bm{s}_m)\in S$.
Each $\|\bm{s}_j\|_2$ is upper bounded by $c_{2j}$.
Therefore, $\mathcal{L}(U, \sigma, S) \ge \sum_{j = 1}^m c_{2j}\sqrt{r_j}$. 
The equation holds when we use the original bases and set
\[
\sigma_i = \sqrt{c_{2j} \over \sqrt{r_j}\sum_{k = 1}^m c_{2k}\sqrt{r_k}}
\]
for all $i\in R_j$.
\end{proof}



\section{Proof of Theorem \ref{thm: privacy coordiante sampling}}
\label{app: thm: privacy coordiante sampling} 
Before start, we first prove a useful lemma.  
\begin{lemma} For arbitrary two positive differential convex functions $f(x)$ and $g(x)$, if both $\log(f(x))$ and $\log(g(x))$ are convex, then $\log(f(x)+g(x))$ is also convex. 
\label{lemma: log convex}
\end{lemma}
\begin{proof}
By definition, $\log(f(x))$ is convex iff 
\[
(\log(f(x)))'' = \frac{f''(x)f(x)-(f'(x))^2}{(f(x))^2} \geq 0.
\]
This suggests that $f''(x)f(x) \geq (f'(x))^2$ for any $x$, and similarly, $g''(x)g(x) \geq (g'(x))^2$. Now, we calculate the second derivative of $\log(f(x)+g(x))$, which equals $$\frac{(f''(x)+g''(x))(f(x)+g(x)) -(f'(x)+g'(x))^2}{(f(x)+g(x))^2}.$$ To show $\log(f(x)+g(x))$ is convex, it suffices to show
$$ (f''(x)+g''(x))(f(x)+g(x)) \geq  (f'(x)+g'(x))^2.$$
On the other hand,
\begin{equation}
    \begin{aligned}
        ~~~&(f''(x)+g''(x))(f(x)+g(x)) -  (f'(x)+g'(x))^2 \\
    =
        ~~~& (f''(x)f(x) - f'(x)^2) + (g''(x)g(x) - g'(x)^2) \\
        ~~~& \quad + f''(x)g(x)+g''(x)f(x) -2f'(x)g'(x) \\
    \geq 
        ~~~&f''(x)g(x)+g''(x)f(x) - 2f'(x)g'(x)\\
    \geq 
        ~~~& 2\sqrt{f''(x)f(x)g''(x)g(x)} - 2f'(x)g'(x) \geq 0.
    \end{aligned}
\label{AM-GM f,g}
\end{equation}
In the last line of (\ref{AM-GM f,g}), we use the fact that $f''(x),g''(x) \geq 0$ due to the convexity assumption and the AM-GM inequality. 
\end{proof}

Now, we are ready to study the RDP of Algorithm \ref{alg: coordinate sampling}.
\PriAmpCoor*
\begin{proof}
For any two adjacent datasets $X$ and $X'$ where without loss of generality $X' = X \cup x$ and $X$ is of $n$ elements, let $\mathcal{J}=\{J_1, J_2, \cdots, J_{2^{n}}\}$ be set of all the subsets of $X$ where $p_j$ is the probability that $J_j$ is selected under $q$-Poisson sampling on $X$. 
We use $\mathcal{F}^{CS}$ to denote Algorithm \ref{alg: coordinate sampling}. 
It is noted that both the sampling and noise in each dimension is independent and thus each coordinate of $\mathcal{F}^{CS}(X)$($\mathcal{F}^{CS}(X')$) is independently generated. 
Therefore, the $\alpha$-Rényi divergence $\mathcal{D}_{\alpha}$ between $\mathcal{F}^{CS}(X')$ and $\mathcal{F}^{CS}(X)$ can be written as 
\begin{equation}
\begin{aligned}
    &\mathcal{D}_{\alpha}(\mathbb{P}_{\mathcal{F}^{CS}(X')}\| \mathbb{P}_{\mathcal{F}^{CS}(X)})\\ 
    &= \frac{1}{\alpha-1}\log \int_{\bm{o}} \mathbb{P}_{\mathcal{F}^{CS}(X)}(\bm{o}) \cdot \big( \frac{\mathbb{P}_{\mathcal{F}^{CS}(X')}(\bm{o})}{\mathbb{P}_{\mathcal{F}^{CS}(X)}(\bm{o})} \big)^{\alpha} d \bm{o}\\
    & = \frac{1}{\alpha-1}\log \int_{\bm{o}} \prod_{l=1}^d\mathbb{P}_{\mathcal{F}^{CS}(X)(l)}(o_l) \cdot \big( \frac{\prod_{l=1}^d \mathbb{P}_{\mathcal{F}^{CS}(X')(l)}(o_l)}{\prod_{l=1}^d \mathbb{P}_{\mathcal{F}^{CS}(X)(l)}({o}_l)} \big)^{\alpha} d \bm{o}\\
    & =  \frac{1}{\alpha-1}\log \prod_{l=1}^d \int_{o_l}  \frac{\big(\mathbb{P}_{\mathcal{F}^{CS}(X')(l)}(o_l)\big)^{\alpha}}{\big(\mathbb{P}_{\mathcal{F}^{CS}(X)(l)}({o}_l)\big)^{\alpha-1}}  d o_l\\
    & = \sum_{l=1}^d \mathcal{D}_{\alpha}(\mathbb{P}_{\mathcal{F}^{CS}(X')(l)}\| \mathbb{P}_{\mathcal{F}^{CS}(X)(l)}),
\end{aligned}
\label{coordiante decomposition}
\end{equation}
where $\bm{o}=(o_1,\cdots, o_d)$. One may also obtain a similar form of $\mathcal{D}_{\alpha}(\mathbb{P}_{\mathcal{F}^{CS}(X)}\| \mathbb{P}_{\mathcal{F}^{CS}(X')})$. In (\ref{coordiante decomposition}), $\mathbb{P}_{\mathcal{F}^{CS}(X)(l)}$ is the density function of the $l$-th coordinate of $\mathcal{F}^{CS}(X)$. Thus, due the independence, from  (\ref{coordiante decomposition}) we know that the RDP analysis of Algorithm \ref{alg: coordinate sampling} is equivalent to studying the sum of coordinate-wise Rényi divergence. Now, we consider the sensitivity set of $\mathcal{F}^{CS}$. Let $s_{l} = |\mathcal{CP}(\mathcal{F}(x))(l)|$ be the difference in the $l$-th coordinate when we happen to select the differing datapoint $x$ in the processing. Thus, the distribution of $\mathcal{F}^{CS}(X)$
is indeed a Gaussian mixture model, where 
\begin{equation}
\mathbb{P}_{\mathcal{F}^{CS}(X)(l)} = \sum_{j=1}p_j \mathcal{N}(\mathcal{F}(J_j)(l), \sigma^2). 
\end{equation}
Similarly, for $\mathcal{F}^{CS}(X')$, it is noted that $J_j$ and $J_j \cup x$ will be selected from $X'$ with probability $p_j(1-q)$ and $p_jq$, respectively, and thus
\begin{equation}
\mathbb{P}_{\mathcal{F}^{CS}(X')(l)} = \sum_{j}  p_j\big((1-q) \mathcal{N}(\mathcal{F}(J_j)(l), \sigma^2) + q\mathcal{N}(\mathcal{F}(J_j)(l)+s_{l}, \sigma^2) \big).  
\end{equation}
By the quasi-convexity of Rényi divergence \cite{van2014renyi,mironov2019r}, we have the following upper bound on the divergence between two mixture distributions by the maximal divergence between their components,  
\begin{equation}
\begin{aligned}
   &\mathcal{D}_{\alpha}(\mathbb{P}_{\mathcal{F}^{CS}(X')(l)}\| \mathbb{P}_{\mathcal{F}^{CS}(X)(l)}) \\
   &\leq \max_{j} \mathcal{D}_{\alpha}\big((1-q) \mathcal{N}(\mathcal{F}(J_j)(l), \sigma^2) + q\mathcal{N}(\mathcal{F}(J_j)(l)+s_{l}, \sigma^2)\\
   & \quad \quad \quad \quad \quad \quad \|\mathcal{N}(\mathcal{F}(J_j)(l), \sigma^2) \big)\\
   & =  \mathcal{D}_{\alpha}\big((1-q) \mathcal{N}(0, \sigma^2) + q\mathcal{N}(s_{l}, \sigma^2) \|\mathcal{N}(0, \sigma^2) \big).
\end{aligned}
\label{quasi convexity}
\end{equation}
Thus, plugging (\ref{quasi convexity}) back to (\ref{coordiante decomposition}), we have that 
\begin{equation}
\begin{aligned}
    &\mathcal{D}_{\alpha}(\mathbb{P}_{\mathcal{F}^{CS}(X')}\| \mathbb{P}_{\mathcal{F}^{CS}(X)})\\ 
    & \leq \sum_{l=1}^d \mathcal{D}_{\alpha}\big((1-q) \mathcal{N}(0, \sigma^2) + q\mathcal{N}(s_{l}, \sigma^2) \|\mathcal{N}(0, \sigma^2) \big) 
\end{aligned}
\label{coordiante decomposition 1}
\end{equation}
Based on our assumption on sensitivity set $\mathsf{S}$ where any $\bm{s} \in \mathsf{S}$ satisfies 
$$ \max_{l} |s_l|<c_{\infty}, \sum_{l=1}^d |s_l|^p \leq (c_p)^p.$$
On the other hand, RDP on one-dimensional subsampled Gaussian mechanism is a known result and has a closed form \cite{mironov2019r}, where 
\begin{equation}
\begin{aligned}
 &\mathcal{D}_{\alpha}\big((1-q) \mathcal{N}(0, \sigma^2) + q\mathcal{N}(s_{l}, \sigma^2) \|\mathcal{N}(0, \sigma^2) \big)\\
 & = \frac{1}{\alpha-1}\log\big((1-q)^{\alpha-1}(\alpha q-q+1) + \sum_{v=2}^{\alpha}\binom{\alpha}{v}(1-q)^{\alpha-v}q^{v}e^{\frac{v(v-1)s^2_l}{2\sigma^2}} \big). 
\end{aligned}
    \label{RDP Gaussian}
\end{equation}
Moreover, it is also proved in \cite{mironov2019r} that 
\begin{align*}
    \mathcal{D}_{\alpha}\big((1-q) \mathcal{N}(0, \sigma^2) & + q\mathcal{N}(s_{l}, \sigma^2) \|\mathcal{N}(0, \sigma^2) \big)\\
    & \geq \mathcal{D}_{\alpha}\big(\mathcal{N}(0, \sigma^2) \| (1-q) \mathcal{N}(0, \sigma^2) + q\mathcal{N}(s_{l}, \sigma^2) \big),
\end{align*}
and thus the above upper bound also works for $\mathcal{D}_{\alpha}(\mathbb{P}_{\mathcal{F}^{CS}(X)}\| \mathbb{P}_{\mathcal{F}^{CS}(X')})$. Now, the remainder problem is to determine the dominating sensitivity for (\ref{coordiante decomposition 1}), which is equivalent to solving the following constraint optimization problem,
\begin{equation}
    \begin{aligned}
        &\sup_{\bm{s}} \sum_{l=1}^d \log\big((1-q)^{\alpha-1}(\alpha q-q+1) + \sum_{v=2}^{\alpha}\underbrace{\binom{\alpha}{v}(1-q)^{\alpha-v}q^{v}e^{\frac{v(v-1)s^2_l}{2\sigma^2}}} \big),\\
        & s.t.  \max_{l} |s_l|<c_{\infty}, \sum_{l=1}^d |s_l|^p \leq (c_p)^p. 
    \end{aligned}
    \label{optimization l_p l_infty}
\end{equation}
It is noted that for each underlined component in the log term of (\ref{optimization l_p l_infty}) can be written in the following form
\begin{align*}
    & b_{1l} \cdot exp(b_{2l} \cdot s^2_l) = b_{1l} \cdot exp\big(b_{2l} \cdot (|s_l|^p)^{2/p} \big) = b_{1l} \cdot exp(b_{2l} \cdot y_l^{2/p}), 
\end{align*}
where $b_{1l}$ and $b_{2l}$ are some positive constants and $y_l = |s_l|^p$. 
Thus, for $p\leq 2$, it can be verified that $\log(b_1 \cdot exp(b_2 \cdot y_l^{2/p}))$ is a convex function with respect to $y_l$, and thus by Lemma \ref{lemma: log convex}, the objective function in (\ref{optimization l_p l_infty}) is also convex with respect to each $y_l$. On the other hand, $\bm{y}= (y_1, \cdots, y_d) = (|s_1|^p, \cdots, |s_d|^p)$,  given $\sum_{l=1}^d |s_l|^p=(c_p)^p$ and $\max_l |s_l|^p \leq (c_{\infty})^p$, is indeed an intersection between an $l_1$ ball and an $l_{\infty}$ ball, which is a polyhedron.

We now use the a folk lemma that the maximum of a convex function on a convex domain must be reached at the boundary.
And if the domain is a polyhedron, then the maximum must be reached at the vertices.
Excluding the trivial vertex at the zeros, the remaining vertices $\bm{y}$ are all in a form $$(\underbrace{c^p_{\infty}, \cdots, c^p_{\infty}}_{d_0},\underbrace{0,\cdots,e0}_{d-d_0}),$$ with permutation on the coordinates. 
Thus, we transform $\bm{y}$ back to $\bm{s}$ and we have determined the dominating sensitivity and the theorem follows.    
\end{proof}

\section{Proof of Theorem \ref{thm: coordinate-wise-sampling improvement}}
\AsympPrivCoor*
\begin{proof}
\label{app: thm: coordinate-wise-sampling improvement}
With the notation that $\tau = (\frac{c_p}{\sqrt{2}\sigma})^2$, the $(\alpha, \epsilon(\alpha))$ bound of Algorithm \ref{alg: coordinate sampling} in (\ref{coordinate-sample-bound}) can be rewritten as follows, 
\begin{equation} 
\begin{aligned}
&\epsilon(\alpha) \\
& = \frac{d_0}{\alpha-1}\cdot \log\big((1-q)^{\alpha-1}(\alpha q-q+1) + \sum_{v=2}^{\alpha}\binom{\alpha}{v}(1-q)^{\alpha-v}q^{v}e^{\frac{v(v-1)c^2_{\infty}}{2\sigma^2}} \big) \\
&= \frac{d_0}{\alpha-1}\cdot \log\big((1-q)^{\alpha-1}(\alpha q-q+1) + \sum_{v=2}^{\alpha}\binom{\alpha}{v}(1-q)^{\alpha-v}q^{v}e^{\frac{v(v-1)\tau}{d_0}} \big)\\
& = \frac{d_0}{\alpha-1}\cdot \log\big((1-q+q)^{\alpha}+ \sum_{v=2}^{\alpha}\binom{\alpha}{v}(1-q)^{\alpha-v}q^{v}(e^{\frac{v(v-1)\tau}{d_0}}-1) \big)
\end{aligned}
\label{l(l-1) regroup}
\end{equation} 
Thus, when $d_0 \to \infty$, $e^{\frac{v(v-1)\tau}{d_0}}$ will approach towards $1+\frac{v(v-1)\tau}{d_0}$ and thus (\ref{l(l-1) regroup}) convergences to 
\begin{equation}
 \lim_{d_0 \to \infty} \epsilon(\alpha) = \frac{d_0}{\alpha-1}\cdot \log\big(1+ \sum_{v=2}^{\alpha}\binom{\alpha}{v}(1-q)^{\alpha-v}q^{v}\cdot\frac{v(v-1)\tau}{d_0} \big).
\label{l(l-1) regroup 1} 
\end{equation}
It is noted that 
\[
\binom{\alpha}{v} = \frac{\alpha(\alpha-1)\cdots(\alpha-v+1)}{v(v-1)\cdots 1} = \frac{\alpha(\alpha-1)}{v(v-1)}\cdot \binom{\alpha-2}{v-2}.
\]
Therefore, (\ref{l(l-1) regroup 1}) can be rewritten as follows when $d_0 \to \infty$,  
\begin{equation*}
\begin{aligned}
& ~~  \lim_{d_0 \to \infty} \epsilon(\alpha) \\
= & ~~  \lim_{d_0 \to \infty}  \frac{d_0}{\alpha-1}\cdot \log\big(1+ \alpha(\alpha - 1)\sum_{v=2}^{\alpha}\binom{\alpha - 2}{v - 2}(1-q)^{\alpha-v}q^{v}\cdot {\tau\over d_0}\big)  \\
= & ~~  \lim_{d_0 \to \infty}  \frac{d_0}{\alpha-1}\cdot \log\big(1+ \alpha(\alpha-1)q^2 \frac{\tau}{d_0} (1-q + q)^{\alpha-2}\big)  \\
= & ~~  \alpha q^2\tau.  
\end{aligned}
\label{l(l-1) regroup 2} 
\end{equation*}
In the following, we focus on the case when $d_0 =1$ and $c_{\infty}=c_p$. It is noted that for any $v \in [2,\alpha]$,  
\[
\binom{\alpha}{v}(1-q)^{\alpha-v}q^{v}(e^{v(v-1)\tau}-1) \leq \alpha^v q^{v} e^{{v(v-1)\tau}} = (\alpha q e^{(v-1)\tau})^v.
\]
Thus, when $q^{1/6} \leq 1/(\alpha e^{\alpha \tau})$, for any $v\ge 3$,
\[
(\alpha q e^{(v-1)\tau})^v
\le
q^v\cdot (\alpha e^{\alpha\tau})^v
\le
q^{2.5} \cdot q^{v/6}\cdot (\alpha e^{\alpha\tau})^v
\le 
q^{2.5}.
\]
Therefore,
$$\sum_{v=2}^{\alpha}\binom{\alpha}{v}(1-q)^{\alpha-v}q^{v}(e^{{v(v-1)\tau}}-1) \leq \frac{\alpha(\alpha-1)}{2}q^2(e^{\tau}-1) + \alpha q^{2.5},$$ and thus from (\ref{l(l-1) regroup}) we have that 
$$ \lim_{q \to 0}\epsilon(\alpha) = O\big(\frac{1}{\alpha-1}\cdot \log\big( 1+ \alpha^2 q^2 (e^{\tau}-1) \big) \big) = O(\alpha q^2 (e^{\tau}-1)),$$
On the other hand, we have that 
$$ \epsilon(\alpha) \geq \frac{1}{\alpha-1}\cdot \log\big((1-q+q)^{\alpha}+ \frac{\alpha(\alpha-1)}{2}(1-q)^{\alpha-2}q^{2}(e^{\tau}-1)\big),$$
whose limit is $O(\alpha q^2 (e^{\tau}-1))$ as $q \to 0$. Thus, when $d_0=1$, $\epsilon = \Theta(\alpha q^2 (e^{\tau}-1))$. 

When $\alpha(\alpha-1)\tau \geq 2$ and $q$ is large such that $q \geq 1/(e^{(\alpha-1)\tau/2})$, then $qe^{(\alpha-1)\tau} \geq e^{(\alpha-1)\tau/2}$.
Take (\ref{l(l-1) regroup 1}) and only consider the last term in the summation (when $v = \alpha$), we have 
\begin{equation}
\begin{aligned}
    \epsilon(\alpha) &\geq \frac{1}{\alpha-1}\cdot \log \big( 1+ q^{\alpha}(e^{\alpha(\alpha-1)\tau}-1)\big)\\ 
    & \geq \frac{1}{\alpha-1}\cdot \log \big( 1+ 0.5 \cdot q^{\alpha}e^{\alpha(\alpha-1)\tau}\big)\\
    & =  \frac{1}{\alpha-1}\cdot \log \big( 1+ 0.5 \cdot (q\cdot e^{(\alpha-1)\tau})^{\alpha}\big)\\
    & = \Omega(\alpha\tau). 
\end{aligned}
\end{equation}
In the second line, we use our assumption that $\alpha(\alpha - 1)\tau \geq 2$.
This completes our proof.
\end{proof}

\section{Proof of Theorem \ref{thm: sampling twice}}
\label{app: thm: sampling twice}
\PrivAmplTS*
\begin{proof}
Twice-sampling is essentially a composition of two sampling subroutines, which, to be specific, forms by a Poisson sampling on sample dimension followed by a coordinate-wise sampling. 
As the first step, we need to characterize the mixture output distribution from twice sampling. With similar notations as those used in Appendix \ref{app: thm: privacy coordiante sampling}, we suppose two adjacent datasets $X$ of $n$ datapoints and $X'=X \cup x$ where $x$ is the differing datapoint. Let $(s_1, s_2, \cdots, s_d)= \mathcal{CP}(\mathcal{F}(x))$ be the clipped processing on the differing datapoint $x$.  
Since we conduct two samplings on input data and coordinate, each sampled instance is different from that in Algorithm \ref{alg: coordinate sampling}. 
In the following, we introduce a set of indicators $\bm{1}_{1i}$ and $\bm{1}_{2i}(l)$, for $i=1,2,\cdots, n$, (denoted as $\bm{1}_{1x}$ and $\bm{1}_{2x}(l)$ for the differing datapoint $x$). 
$\bm{1}_{1i}$ and $\bm{1}_{2i}(l)$ are independent Bernoulli variables of parameter $q_1$ and $q_2$.
\begin{itemize}
\item $\bm{1}_{1i}$ equals $1$ if and only if sample $i$ is selected in the first round of input-level sampling, and
\item $\bm{1}_{2i}(l)$ equals $1$ if and only if the $l$-{th} coordinate of sample $i$ is selected in the second round of coordinate processing. 
\end{itemize}
For the common set part between $X$ and $X'$, we use $\bm{I}_C$ to denote the selection of  $\bm{1}_{1i}$ and $\bm{1}_{2i}(l)$, for $i=1,2,\cdots,n$ and $l=1,2,\cdots,d$. 
Similarly, $\bm{I}_x$ is used to denote the indicators for $x$.
We use $p(\bm{I}_C)$ to represent the probability that $\bm{I}_C$ is selected by running twice sampling on $X$. 
Similarly, we can define $p(\bm{I}_C, \bm{I}_x)$ when running twice sampling on $X'$. 

We use $\tilde{\mathcal{F}}$ to represent Algorithm \ref{alg: sampling twice}, the privatized $\mathcal{F}$ combined with twice-sampling where each coordinate is perturbed by an independent Gaussian noise distributed in $\mathcal{N}(0,\sigma^2)$. 
With a slight abuse of the notation, we use $\mu_0(\bm{I}_C)$ to denote the distribution of applying $\tilde{\mathcal{F}}$ on $X$ where the selection from twice-sampling is determined as $\bm{I}_C$. 
We want to stress that even after $\bm{I}_C$ is determined, the final output still depends on the noise.
Therefore, $\mu_0(\bm{I}_C)$ is a distribution, not a fixed output.
Similarly, we use $\mu_1(\bm{I}_C)$ to denote the distribution of applying $\tilde{\mathcal{F}}$ on $X'$ given that $x$ is selected in the input sampling.
Here, the final output not only depends on the noise, but also $\bm{I}_x$.
Therefore, $\mu_1(\bm{I}_C)$ is also a distribution.

Given $(\bm{I}_C, \bm{I}_x)$, when we apply $\tilde{\mathcal{F}}$ to $X$, the result is
\begin{equation}\label{equ:PCI}
P(\bm{I}_C, \bm{I}_x) = \mu_0(\bm{I}_C).
\end{equation}
Note that $x\notin X$ and thus $\bm{I}_x$ does not affect the final results.
On the other hand, when we apply $\tilde{\mathcal{F}}$ to $X'$, the result is
\begin{equation}\label{equ:QCI}
Q(\bm{I}_C, \bm{I}_x) = (1 - q_1)\mu_0(\bm{I}_C) + q_1\mu_1(\bm{I}_C).
\end{equation}
This is because $Q(\bm{I}_C, \bm{I}_x)$ equals $\mu_0(\bm{I}_C)$ conditioned on $\bm{1}_{1x}=0$ and equals $\mu_1(\bm{I}_C)$ conditioned on $\bm{1}_{1x}=1$.

With these notations, the distribution of $\tilde{\mathcal{F}}(X)$ can be written in a mixture form  
$\sum_{\bm{I}_C} p(\bm{I}_C) \mu_0(\bm{I}_C).$
Similarly, for $\tilde{\mathcal{F}}(X')$, the distribution of  $\tilde{\mathcal{F}}(X)$ can be expressed as 
$\sum_{\bm{I}_C} p(\bm{I}_C)\big( (1-q_1)\mu_0(\bm{I}_C) + q_1\mu_1(\bm{I}_C) \big).$
Then, by definition
\begin{equation}
    \begin{aligned}
        & e^{(\alpha-1)\mathcal{D}_{\alpha}(\mathbb{P}_{\tilde{F}(X)}\| \mathbb{P}_{\tilde{F}(X')})} \\ 
        & = \int_{\bm{o}} \frac{ \big(\sum_{\bm{I}_C} p(\bm{I}_C) \cdot \mu_0(\bm{I}_C)(\bm{o}) \big)^{\alpha} }{\big(\sum_{\bm{I}_C} p(\bm{I}_C)\big((1-q_1)\mu_0(\bm{I}_C)(\bm{o}) + q_1\mu_1(\bm{I}_C)(\bm{o})\big)\big)^{\alpha-1}} 
        d\bm{o}.
    \end{aligned}
    \label{generic sampling RDP 0}
\end{equation}
For $\mathcal{D}_{\alpha}( \mathbb{P}_{\tilde{F}(X')}\| \mathbb{P}_{\tilde{F}(X)})$, by Jensen inequality \cite{zhu2019poission}
\begin{equation}
    \begin{aligned}
        & e^{(\alpha-1)\mathcal{D}_{\alpha}( \mathbb{P}_{\tilde{F}(X')}\| \mathbb{P}_{\tilde{F}(X)})} \\ 
        & = \int_{\bm{o}} \frac{\big(\sum_{\bm{I}_C} p(\bm{I}_C)\big((1-q_1)\mu_0(\bm{I}_C)(\bm{o}) + q_1\mu_1(\bm{I}_C)(\bm{o})\big)\big)^{\alpha}}{ \big(\sum_{\bm{I}_C} p(\bm{I}_C) \cdot \mu_0(\bm{I}_C)(\bm{o}) \big)^{\alpha-1} } 
        d\bm{o}\\
        & \leq \sum_{\bm{I}_C} p(\bm{I}_C) \mathbb{E}_{\mu_0(\bm{I}_C)} \big( \frac{(1-q_1)\mu_0(\bm{I}_C) + q_1\mu_1(\bm{I}_C)}{\mu_0(\bm{I}_C)}\big)^{\alpha}\\
        & = \sum_{\bm{I}_C} p(\bm{I}_C) \mathbb{E}_{\mu_0(\bm{I}_C)} \big( (1-q_1) + q_1 \cdot \frac{\mu_1(\bm{I}_C)}{\mu_0(\bm{I}_C)}\big)^{\alpha}\\
        & = \sum_{\bm{I}_C}p(\bm{I}_C)\big(\sum_{v=0}^{\alpha}\binom{\alpha}{v}(1-q_1)^{\alpha-v}q_1^{v}\mathbb{E}_{\mu_0(\bm{I}_C)}(\frac{\mu_1(\bm{I}_C)}{\mu_0(\bm{I}_C)})^{v}\big).
    \end{aligned}
    \label{generic sampling RDP 1}
\end{equation}
It is noted that in Theorem \ref{thm: privacy coordiante sampling} we have already studied and provided the upper bound of $\mathbb{E}_{\mu_0(\bm{I}_C)}(\frac{\mu_1(\bm{I}_C)}{\mu_0(\bm{I}_C)})^{v}$. 
For any fixed $\bm{I}_C$, i.e.,  $\mu_0(\bm{I}_C)$ corresponds to a Gaussian distribution, where each coordinate of the mean is determined by the selected samples in $\bm{I}_C$. 
On the other hand, the $l$-th coordinate of $\mu_1(\bm{I}_C)$ is independently distributed in a Gaussian mixture in a form $$(1-q_2)\mathcal{N}(\mathbb{E}[\mu_0(\bm{I}_C)(l)]  ,\sigma^2) + q_2\mathcal{N}(\mathbb{E}[\mu_0(\bm{I}_C)(l)]+s_l ,\sigma^2).$$
Thus, it is exactly reduced to the coordinate-sampling scenario and we have 
\begin{equation}
    \mathbb{E}_{\mu_0(\bm{I}_C)}(\frac{\mu_1(\bm{I}_C)}{\mu_0(\bm{I}_C)})^{v} \leq e^{(v-1)\cdot \epsilon(v)},
\end{equation}
as we assume that the provided the Gaussian noise coordinate-wise sampling achieves $(v,\epsilon(v))$-RDP.  

The more tricky part is to upper bound (\ref{generic sampling RDP 0}). To handle this, we borrow the decomposition idea in \cite{zhu2019poission}.
We first combine (\ref{equ:PCI}) and (\ref{equ:QCI}) and observe that
\[
P(\bm{I}_C, \bm{I}_x) = Q(\bm{I}_C, \bm{I}_x) + q_1\mu_0(\bm{I}_C) - q_1\mu_1(\bm{I}_C).
\]
Therefore, we have
\begin{equation}
    \begin{aligned}
        & e^{(\alpha-1)\mathcal{D}_{\alpha}(\mathbb{P}_{\tilde{F}(X)}\| \mathbb{P}_{\tilde{F}(X')})} \\ 
        & = \mathbb{E}_{Q(\bm{I}_C, \bm{I}_x)} \big({P(\bm{I}_C, \bm{I}_x) \over Q(\bm{I}_C, \bm{I}_x)}\big)^{\alpha} \\
        & = \mathbb{E}_{Q(\bm{I}_C, \bm{I}_x)}[\big( \frac{Q(\bm{I}_C, \bm{I}_x)+q_1(\mu_0(\bm{I}_C)-\mu_1(\bm{I}_C))}{Q(\bm{I}_C, \bm{I}_x)}\big)^{\alpha}]\\
        & = q_1\mathbb{E}_{\bm{\bm{I}_C}}\mathbb{E}_{\mu_1}\big[\big( \frac{(1-q_1)\mu_1(\bm{I}_C)+q_1\mu_0(\bm{I}_C)}{\mu_1(\bm{I}_C)}\big)^{\alpha}| \bm{1}_{1x}=1\big]\\
        & ~~~~~~+ (1-q_1)\mathbb{E}_{\bm{\bm{I}_C}}\mathbb{E}_{\mu_0}\big[\big( \frac{(1+q_1)\mu_0(\bm{I}_C)-q_1\mu_1(\bm{I}_C)}{\mu_0(\bm{I}_C)}\big)^{\alpha}| \bm{1}_{1x}=0\big]\\
        & = \mathbb{E}_{\bm{I}_C}\big[q_1\mathbb{E}_{\mu_1}\big[(1-q_1+q_1\cdot \frac{\mu_0}{\mu_1})^{\alpha}\big] \\
        & \quad \quad \quad \quad + (1-q_1) \mathbb{E}_{\mu_0}\big[(1-q_1+q_1(2-\frac{\mu_1}{\mu_0}))^{\alpha}\big]\big]\\
        & = \mathbb{E}_{\bm{I}_C}\big[\sum_{v=0}^{\alpha}\binom{\alpha}{v}(1-q_1)^{\alpha-v}q_1^{v}\big\{ q_1 \mathbb{E}_{\mu_1}[(\frac{\mu_0}{\mu_1})^v] \\
        & \quad \quad \quad \quad + (1-q_1)\mathbb{E}_{\mu_0}[(2-\frac{\mu_1}{\mu_0})^v]\big\}  \big]. 
    \end{aligned}
    \label{generic sampling RDP 0-1}
\end{equation}
Comparing (\ref{generic sampling RDP 1}) and (\ref{generic sampling RDP 0-1}), in the following, we will prove 
\[
e^{(\alpha-1)\mathcal{D}_{\alpha}(\mathbb{P}_{\tilde{F}(X)} \| \mathbb{P}_{\tilde{F}(X')})} \leq 
e^{(\alpha-1)\mathcal{D}_{\alpha}(\mathbb{P}_{\tilde{F}(X')} \| \mathbb{P}_{\tilde{F}(X)})}
\]
by showing that 
$  \mathbb{E}_{\mu_0}[(2-\frac{\mu_1}{\mu_0})^v] \leq \mathbb{E}_{\mu_0}[(\frac{\mu_1}{\mu_0})^v]$.
A similar result was proved in \cite{zhu2019poission}.
We follow their high level idea and provide the proof as follows. 

First, we have
\[
\begin{aligned}
& \mathbb{E}_{\mu_0}[(2-\frac{\mu_1}{\mu_0})^v] - \mathbb{E}_{\mu_0}(\frac{\mu_1}{\mu_0})^v \\
& = \sum_{j = 0}^v {v\choose j} \cdot \Big(\big((-1)^j - 1\big) \cdot \mathbb{E}_{\mu_0}[(\frac{\mu_1}{\mu_0} - 1)^v] \Big) \\
& = - 2\sum_{\text{j is odd},j\leq v} \binom{v}{l} \mathbb{E}_{\mu_0}\big[(\frac{\mu_1}{\mu_0}-1)^j\big].
\end{aligned}
\]
Therefore,
\begin{equation}
     \mathbb{E}_{\mu_0}[(2-\frac{\mu_1}{\mu_0})^v]  = \mathbb{E}_{\mu_0}(\frac{\mu_1}{\mu_0})^v - 2\sum_{\text{j is odd},j\leq v} \binom{v}{l} \mathbb{E}_{\mu_0}\big[(\frac{\mu_1}{\mu_0}-1)^j\big].
\end{equation}
It suffices to prove for any $j$, $\mathbb{E}_{\mu_0}\big[(\frac{\mu_1}{\mu_0}-1)^j\big]\geq 0$. Now, for any given $\bm{I}_C$,  let $\bm{\kappa} = (\kappa_1, \kappa_2, \cdots, \kappa_d) = \mathbb{E}[\mu_0(\bm{I}_C)]$, and from the independent coordinate sampling, we know the $l$-th coordinate of $\mu_1$ is independently distributed in a Gaussian mixture $(1-q_2)\mathcal{N}(\kappa_l,\sigma^2)+q_2\mathcal{N}(\kappa_l+s_l,\sigma^2)$, while that of $\mu_0$ is a pure Gaussian $\mathcal{N}(\kappa_l,\sigma^2)$. 
For simplicity, in the following we use $\mathcal{G}_{1l}$ to denote the probability density function of $\mathcal{N}(\kappa_l+s_l,\sigma^2)$ and  $\mathcal{G}_{0l}$ for that of $\mathcal{N}(\kappa_l,\sigma^2)$. 
With the preparation, $\mathbb{E}_{\mu_0}\big[(\frac{\mu_1}{\mu_0}-1)^j\big]$ can be rewritten as 
\begin{equation}
    \begin{aligned}
        &\mathbb{E}_{\mu_0}\big[(\frac{\mu_1}{\mu_0}-1)^j\big] \\
        &= \int_{\bm{o}} \prod_{l=1}^d \mathcal{G}_{0l}(o_l) (\prod_{l=1}^d \frac{(1-q_2)\mathcal{G}_{0l}(o_l)+q_2\mathcal{G}_{1l}(o_l)}{\mathcal{G}_{0l}(o_l)}-1)^j d\bm{o}\\
        & = \int_{\bm{o}} \prod_{l=1}^d \mathcal{G}_{0l}(o_l)\big(\prod_{l=1}^d ((1-q_2) + q_2\frac{\mathcal{G}_{1l}(o_l)}{\mathcal{G}_{0l}(o_l)})-1\big)^j d\bm{o}\\
        & = \int_{\bm{o}} \prod_{l=1}^d \mathcal{G}_{0l}(o_l)\big(\prod_{l=1}^d (1 + q_2(\frac{\mathcal{G}_{1l}(o_l)}{\mathcal{G}_{0l}(o_l)}-1))-1\big)^j d\bm{o}\\
        & = \int_{\bm{o}} \prod_{l=1}^d \mathcal{G}_{0l}(o_l)\big(\sum_{J \subset \{1,2,\cdots, d\}, J\neq\emptyset} q_2^{d-|J|}\prod_{l \in J} (\frac{\mathcal{G}_{1l}(o_l)}{\mathcal{G}_{0l}(o_l)}-1)\big)^j d\bm{o}\\
        & = \text{Poly}\big(\mathbb{E}_{\mathcal{G}_{0l}}\big (\frac{\mathcal{G}_{1l}}{\mathcal{G}_{0l}}-1)^t\big), l=1,2,\cdots, d, t=0,1,\cdots,j\big). 
    \end{aligned}
    \label{cross product}
\end{equation}
From the second to the last line of (\ref{cross product}), we can see that $\mathbb{E}_{\mu_0}\big[(\frac{\mu_1}{\mu_0}-1)^j\big]$ can be expressed as a polynomial $\text{Poly}(\cdot)$ of the terms $\mathbb{E}_{\mathcal{G}_{0l}}[\big (\frac{\mathcal{G}_{1l}}{\mathcal{G}_{0l}}-1)^t]$ , which is known as the Pearson-Vajda $\chi^t$-pseudo-divergence, which is known to be positive for Gaussian distributions with the same variance (please see Theorem 17 in \cite{zhu2019poission} for the proof). 

Further, the coefficients of $\text{Poly}(\cdot)$ are all positive.
Thus, for any positive integer $j$, $\mathbb{E}_{\mu_0}\big[(\frac{\mu_1}{\mu_0}-1)^j\big]\geq 0$ and (\ref{generic sampling RDP 1}) is a global bound for both cases. 
This completes our proof. 
\end{proof}

\section{Proof of Corollary \ref{cor: twice-sampling improvement}}
\label{app: cor: twice-sampling improvement}

\begin{proof}
    By Theorem \ref{thm: coordinate-wise-sampling improvement}, when $d_0$ is sufficiently large, we know $q_2$-coordinate-wise sampling satisfying $(v, v q^2_2\tau)$-RDP. Plugging it into (\ref{twice-sampling-bound}) in Theorem \ref{thm: sampling twice}, we have that $(q_1, q_2)$-twice sampling in the same setup satisfies $(\alpha, \epsilon(\alpha))$, where
    \begin{equation}
        \epsilon(\alpha) = \frac{\log\big((1-q_1)^{\alpha}+ \sum_{v=1}^{\alpha}\binom{\alpha}{v}(1-q_1)^{\alpha-l}q^{v}_1e^{v(v-1)q^2_2\tau} \big)}{\alpha-1}. 
         \label{pf twice sampling bound}
    \end{equation}    

    With the same trick we used in (\ref{l(l-1) regroup}) in Theorem \ref{thm: coordinate-wise-sampling improvement}, $\epsilon(\alpha)$ can be rewritten as
    \begin{equation}
        \epsilon(\alpha) = \frac{1}{\alpha-1}\cdot \log\big(1 + \sum_{v=2}^{\alpha}\binom{\alpha}{v}(1-q_1)^{\alpha-v}q^{v}_1(e^{{v(v-1)q^2_2\tau}}-1) \big).
        \label{pf twice sampling bound 1}
    \end{equation}
    When $q_2$ is sufficiently small such that $q^2_2 < {{1} \over {\alpha(\alpha-1)\tau}}$, then 
    $$e^{{v(v-1)q^2_2\tau}}-1 \leq  2\big(v(v-1)q^2_2\tau\big).$$
    Therefore,
    \begin{equation*}
    \begin{aligned}
        \epsilon(\alpha) = \frac{1}{\alpha-1}\cdot \log\big(1 + \sum_{v=2}^{\alpha}\binom{\alpha}{v}(1-q_1)^{\alpha-v}q^{v}_1(2{v(v-1)q^2_2\tau}) \big).
        \label{pf twice sampling bound 1}
    \end{aligned}
    \end{equation*}
    Using the fact that
    \[
    {\alpha \choose v} = {\alpha(\alpha - 1)\over v(v-1)}\cdot {\alpha - 2\choose v - 2},
    \]
    We have
    \begin{equation}
    \begin{aligned}
        \epsilon(\alpha) 
        & = \frac{1}{\alpha-1}\cdot \log\big(1 + 2\alpha(\alpha - 1)q_2^2\tau\sum_{v=2}^{\alpha}\binom{\alpha-2}{v-2}(1-q_1)^{\alpha-v}q^{v}_1 \big) \\
        & = \frac{1}{\alpha-1}\cdot \log\big(1 + 2\alpha(\alpha - 1)q_1^2q_2^2\tau \big).
    \label{pf twice sampling bound 2}
    \end{aligned}
    \end{equation}
    Thus, as both $q_1$ and $q_2$ convergences to $0$, we have that $\epsilon(\alpha) = O(\alpha q^2_1q^2_2\tau)$.  
\end{proof}

\section{Algorithm Description of Hybrid Clipping with Twice Sampling}
\label{app: hybrid+twice sampling}
\begin{algorithm}[t]
\caption{{Hybrid Clipping with Twice Sampling}}
\begin{algorithmic}[1]
\STATE \textbf{Input:} We receive the following as inputs.
\begin{enumerate}[leftmargin=*]
\item a processing function $\mathcal{F}(\cdot): \mathcal{X}^* \to \mathbb{R}^d$.
\item sensitive input set $X = \big\{x_1, x_2, ... ,x_n\big\} \in \mathcal{X}^n$ of $n$ datapoints, 
\item Poisson sampling rates for the input-wise $q_1$ and the coordinate-wise $q_2$, 
\item a set of orthogonal unit basis in $m$ subsets where each subset is in a form $\{\bm{u}_{jl}, l=1,2,\cdots, r_j\}$, $\sum_{j=1}^m r_j =1$, 
\item $l_2$-norm clipping thresholds $c_{21}, c_{22}, \cdots , c_{2m}$, 
\item $l_{\infty}$-norm clipping thresholds $c_{\infty 1}, c_{\infty 2}, \cdots c_{\infty m}$ such that $c^2_{\infty j}d_{0j} = c^2_{2j}$ for any $j$, 
\item Gaussian noise parameters $\bm{\sigma}=(\sigma_1, \cdots, \sigma_m)$.  
\end{enumerate}
\STATE Apply Poisson sampling with rate $q_1$ on index $[1:n]$ and obtain an index set $\mathcal{I} = \big\{[1], [2], \cdots, [B]\big\}$ of size $B$. Let $X_{\mathcal{I}} = \{x_{[1]}, \cdots, x_{[B]} \}$.
\STATE Let $U = (\bm{u}_{1,1}, \dotsb, \bm{u}_{1,r_1}, \dotsb, \bm{u}_{m,1}, \dotsb, \bm{u}_{m,r_m})$.
\FOR{$i=1,2,\cdots, B$}
\STATE Compute $y_{[i]} = \mathcal{F}(x_{[i]})\cdot U  = (\bar{y}_{[i]1}, \cdots, \bar{y}_{[i]m})$, where $\bar{y}_{[i]j}$ is the $j$-th segment containing the $(1 + \sum_{w=1}^{j-1}r_w)$-th coordinate to the $(\sum_{w=1}^{j}r_w)$-th coordinate of $y_{[i]}$.  
\STATE Implement $l_2$-norm clipping with parameter $c_{2j}$ on $\bar{y}_{[i]j}$ and obtain
\begin{equation*}
    \begin{aligned}
        \tilde{y}_{[i]} & = \big(\bar{y}_{[i]1}\min\{1, \frac{c_{21}}{\|\bar{y}_{[i]1}\|_2}\}, \cdots, \bar{y}_{[i]m}\min\{1, \frac{c_{2m}}{\|\bar{y}_{[i]m}\|_2}\}\big) \\
        & = (\tilde{v}_{1,1}, \cdots, \tilde{v}_{1,r_1}, \cdots, \tilde{v}_{m,1}, \cdots, \tilde{v}_{m, r_m}).
    \end{aligned}
\end{equation*}
\STATE Implement an additional $l_{\infty}$-norm clipping with parameters $c_{\infty j}$ on the $j$-th segment for $j=1,2,\cdots, m$, and obtain $\vardbtilde{y}_{[i]} = \big(\tilde{v}_{1, 1}\cdot \min\{1, \frac{c_{\infty 1}}{|\tilde{v}_{1, 1}|}\}, \cdots, \tilde{v}_{m, r_m}\cdot \min\{1, \frac{c_{\infty m}}{|\tilde{v}_{m, r_m}|}\} \big).$
\ENDFOR 
\STATE Apply Algorithm \ref{alg: coordinate sampling} with coordinate-wise Poisson sampling with parameter $q_2$ on $\vardbtilde{Y}_{\mathcal{I}} = \{ \vardbtilde{y}_{[i]}, i=1,2, \cdots, B\}$ in the representation under $U$ and the output is $\bm{o} \in \mathbb{R}^d$. 
\STATE Generate $d$ independent Gaussian noises where ${e}_{jl} \sim \mathcal{N}(\bm{0},1)$ and let $\tilde{\bm{e}} = \sum_{j=1}^m\sum_{l=1}^{r_j} \sigma_{j}e_{jl}\cdot \bm{u}_{jl}$.
\STATE \textbf{Output}: return $\bm{o}U^T+\tilde{\bm{e}}$.
\end{algorithmic}
\label{alg: hybrid+twice-sampling}
\end{algorithm}

\begin{algorithm}
\caption{Iterative Principal Space Approximation from Public Samples with Power Method}
\begin{algorithmic}[1]
\STATE \textbf{Generating-Basis}($\mathcal{F}, X^{p}$):
\STATE Apply $\mathcal{F}(\cdot)$ to the public dataset and get $\{y^p_i\} = \{\mathcal{F}(x^p_i)\}$.
\STATE Let $Y^p = (y^p_1, \dotsb, y^p_k)$ be a $d\times k$ matrix.
\FOR {$i = 1, 2, \dotsb, m - 1$}
    \STATE $U_i\leftarrow \text{Approx-Eigen}(Y^p, r_i, t)$.
    \STATE $Y^p\leftarrow Y^p - U_iU_i^{T}Y^{p}$.
\ENDFOR
\STATE Choose $U_m$ such that $U = (U_1, \dotsb, U_m)$ forms an orthogonal basis.
\STATE Return $(U_1, \dotsb, U_m)$ as the basis.
\end{algorithmic}
\hrulefill
\textbf{Subroutine: Approx-Eigen}($M$, $r$, $t$)
\begin{algorithmic}[1]
\STATE Suppose $M$ is a $d\times k$ matrix, first generate a random $d\times r$ matrix $U$.
\FOR {$i = 1, 2, \dotsb, t$}
    \STATE $U\leftarrow (M M^{T})\cdot U$.
    \STATE Apply Gram–Schmidt process on $U$ such that columns of $E$ are orthogonal.
\ENDFOR
\STATE Return $U = (\bm{u}_1, \dotsb, \bm{u}_r)$ where $\bm{u}_i$ is $U$'s $i$-{th} column. 
\end{algorithmic}
\label{alg:basis}
\end{algorithm}

\noindent We formally describe the hybrid clipping with twice sampling in Algorithm \ref{alg: hybrid+twice-sampling}. In Algorithm \ref{alg: hybrid+twice-sampling}, we describe the set of orthogonal unit basis $\{\bm{u}_{jl}, l=1,2,\cdots, r_j\}$ as inputs.
This may be determined by either prior knowledge or public data. In Section \ref{sec: additional exp}, we use public data to approximate the basis in principal component space, inspired by the implementation in \cite{embedyu2021}. 
Given a public dataset $X^p = \big\{x^p_1, x^p_2, ... ,x^p_k\big\}$ where $k = \sum_{j = 1}^{m - 1}r_{j}$ and an approximation parameter $t$, we describe how to compute the basis in Algorithm \ref{alg:basis}.


\section{Proof of Theorem \ref{cor: optimized noise hybrid+twice}}
\label{app: cor: optimized noise hybrid+twice}
\FinalTheorem*
\begin{proof}
Without loss of generality, we assume that $U_j= \{\bm{u}_{j1}, \cdots,$ $\bm{u}_{j r_j}\}$ are natural bases. 
Given the clipping strategy, for any vector $\bm{v} \in \mathbb{R}^d$, we split it into $m$ segments, and each is clipped by a mixture of $l_{\infty}$ and $l_2$ norm. 
Thus, for any element $\bm{s} \in \mathsf{S}$, we use $\bm{s} = (\bm{s}_1, \cdots, \bm{s}_m)$ to represent its expression in $m$ segments, where $\bm{s}_j \in \mathbb{R}^{m_j}$ and $\|\bm{s}_j\|_2 \leq c_{2j}$, $\|\bm{s}_j\|_{\infty} \leq c_{\infty j} = c_{2j}/\sqrt{d_{0j}}$. We generate a Gaussian noise in distribution $\mathcal{N}(0, \sigma^2_j)$ for each coordinate in the $j$-th segment.

Now, first consider the $(\alpha_0,\epsilon(\alpha_0))$-RDP guarantee if we apply $q_2$-coordinate-wise sampling given such clipping and randomization, denoted by $\mathcal{F}^{CS}$. By (\ref{coordiante decomposition}) in Appendix \ref{app: thm: privacy coordiante sampling}, since the distribution of the output on each coordinate is still independent, we have that
\begin{equation}
\begin{aligned}
    &\mathcal{D}_{\alpha_0}(\mathbb{P}_{\mathcal{F}^{CS}(X')}\| \mathbb{P}_{\mathcal{F}^{CS}(X)}) = \sum_{l=1}^d \mathcal{D}_{\alpha}(\mathbb{P}_{\mathcal{F}^{CS}(X')(l)}\| \mathbb{P}_{\mathcal{F}^{CS}(X)(l)})\\
    & \leq \sum_{j=1}^m \sum_{l=1}^{r_j} \mathcal{D}_{\alpha_0}\big((1-q_2) \mathcal{N}(0, \sigma^2_j) + q_2\mathcal{N}(s_{jl}, \sigma^2) \|\mathcal{N}(0, \sigma^2_j) \big) \\
    &  = \sum_{j=1}^m \sum_{l=1}^{r_j}\frac{1}{\alpha_0-1}\log\big((1-q_2)^{\alpha_0}+\sum_{v=1}^{\alpha_0}\binom{\alpha}{v}(1-q_2)^{\alpha_0-v}q^{v}_2e^{\frac{v(v-1)s^2_{jl}}{2\sigma^2_{j}}} \big)\\
    & \leq \sum_{j=1}^m \frac{d_{0j}}{\alpha_0-1}\log\big((1-q_2)^{\alpha_0}+\sum_{v=1}^{\alpha_0}\binom{\alpha}{v}(1-q_2)^{\alpha_0-v}q^{v}_2e^{\frac{v(v-1)c^2_{\infty j}}{2\sigma^2_{j}}} \big).
\end{aligned}
\label{mix coordiante decomposition}
\end{equation}
The last line of (\ref{mix coordiante decomposition}) is simply from applying 
Theorem \ref{thm: privacy coordiante sampling} on the divergence sum in each segment. 
With an identical reasoning as that in Theorem \ref{thm: privacy coordiante sampling}, by the dominating divergence of one-dimensional subsampled Gaussian \cite{mironov2019r}, 
\begin{equation}
\begin{aligned}
    &\mathcal{D}_{\alpha_0}(\mathbb{P}_{\mathcal{F}^{CS}(X)}\| \mathbb{P}_{\mathcal{F}^{CS}(X')}) = \sum_{l=1}^d \mathcal{D}_{\alpha}(\mathbb{P}_{\mathcal{F}^{CS}(X)(l)}\| \mathbb{P}_{\mathcal{F}^{CS}(X')(l)})\\
    & \leq \sum_{j=1}^m \sum_{l=1}^{r_j} \mathcal{D}_{\alpha_0}\big(\mathcal{N}(0, \sigma^2_j)\|(1-q_2) \mathcal{N}(0, \sigma^2_j) + q_2\mathcal{N}(s_{l}, \sigma^2)  \big) \\
    &\leq \sum_{j=1}^m \sum_{l=1}^{r_j} \mathcal{D}_{\alpha_0}\big((1-q_2) \mathcal{N}(0, \sigma^2_j) + q_2\mathcal{N}(s_{l}, \sigma^2) \|\mathcal{N}(0, \sigma^2_j) \big),
\end{aligned}
\label{mix coordiante decomposition 1}
\end{equation}
Therefore, the bound in (\ref{mix coordiante decomposition}) is a global upper bound for both $\mathcal{D}_{\alpha_0}(\mathbb{P}_{\mathcal{F}^{CS}(X')}\| \mathbb{P}_{\mathcal{F}^{CS}(X)})$ and $\mathcal{D}_{\alpha_0}(\mathbb{P}_{\mathcal{F}^{CS}(X)}\| \mathbb{P}_{\mathcal{F}^{CS}(X')})$, and thus an upper bound of $\epsilon(\alpha_0)$.

We proceed to consider the twice sampling where the $q_2$ coordinate-wise sampling is implemented on a subsampled dataset by $q_1$-input-wise sampling. 
The rest privacy analysis is then the same as that in Appendix \ref{app: thm: sampling twice}, except that we add different amount of noise to each coordinate. 
This does not affect the conclusion, since finally it is still reduced to the form in (\ref{cross product}), a polynomial with positive coefficients on multiple one-dimensional Pearson-Vajda $\chi^t$-pseudo-divergences on Gaussian distributions with different variance $\sigma_j$. Thus, plugging (\ref{mix coordiante decomposition 1}) to (\ref{twice-sampling-bound}) in Theorem \ref{thm: sampling twice}, we obtain the RDP bound for the mixture of hybird clipping and twice sampling claimed.  
\end{proof}



\end{document}